%% file: main.tex
\documentclass[twoside,12pt]{article}
\usepackage[utf8]{inputenc}
\usepackage[T1]{fontenc}
\usepackage{enumitem}
\setitemize{itemsep=2pt,topsep=3pt,parsep=2pt}
\usepackage{textcomp}
\usepackage{relsize}
\usepackage{setspace}
\usepackage[fleqn]{amsmath}
\usepackage{amssymb}
\usepackage{amsthm}
\usepackage{mathtools}
\usepackage{array,multirow}
\usepackage{tabularx}
\usepackage{lmodern}
\usepackage{yfonts}
\usepackage[a4paper,left=2cm,right=2cm,top=2cm,bottom=2cm]{geometry}
\usepackage{graphicx}
\usepackage[table,usenames,dvipsnames]{xcolor}
\usepackage{microtype}
\usepackage{url}
\usepackage{tocloft}
\usepackage{titling}
\usepackage{caption}
\usepackage{float}
\usepackage{titlesec}
\usepackage{authblk}
\usepackage{xspace}
\usepackage[margin=2cm]{caption}
\newcommand{\I}[2]{I_{#1}\left(#2\right)}
\newcommand{\F}{\mathcal{F}}
\renewcommand{\P}{\mathcal{P}}
\renewcommand{\S}{\mathcal{S}}
\newcommand{\T}{\mathcal{T}}
\newcommand{\union}[1]{\langle #1 \rangle}
\newcommand{\N}{\mathcal{N}}
\newcommand{\X}{\mathcal{X}}
\newcommand{\Y}{\mathcal{Y}}
\newcommand{\D}{\mathcal{D}}
\newcommand{\ola}[1]{\protect\overleftarrow{#1}}
\newcommand{\ora}[1]{\protect\overrightarrow{#1}}
\newcommand{\dist}{\textup{dist}}
\newcommand{\Start}{\textup{start}}
\newcommand{\End}{\textup{end}}
\newcommand{\inn}{\textup{inn}}
\newcommand{\out}{\textup{out}}
\newcommand{\Ps}{P^{\overline{s}}}
\newcommand{\proj}[3]{\mathbf{P}_{#1}(#2,#3)}
\newcommand{\HH}[1]{[#1]}
\newcommand{\makerbreaker}{\textup{\textsc{MakerBreaker}}\xspace}
\newcommand{\PSPACE}{\textup{\textsf{PSPACE}}\xspace}
\newcommand{\QBF}{\textup{\textsc{}{QBF}}\xspace}
\newcommand{\UQBF}{\textup{\textsc{UnorderedQBF}}\xspace}
\newcommand{\Def}{Def.~}
\newcommand{\Not}{Not.~}
\newcommand{\Fig}{Fig.~}

\let\temp\forall \renewcommand{\forall}{\temp \,\,}
\let\tempo\exists \renewcommand{\exists}{\tempo \,\,}

\let\originalleft\left
\let\originalright\right
\renewcommand{\left}{\mathopen{}\mathclose\bgroup\originalleft}
\renewcommand{\right}{\aftergroup\egroup\originalright}

\newtheoremstyle{the}{15pt}{15pt}{\it}{}{\bfseries}{.}{ }{}
\newtheoremstyle{conj}{10pt}{10pt}{\it}{}{\bfseries}{.}{ }{}
\newtheoremstyle{cla}{15pt}{15pt}{\it}{}{$\quad$}{.}{ }{}
\newtheoremstyle{def}{15pt}{15pt}{}{}{\bfseries}{.}{ }{}
\newtheoremstyle{rem}{15pt}{15pt}{}{}{\it}{.}{ }{}

{\theoremstyle{the} \newtheorem{theoreme}{Theorem}[section]}
{\theoremstyle{the} \newtheorem*{theoreme*}{Theorem}}
{\theoremstyle{the} \newtheorem{lemme}[theoreme]{Lemma}}
{\theoremstyle{the} \newtheorem*{lemme*}{Lemma}}
{\theoremstyle{the} \newtheorem{proposition}[theoreme]{Proposition}}
{\theoremstyle{the} \newtheorem{corollaire}[theoreme]{Corollary}}
{\theoremstyle{the} }
{\theoremstyle{cla} \newtheorem{claim}{Claim}}
{\theoremstyle{def} \newtheorem{definition}[theoreme]{Definition}}
{\theoremstyle{def} \newtheorem{notation}[theoreme]{Notation}}
{\theoremstyle{rem} \newtheorem*{remarque}{Remark}}
{\theoremstyle{rem} }
{\theoremstyle{conj} \newtheorem{conjecture}[theoreme]{Conjecture}}

\newenvironment{proofclaim}[1][Proof]{\begin{proof}[#1]}{\end{proof}}

\definecolor{refcolor}{RGB}{150,0,0}
\definecolor{mygray1}{RGB}{225,225,225}

\usepackage[colorlinks=true, allcolors=refcolor]{hyperref}
\hypersetup{pdfstartview=XYZ}

\newcommand{\nom}[2]{#1 \textsc{#2}}

\title{Maker-Breaker is solved in polynomial time on hypergraphs of rank 3}
\author[1]{\nom{Florian}{Galliot}}
\author[2]{\nom{Sylvain}{Gravier}}
\author[3]{\nom{Isabelle}{Sivignon}}

\affil[1]{Aix-Marseille Université, CNRS, I2M, UMR 7373, 13453 Marseille, France}
\affil[2]{Univ. Grenoble Alpes, CNRS, Institut Fourier, 38000 Grenoble, France}
\affil[3]{Univ. Grenoble Alpes, CNRS, Grenoble INP, GIPSA-lab, 38000 Grenoble, France}

\date{}

\begin{document}

\maketitle

\input{abstract}

\input{body}

\bibliographystyle{biblio_style}
\bibliography{biblio}

\input{appendix}

\end{document}

%% file: abstract.tex
\begin{abstract}
	\noindent In the Maker-Breaker positional game, Maker and Breaker take turns picking vertices of a hypergraph $H$, and Maker wins if and only if she possesses all the vertices of some edge of $H$. Deciding the outcome (i.e., which player has a winning strategy) is a \PSPACE-complete problem even when restricted to 4-uniform hypergraphs \cite{Gal25}. 
	As for hypergraphs of rank 3, Kutz \cite{Kut05} has solved the linear subcase (i.e., any two distinct edges intersect on at most one vertex), by obtaining a structural characterization of the outcome and a polynomial-time algorithm to decide it. A conjecture by Rahman and Watson \cite{RW20} implies that the same results can be obtained for general hypergraphs of rank 3, which we confirm in this paper. We provide a structural characterization of the outcome and a description of both players' optimal strategies, all based on intersections of some key subhypergraph collections. From this, we derive a polynomial-time algorithm, thus closing the complexity gap for Maker-Breaker games relative to the size of the edges. Another corollary of our structural result is that, if Maker has a winning strategy on a hypergraph of rank 3, then she can ensure to win the game within a number of rounds that is logarithmic in the number of vertices.
\end{abstract}

%% file: body.tex
\section{Introduction}\label{Section1}

\hspace{\parindent}\textbf{Maker-Breaker games.} A \textit{positional game} is played on a hypergraph $H$, with two players who take turns picking vertices of $H$. There are several possible conventions of play, which determine the winner of the game. Let us mention two major ones:
\begin{itemize}[noitemsep,nolistsep]
	\item In the \textit{Maker-Maker} convention, the winner is the player who first possesses all vertices of some edge of $H$. If no player achieves this, then the game ends in a draw. The first general formulation of this convention goes back to Hales and Jewett \cite{HJ63}, while the first general results are due to Erd\H{o}s and Selfridge \cite{ES73}. The game of tic-tac-toe and its generalizations \cite{HJ63,Bec08} are the most famous examples of Maker-Maker positional games.
	\item In the \textit{Maker-Breaker} convention, one player ("Maker") wins if she possesses all vertices of some edge of $H$, while the other ("Breaker") wins if he can prevent this from happening. No draw is possible here. The first general formulation of this convention is due to Chv\'atal and Erd\H{o}s \cite{CE78}. The board game \textsc{Hex} \cite{Hei42,Nas52} and the Shannon switching game \cite{Gar61,Leh64,CE78} are the most famous examples of Maker-Breaker positional games.
\end{itemize}
\hspace{\parindent}This paper deals with the Maker-Breaker convention, which is the most studied in the literature as it possesses some nice properties (e.g., monotonicity properties). We always assume that Maker plays first: if Breaker plays first, then we can consider all possibilities of his first pick to reduce to the case where Maker plays first. Given a hypergraph $H$, there are two possibilities for the \textit{outcome} of the game: either Maker or Breaker has a winning strategy on $H$, and we say $H$ is a \textit{Maker win} or a \textit{Breaker win} accordingly. 
From an algorithmic point of view, the natural decision problem \makerbreaker takes an input hypergraph $H$ and returns "yes" if $H$ is a Maker win or "no" if $H$ is a Breaker win.
\\ \indent The Maker-Breaker game can be interpreted as a propositional logic problem. Indeed, it is directly linked with the \textsc{Quantified Boolean Formula} problem (\QBF), also known as \textsc{Quantified SAT} (\textsc{QSAT}), which differs from \textsc{SAT} in that each quantifier may be existential or universal. These types can be assumed to alternate, so that \QBF comes down to a game played on a formula in conjunctive normal form (CNF) where two players take turns choosing truth values for the variables: "Satisfier" (resp. "Falsifier") wants the formula to end up true (resp. false). Schaefer \cite{Sch78} introduced an unordered version of \QBF, which we will call \UQBF from here onwards, where a player's turn consists in choosing a variable and a truth value for it, instead of the variables arriving in a predetermined order. We always assume that Falsifier starts. In the particular case of a CNF formula which is positive, i.e., whose literals are all positive, Satisfier (resp. Falsifier) always puts the variables he (resp. she) picks to true (resp. false), so the game is equivalent to the Maker-Breaker game: Falsifier is Maker, Satisfier is Breaker, and clauses correspond to edges. Therefore, \makerbreaker is the restriction of \UQBF to positive CNF formulas.
\bigskip
\\\indent \textbf{Our problem.} Research on the Maker-Breaker game mainly consists in finding criteria for being a Maker win or a Breaker win, as well as evaluating the algorithmic complexity on various hypergraph classes of the \makerbreaker decision problem. Counting-type results, based on numerical formulas involving quantities such as the number of edges and their size for example, can provide conditions for a Maker win that are either necessary (like the Erd\H{o}s-Selfridge theorem \cite{ES73,Bec82}) or sufficient (like Beck's criterion involving the pair degree \cite{Bec82}) but usually not both. Instead, this paper focuses on structural results, corresponding to strategies where the player's picks are based on the existence and the interdependence of some key subhypergraphs. A structural characterization of the outcome on some hypergraph class, if simple enough, would imply that the restriction of \makerbreaker to that class is in \textsf{NP}, or even possibly in \textsf{P}. It should be noted that, as structural studies in hypergraphs can be notoriously difficult, not all hypergraph classes are suited to this approach.
\\ \indent One direction is to consider hypergraphs with small edges, as their structure is less complex. A hypergraph is of \textit{rank} $k$ if all its edges are of size at most $k$, and it is \textit{$k$-uniform} if all its edges are of size exactly $k$. For instance, the Maker-Breaker game on hypergraphs of rank 2 is easily solved: Maker wins if and only if there is an edge of size 1 or a vertex of degree at least 2. However, the game soon becomes difficult as the edges get bigger. The founding result by Schaefer \cite{Sch78} states that \makerbreaker is \PSPACE-complete on hypergraphs of rank 11. The \PSPACE-completeness result has since been extended successively to 6-uniform hypergraphs \cite{RW21}, to 5-uniform hypergraphs \cite{Koe25}, and then to 4-uniform hypergraphs in a recent preprint \cite{Gal25}. This means that there is no hope of an efficiently verifiable structural characterization of the outcome on hypergraphs of rank 4 or more. We thus investigate the missing case for the Maker-Breaker game, namely, hypergraphs of rank 3.
\bigskip
\\\indent \textbf{Previous work.} Kutz worked on the Maker-Breaker game in hypergraphs of rank 3 \cite{Kut04,Kut05} and solved the \textit{linear} subcase, meaning that any two distinct edges intersect on at most one vertex. He reduces to a subclass, on which he provides a precise structural characterization of Breaker wins. From this, he derives a polynomial-time algorithm for \makerbreaker on linear hypergraphs of rank 3. The central substructure at play here is called a \textit{chain}, which is a linear hypergraph formed by a sequence of edges of size 3 where two distinct edges intersect on 1 (resp. 0) vertex if they are (resp. are not) consecutive in the sequence. Things are different in general hypergraphs of rank 3, where intersections of size 2 somehow may hamper connections between vertices. Indeed, a major hurdle in the non-linear case is that the union of two chains, the first between $x$ and $y$ and the second between $y$ and $z$, does not necessarily contain a chain between $x$ and $z$. In particular, Kutz's structural result seems difficult to generalize.
\\ \indent More relevant to us is existing work regarding \UQBF on 3-CNF formulas (i.e., CNF formulas with clauses of size at most 3). This encompasses the Maker-Breaker game on hypergraphs of rank 3, which corresponds to positive 3-CNF formulas. Rahman and Watson conjecture that \UQBF is tractable on 3-CNF formulas \cite{RW20}, in striking contrast with \QBF which is a canonical \PSPACE-complete problem on 3-CNF formulas \cite{SM73}, and that Falsifier wins if and only if she can guarantee that a simple winning motif appears within the first few rounds of the game (we are talking about full rounds, i.e., the motif needs to be there before Falsifier's turn). The key motif is called a \textit{manriki}, which consists of two clauses of size 2 linked by a chain of any nonnegative number of clauses of size 3 (with the same definition of chain as above), where the literals can have any sign, e.g., $(x_1 \vee \neg x_2) \wedge (x_2 \vee x_3 \vee \neg x_4) \wedge (x_4 \vee x_5 \vee x_6) \wedge (\neg x_6 \vee \neg x_7 \vee x_8) \wedge (x_8 \vee x_9)$ is a manriki.

\begin{conjecture}[\textup{\cite{RW20}}]\label{conjecture}
    \UQBF is solvable in polynomial time on 3-CNF formulas. More specifically, we have the following structural characterization: there exists an integer $r$ such that Falsifier has a winning strategy for the \UQBF game on a 3-CNF formula $\phi$ if and only if she has a strategy ensuring that, after at most $r$ full rounds of play, one of the following "obstacles" appears in the updated formula obtained from $\phi$ by removing all clauses containing a true literal and removing all the false literals in the other clauses:
    \begin{enumerate}[label={(\arabic*)},nolistsep,noitemsep]
    		\item a clause of size 0 or 1;
    		\item a pair of clauses of size 2 on the same variables, where, if the total number of variables in $\phi$ is even, we exclude pairs of the form $\{(x_i \vee x_j),(\neg x_i \vee \neg x_j)\}$ and $\{(x_i \vee \neg x_j),(\neg x_i \vee x_j)\}$ (note that obstacle (2) cannot occur if $\phi$ is positive);
    		\item a manriki.
    \end{enumerate}
\end{conjecture}

It has since been shown that the existence of a chain between two given vertices in a hypergraph of rank 3 can be checked in polynomial time \cite{GGS22}. Therefore, checking for manrikis can also be done in polynomial time, so that the conjectured structural characterization indeed implies tractability through bruteforcing the players' moves during the first $r$ rounds.
\\ \indent Rahman and Watson have formulated Conjecture \ref{conjecture} after showing that the result holds with $r=3$ for 3-CNF formulas in which each clause contains a variable that occurs in no other clauses \cite{RW20}. This result in itself brings little to the study of the Maker-Breaker game on hypergraphs of rank 3, since restricting such formulas to the positive case yields hypergraphs which are essentially linear, a case which had already been solved by Kutz. However, Conjecture \ref{conjecture} suggests that the same approach can be used much more generally, which we confirm in this paper.
\bigskip
\\\indent \textbf{Our approach and results.} We center our approach on the notion of \textit{danger}. A danger is a subhypergraph $D$ representing a potential threat for Breaker, that Maker can activate in one move. If Maker plays a move that does activate the danger, then Breaker is forced to play his next move inside $D$, otherwise he will eventually lose facing optimal play. As this definition is very general, we only consider specific types of dangers, belonging to some fixed family $\D$ (we call them $\D$-dangers). We are particularly interested in \textit{forks}, i.e., collections of dangers that Maker can activate simultaneously and that do not intersect, so that Breaker's next move will necessarily miss one of them. Obviously, for any fixed family $\D$ of dangers, if there is a \textit{fork} of $\D$-dangers before the start of the game, then we have a Maker win. This is a general implication in hypergraphs of any rank. To solve the case of rank 3, we are looking for a family $\D$ of dangers with reasonably simple structure such that the converse holds.
\\ \indent In this paper, we construct such a family of dangers, thereafter named $\D_2$. This provides a structural characterization of the outcome, as well as a description of optimal strategies for both players, all based on danger intersections. From this, we derive a polynomial-time algorithm to solve \makerbreaker on hypergraphs of rank 3. Since the family $\D_2$ is actually designed to exactly represent the threat of some obstacle from Conjecture \ref{conjecture} appearing within three rounds of play, our result proves Conjecture \ref{conjecture} with $r=3$ for positive 3-CNF formulas. We also exhibit a substantial subclass on which Conjecture \ref{conjecture} holds with $r=2$. Finally, another parameter that is typically studied in positional games is the duration of the game when players try to win as fast as possible: in a hypergraph $H$ of rank 3, a corollary of our structural result is that Maker either can ensure to complete an edge in just $O(\log(|V(H)|))$ rounds or does not have a winning strategy at all.
\\\indent The outline of this paper is as follows. In Section \ref{Section2}, we give basic definitions and results around the Maker-Breaker game, before introducing the elementary hypergraphs of rank 3 that play a crucial role for the game. Section \ref{Section3} presents the notion of danger as well as the construction of the aforementioned family $\D_2$, and ends with the statements of the main results of this paper (structural characterizations, algorithmic complexity, duration of the game). Section \ref{Section4} aims at establishing a number of structural lemmas, as well as understanding the structural properties of the most complex $\D_2$-dangers. In Section \ref{Section5}, we prove all of our main results. Section \ref{Section6} concludes the paper and suggests some perspectives for future research. Additionally, Appendix \ref{appendix} lists all non-standard technical terms and mathematical symbols used in the paper, along with references to definitions and/or figures that present them.

\section{Definitions and basics}\label{Section2}

\subsection{Marked hypergraphs and the Maker-Breaker game}

\noindent For reasons that we will shortly explain, we consider the Maker-Breaker game on what we call marked hypergraphs, which are a generalization of hypergraphs. The marked vertices are those possessed by Maker.

\subsubsection{Marked hypergraphs}

\begin{definition}\label{def:markedhypergraph}
	A \textit{marked hypergraph} $H$ is defined by:
	\begin{itemize}[noitemsep,nolistsep]
		\item[--] a finite nonempty \textit{vertex set} $V(H)$;
		\item[--] an \textit{edge set} $E(H)$ consisting of nonempty subsets of $V(H)$;
		\item[--] a set of \textit{marked vertices} $M(H) \subseteq V(H)$.
	\end{itemize}
\end{definition}

\begin{remarque}
	A hypergraph may be seen as a marked hypergraph with no marked vertices, so that all definitions and notations associated with marked hypergraphs apply to hypergraphs as well. Also note that $E(H)$ is defined as a set, not a multiset, so that the edges are always pairwise distinct.
\end{remarque}

\begin{notation}
	A marked hypergraph consisting of a single edge $e$ may be simply denoted by $e$.
\end{notation}

\begin{definition}
	Let $H$ be a marked hypergraph.
	\begin{itemize}[noitemsep,nolistsep]
		\item Let $v_1,v_2 \in V(H)$ be distinct. We say $v_1$ and $v_2$ are \textit{adjacent in $H$} if there exists $e \in E(H)$ such that $\{v_1,v_2\} \subseteq e$.
		\item Let $v \in V(H)$. We say an edge $e \in E(H)$ is \textit{incident to $v$} if $v \in e$. The \textit{degree of $v$ in $H$} is defined as the number of edges of $H$ that are incident to $v$.
	\end{itemize}
\end{definition}

\begin{definition}
	Let $H$ be a marked hypergraph, and let $k \geq 1$ be an integer.
	\begin{itemize}[noitemsep,nolistsep]
		\item We say $H$ is of \textit{rank} $k$ if all its edges are of size at most $k$.
		\item We say $H$ is \textit{$k$-uniform} if all its edges are of size exactly $k$.
	\end{itemize}
\end{definition}

\begin{definition}\label{def:subhypergraph}
	Let $H$ be a marked hypergraph. A \textit{subhypergraph} of $H$ is a marked hypergraph $X$ such that: $V(X) \subseteq V(H)$, $E(X) \subseteq E(H)$ and $M(X)= V(X) \cap M(H)$. The notation $X \subseteq H$ means that $X$ is a subhypergraph of $H$.
\end{definition}

\begin{definition}\label{def:union}
	Let $\X=\{X_1,\ldots,X_t\}$ be a finite collection of subhypergraphs of some common marked hypergraph. The \textit{union} of $\X$, denoted by $\union{\X}$, is the marked hypergraph defined by: $V(\union{\X}) = \bigcup_{X \in \X}V(X)$, $E(\union{\X}) = \bigcup_{X \in \X}E(X)$ and $M(\union{\X}) = \bigcup_{X \in \X}M(X)$. We may also use the notation $\union{\X}=X_1 \cup \ldots \cup X_t$. 
\end{definition}

\begin{notation}\label{not:updated}
	Let $H$ be a marked hypergraph, and let $x,y \in V(H) \setminus M(H)$.
	\begin{itemize}[noitemsep,nolistsep]
		\item We denote by $H^{+x}$ the marked hypergraph obtained from $H$ by marking $x$, i.e.: $V(H^{+x})=V(H)$, $E(H^{+x})=E(H)$, $M(H^{+x})=M(H) \cup \{x\}$.
		\\ By convention, if $X \subseteq H$ does not contain $x$, then we define $X^{+x}=X$.
		\\ If $\X$ is a collection of subhypergraphs of $H$, then we define $\X^{+x}=\{X^{+x} \mid X \in \X\}$ which is a collection of subhypergraphs of $H^{+x}$.
		\item We denote by $H^{-y}$ the marked hypergraph obtained from $H$ by deleting $y$, assuming $V(H) \neq \{y\}$, i.e.: $V(H^{-y})=V(H)\setminus\{y\}$, $E(H^{-y})=\{e \in E(H) \mid y\not\in e\}$, $M(H^{-y})=M(H)$.
		\\ By convention, if $X \subseteq H$ does not contain $y$, then we define $X^{-y}=X$.
		\item We may combine these notations, as in $H^{+x-y}=(H^{+x})^{-y}=(H^{-y})^{+x}$ if $x \neq y$ for instance.
	\end{itemize}
\end{notation}

\begin{remarque}
	It should be noted that $H^{-y}$ is a subhypergraph of $H$, while $H^{+x}$ is not because of the additional marked vertex.
\end{remarque}

\subsubsection{The Maker-Breaker game on marked hypergraphs}

\noindent In the literature, the Maker-Breaker game is played on a standard hypergraph $H$ rather than a marked hypergraph. Maker and Breaker take turns picking vertices of $H$, and Maker wins if and only if she manages to possess all the vertices of some edge of $H$. Equivalently, Breaker wins if and only if he possesses all the vertices of some transversal (vertex cover) of the hypergraph. The actions of both players can be seen as follows: Maker \textit{marks} vertices, while Breaker \textit{deletes} vertices. Indeed, Breaker picking a vertex $y$ makes every edge $e \ni y$ irrelevant for the rest of the game, as Maker will never be able to possess all the vertices of $e$ from there.
\\ For this reason, it is natural to consider the game as played on marked hypergraphs. On each turn, Maker selects a non-marked vertex and marks it, then Breaker selects a non-marked vertex and deletes it (meaning the vertex is removed as well as all edges containing it). Some vertices may be marked already before the game starts. Maker wins if and only if, at some point during the game, there is a fully marked edge, i.e., an edge whose vertices are all marked. An example on the "tic-tac-toe hypergraph" is given in Figure \ref{Example_Game}: here, we see that Maker wins by completing the middle row of the hypergraph. For convenience, we will assume that the game continues until all vertices have been picked (i.e., the game continues even if Maker has already won for instance).

\begin{figure}[h]
	\centering
	\includegraphics[scale=.5]{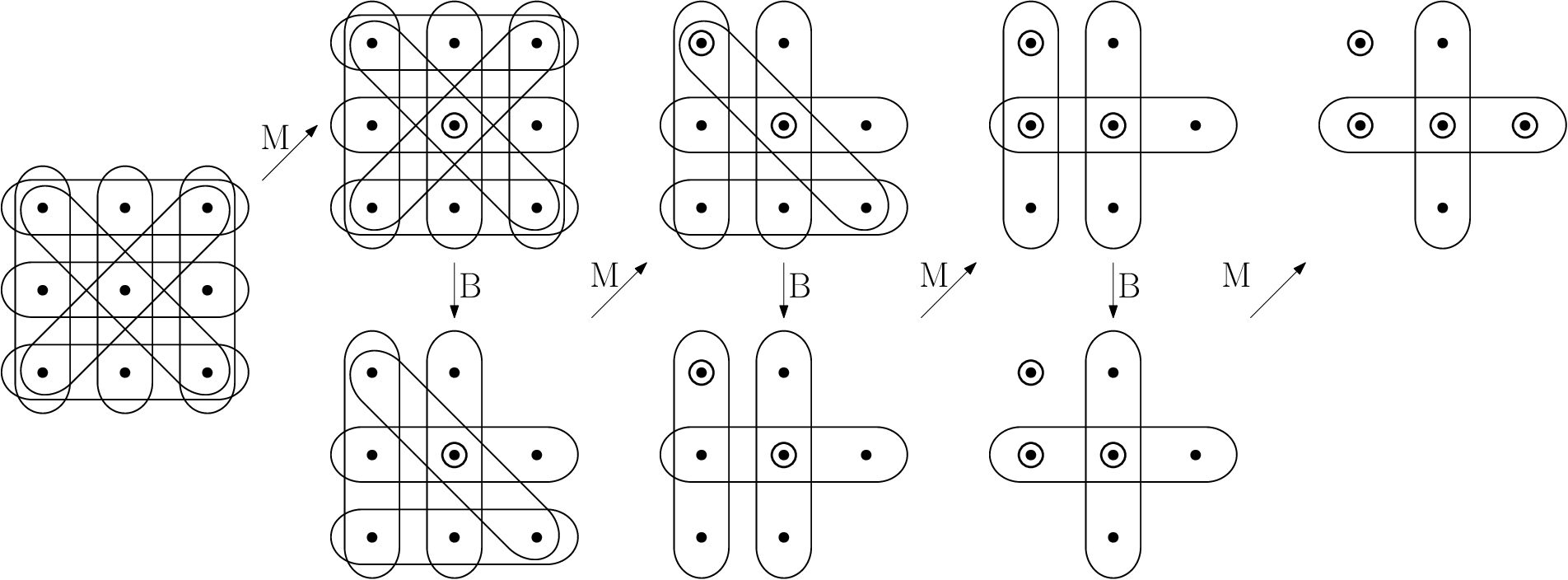}
	\caption{Evolution of the marked hypergraph during a game. The marked vertices are circled, as they will be in all figures.}\label{Example_Game}
\end{figure}

\noindent A \textit{strategy} dictates which vertex to pick next depending on all the previous moves. Formally, a strategy for Maker (resp. Breaker) on a marked hypergraph $H$ is a function which maps every sequence $S$ of pairwise distinct non-marked vertices of $H$, of even (resp. odd) length less than $|V(H) \setminus M(H)|$, to a non-marked vertex of $H$ that is not in $S$. A \textit{winning strategy} is a strategy $\Sigma$ for a player such that this player wins the game by following the strategy $\Sigma$, no matter what strategy the opponent follows. Since Maker and Breaker have complementary goals, the game cannot end in a draw, and so there are only two possibilities for the \textit{outcome} of the Maker-Breaker game played on a given marked hypergraph $H$: either Maker has a winning strategy, in which case we say that $H$ is a \textit{Maker win}, or Breaker has a winning strategy, in which case we say that $H$ is a \textit{Breaker win}. In practice, we will most often use a recursive definition of the outcome instead, which is equivalent. Indeed, the operators $^{+x}$ and $^{-y}$ can be interpreted as the effect of Maker picking $x$ and Breaker picking $y$ respectively. Therefore, after a \textit{round} of play (i.e., one move by each player) on a marked hypergraph $H$ where Maker marks $x$ and Breaker deletes $y$, it is as if a fresh game starts on the marked hypergraph $H^{+x-y}$. The outcome of the Maker-Breaker game can thus be defined in the following way:

\begin{definition}\label{def:trivialmakerwin}
	Let $H$ be a marked hypergraph. We say $H$ is a \textit{trivial Maker win} if some edge $e \in E(H)$ satisfies $|e \setminus M(H)| \leq 1$.
\end{definition}

\begin{definition}\label{def:makerwin}
	Let $H$ be a marked hypergraph. The fact that $H$ is a \textit{Maker win} is defined recursively as follows:
	\begin{enumerate}[noitemsep,nolistsep,label=(\arabic*)]
		\item If $|V(H) \setminus M(H)| \leq 1$, then $H$ is a Maker win if and only if $H$ is a trivial Maker win.
		\item If $|V(H) \setminus M(H)| \geq 2$, then $H$ is a Maker win if and only if there exists $x \in V(H) \setminus M(H)$ such that, for all $y \in V(H^{+x}) \setminus M(H^{+x})$, $H^{+x-y}$ is a Maker win.
	\end{enumerate}
	Otherwise, we say $H$ is a \textit{Breaker win}.
\end{definition}

\begin{notation}\label{not:makerbreaker}
	Let \makerbreaker be the decision problem that takes as input a marked hypergraph $H$ and outputs "yes" if and only if $H$ is a Maker win.
\end{notation}

\begin{remarque}
	It is very important to remember that Maker always plays first. Note that a marked hypergraph $H$ would be a Breaker win with Breaker playing first if and only if there exists $y \in V(H) \setminus M(H)$ such that $H^{-y}$ is a Breaker win with Maker playing first. Therefore, our assumption that Maker plays first is reasonable, since Breaker playing first would reduce to that case anyway. This explains why what we call a trivial Maker win should indeed fall under the definition of a Maker win:
	\begin{itemize}[noitemsep,nolistsep]
		\item If there exists some edge $e \in E(H)$ such that $e \setminus M(H)=\{x\}$, then Maker can win in one move, by picking $x$.
		\item If there exists some edge $e \in E(H)$ such that $e \setminus M(H)= \varnothing$, then Maker has already won.
	\end{itemize}
\end{remarque}

\noindent It can also be interesting to consider a version of the game where Maker tries to win in as few moves as possible. The following notation is introduced in \cite{HKS14}, and we adapt it to marked hypergraphs:
\begin{notation}\label{not:tau}
	Let $H$ be a marked hypergraph. We define $\tau_M(H)$ as the minimum number of rounds in which Maker can guarantee to get a fully marked edge when playing the Maker-Breaker game on $H$, with $\tau_M(H)=\infty$ by convention if $H$ is a Breaker win. Equivalently, $\tau_M(H)$ may be defined recursively as follows:
	\begin{enumerate}[noitemsep,nolistsep,label=(\arabic*),start=0]
		\item If $H$ is a trivial Maker win, then define $\tau_M(H) \in \{0,1\}$ as the minimum number of non-marked vertices in an edge of $H$.
		\item If $H$ is not a trivial Maker win and $|V(H) \setminus M(H)| \leq 1$, then define $\tau_M(H)=\infty$.
		\item If $H$ is not a trivial Maker win and $|V(H) \setminus M(H)| \geq 2$, then define $$\tau_M(H) = 1 + \underset{x \in V(H) \setminus M(H)}{\min}\,\,\,\underset{y \in V(H^{+x}) \setminus M(H^{+x})}{\max}\,\,\tau_M(H^{+x-y}).$$
	\end{enumerate}
\end{notation}

\noindent The study of the Maker-Breaker game revolves around considering certain classes of (marked) hypergraphs for which we try to:
\begin{itemize}[noitemsep,nolistsep]
	\item[--] identify criteria ensuring a Maker win or a Breaker win;
	\item[--] evaluate $\tau_M(\cdot)$ in the case of a Maker win, i.e., find fast-winning strategies for Maker;
	\item[--] determine the algorithmic complexity of \textsc{MakerBreaker}.
\end{itemize}
About that last problem, notice that we can always restrict ourselves to the uniform case:

\begin{proposition}\label{prop_reduction}
	For any $k \geq 2$, the following three decision problems all reduce polynomially to one another: 
	\begin{enumerate}[noitemsep,nolistsep,label={\textup{(\alph*)}}]
		\item \makerbreaker on hypergraphs of rank $k$;
		\item \makerbreaker on marked hypergraphs of rank $k$;
		\item \makerbreaker on $k$-uniform marked hypergraphs.
	\end{enumerate}
\end{proposition}

\begin{proof}
	The reduction from (a) to (b) is trivial since hypergraphs are special cases of marked hypergraphs. For the reduction from (b) to (c), let $H$ be a marked hypergraph of rank $k$ and define the $k$-uniform marked hypergraph $H_0$ obtained from $H$ as follows: for each edge $e$ of $H$, we create $k-|e|$ new marked vertices, and we add them to $e$. It is clear that $H$ is a Maker win if and only if $H_0$ is a Maker win. For the reduction from (c) to (a), we reverse this idea. Let $H$ be a $k$-uniform marked hypergraph, and let $H_0$ be the hypergraph of rank $k$ obtained from $H$ by removing all marked vertices and replacing each edge $e$ of $H$ by $e \setminus M(H)$. It is clear that $H$ is a Maker win if and only if $H_0$ is a Maker win.
\end{proof}

\noindent Compared to other ways of updating the hypergraph throughout the game that can be found in the literature, the use of marked hypergraphs allows us to manipulate uniform marked hypergraphs exclusively, since the operators $^{+x}$ and $^{-y}$ preserve uniformity. This will be important during our structural study of the 3-uniform case, which would become very tedious if we had to consider edges that could also be of size 2. Finally, this choice will help for proofs by induction (after one round of play, Maker's last pick is still present, albeit marked). \newline
Note that it is also possible to reduce to the uniform non-marked case, by successive duplications of edges into bigger edges. This reduction alters the hypergraph structure, however it will prove useful for algorithmic purposes.

\begin{proposition}\label{prop_duplication}
	Let $H$ be a marked hypergraph, and let $F \subseteq E(H)$. Define the marked hypergraph $H_0$ obtained from $H$ by adding two new (non-marked) vertices $a$ and $b$ and replacing each edge $e \in F$ by two edges $e_a = e \cup \{a\}$ and $e_b = e \cup \{b\}$. Then $H$ is a Maker win if and only if $H_0$ is a Maker win.
\end{proposition}

\begin{proof}
	Let $P \in \{\text{Maker},\text{Breaker}\}$ be the player who has a winning strategy $\Sigma$ on $H$. We design the following strategy for $P$ when playing on $H_0$. $P$ follows the strategy $\Sigma$, except when the opponent picks $a$ (resp. $b$) in which case $P$ answers by picking $b$ (resp. $a$). It can happen that $a$ and $b$ are the only non-marked vertices left, in which case $P$ picks one arbitrarily and the opponent has to pick the other. In all cases, Maker and Breaker share the vertices $a$ and $b$. Therefore, for each $e \in F$, the vertices of $e$ all end up marked if and only if the vertices of either $e_a$ or $e_b$ all end up marked. Since all the other edges are common to both $H$ and $H_0$, we can conclude that an edge of $H_0$ ends up with all its vertices marked if and only if some edge of $H$ does. All in all, this strategy is indeed winning for $P$ on $H_0$ since $\Sigma$ is on $H$.
\end{proof}

\subsubsection{Monotonicity properties of the Maker-Breaker game}

\noindent Let us exhibit two different ways in which the Maker-Breaker game is monotone. The first (resp. second) result states that marking vertices (resp. adding vertices and/or edges) cannot harm Maker.

\begin{proposition}[Marking Monotonicity]\label{prop_extra}
	Let $H$ be a marked hypergraph, and let $x \in V(H) \setminus M(H)$. If $H$ is a Maker win, then $H^{+x}$ is a Maker win. More precisely, we have $\tau_M(H^{+x}) \leq \tau_M(H)$.
\end{proposition}

\begin{proof}
	This is essentially the famous "strategy-stealing" argument \cite{HJ63}. Suppose Maker has a strategy $\Sigma$ to win within $t$ rounds on $H$. Playing on $H^{+x}$, in which $x$ is an "extra" (i.e., an extra marked vertex compared to $H$), Maker follows the strategy $\Sigma$. At some point, she will be instructed to pick the extra, which she cannot do as it is already marked in reality: instead, she picks an arbitrary non-marked vertex, which becomes the new extra. Continuing so, all throughout the game, the marked vertices are a superset of what they would be if the game was played on $H$: it is the same set plus the extra. Therefore, within the first $t$ rounds, some edge will be fully marked.
\end{proof}

\begin{proposition}[Subhypergraph Monotonicity]\label{prop_subwin}
	Let $H$ be a marked hypergraph, and let $X$ be a subhypergraph of $H$. If $X$ is a Maker win, then $H$ is a Maker win. More precisely, we have $\tau_M(H) \leq \tau_M(X)$.
\end{proposition}

\begin{proof}
	Suppose Maker has a strategy $\Sigma$ to win in at most $t$ rounds on $X$. Playing on $H$, Maker simply makes all her picks inside $X$, following the strategy $\Sigma$. Whenever Breaker picks a vertex outside $X$, Maker pretends that Breaker has picked some arbitrary non-marked $y \in V(X)$ instead, so as to get a normal game on $X$ (i.e., a game with alternating turns). Maker will get a fully marked edge of $X$ in at most $t$ rounds, as she cannot win any later than she would have if Breaker had actually made all his picks inside $X$.
\end{proof}

\noindent The previous result is well known and absolutely essential to this paper. Indeed, our approach is based on identifying some elementary Maker wins and determining whether Breaker can prevent them from appearing as a subhypergraph.

\subsection{Basic structures in 3-uniform marked hypergraphs}

\noindent During our study of the Maker-Breaker game on 3-uniform marked hypergraphs, some key substructures are going to arise. In this subsection, we define the most basic ones and study some of their properties.

\subsubsection{Walks}

\begin{definition}\label{def:walk}
    A \textit{walk} is a finite sequence $\ora{W}=(U_0,\ldots,U_{\ell})$ such that:
    \begin{itemize}[noitemsep,nolistsep]
        \item $U_0,\ldots,U_{\ell}$ are subsets of some common set of vertices $U$, of which some subset $M(U)$ is marked;
        \item $|U_i| \in \{1,3\}$  for all $0 \leq i \leq \ell$;
        \item $U_i \cap U_{i+1} \neq \varnothing$ for all $0 \leq i \leq \ell-1$.
    \end{itemize}
    A singleton $U_i=\{x\}$ might be simply denoted as $x$. Finally, we define $V(\ora{W}) = \bigcup_{0 \leq i \leq \ell}U_i$, $E(\ora{W}) = \{U_i \mid 0 \leq i \leq \ell \text{ and } |U_i| = 3\}$ and $M(\ora{W})=V(\ora{W}) \cap M(U)$.
\end{definition}

\noindent We are going to use walks as a way to navigate inside 3-uniform marked hypergraphs. They also help defining some elementary structures whose edge sets have a natural ordering. The elements of the walk will correspond to edges, plus some singletons which are useful to give information about intersections: for example, if a subsequence $(\ldots,e,x,e',\ldots)$ appears inside of a walk where $e$ and $e'$ are edges and $x$ is a vertex, then we know that $x \in e \cap e'$. Note that elements of a walk are not necessarily pairwise distinct.

\begin{definition}\label{def:equivalent}
	Two walks are said to be \textit{equivalent} if they coincide when removing all their singleton elements.
\end{definition}

\begin{notation}\label{notation_walk}
	Let $\ora{W}=(U_0,\ldots,U_{\ell})$ be a walk.
	\begin{itemize}[noitemsep,nolistsep]
		\item Provided $U_0,\ldots,U_{\ell}$ are not all singletons, we denote by $\Start(\ora{W})$ (resp. $\End(\ora{W})$) the non-singleton element $U_i$ of smallest (resp. largest) index $i$.
		\item We define the reverse walk $\ola{W}=(U_{\ell},\ldots,U_0)$.
		\item If $\ora{W'}=(U'_0,\ldots,U'_{\ell'})$ is another walk such that $U_{\ell} \cap U'_0 \neq \varnothing$, then we define the concatenated walk $\ora{W} \oplus \ora{W'} = (U_0,\ldots,U_{\ell},U'_0,\ldots,U'_{\ell'})$.
		\item Given a set $Z$ such that $Z \cap V(\ora{W}) \neq \varnothing$, we define the walk ``cut at $Z$'' as $\ora{W}\vert_Z=(U_0,\ldots,U_j)$ where $j=\min\{0 \leq i \leq \ell \mid Z \cap U_i \neq \varnothing\}$ (so that the walk $\ora{W}$ is stopped the first time that it reaches the set $Z$).
	\end{itemize}
\end{notation}

\begin{definition}\label{def:induced}
	The marked hypergraph \textit{induced} by a walk $\ora{W}$ is the marked hypergraph, denoted by $\HH{\ora{W}}$, defined by $V(\HH{\ora{W}})=V(\ora{W})$, $E(\HH{\ora{W}})=E(\ora{W})$ and $M(\HH{\ora{W}})=M(\ora{W})$.
\end{definition}

\noindent Some elementary 3-uniform marked hypergraphs can be defined as induced by a walk. In this paper, almost all of them will be linear.

\begin{definition}\label{def:linearhypergraph}
	We say a marked hypergraph $H$ is \textit{linear} if $|e \cap e'| \leq 1$ for all distinct $e,e' \in E(H)$.
\end{definition}

\noindent We now define linearity for walks, in a way that is consistent with the definition for hypergraphs, as well as the notion of simple walk.

\begin{definition}\label{def:linearwalk}
	We say a walk $\ora{W}$ is \textit{linear} if its induced marked hypergraph $\HH{\ora{W}}$ is linear.
\end{definition}

\begin{definition}\label{def:simplewalk}
	Let $\ora{W}=(U_0,\ldots,U_{\ell})$ be a walk.
	\begin{itemize}[noitemsep,nolistsep]
		\item Let $x \in V(\ora{W})$. We say $x$ is a \textit{repeated vertex} in $\ora{W}$ if there exist indices $i,j$ such that $|i-j|\geq 2$ and $x \in U_i \cap U_j$.
		\item We say $\ora{W}$ is \textit{simple} if there are no repeated vertices in $\ora{W}$, i.e., if $U_i \cap U_j = \varnothing$ for all $i,j$ such that $|i-j| \geq 2$.
	\end{itemize}
\end{definition}

\subsubsection{Chains, cycles and tadpoles}

\noindent The following object is often referred to as a \textit{linear path} or a \textit{loose path} in the literature \cite{OS14}. In this paper, we call it a chain, coinciding with the definition from \cite{RW20}.

\begin{definition}\label{def_chain}
	Let $P$ be a marked hypergraph, and let $a,b \in V(P)$. We say $P$ is an \textit{$ab$-chain} if there exists a walk inducing $P$ of the form $\ora{W}=(a,e_1,\ldots,e_L,b)$ where:
	\begin{itemize}[noitemsep,nolistsep]
		\item $e_1,\ldots,e_L$ are each of size 3;
		\item $\ora{W}$ is linear;
		\item $\ora{W}$ is simple.
	\end{itemize}
	Any walk $\ora{W}$ that satisfies this definition or is equivalent to one that does is then said to \textit{represent} $P$. We say $L=|E(P)|$ is the \textit{length} of $P$. An $ab$-chain may also be referred to as an \textit{$a$-chain} if we desire to highlight just one end point, or a \textit{chain} if we desire to highlight none. See Figure \ref{Example_Paths}.
\end{definition}

\begin{remarque}
	Any $ab$-chain is also a $ba$-chain (take the reverse walk in the definition), an $a$-chain, a $b$-chain and a chain.
\end{remarque}

\begin{figure}[h]
	\centering
	\includegraphics[scale=.58]{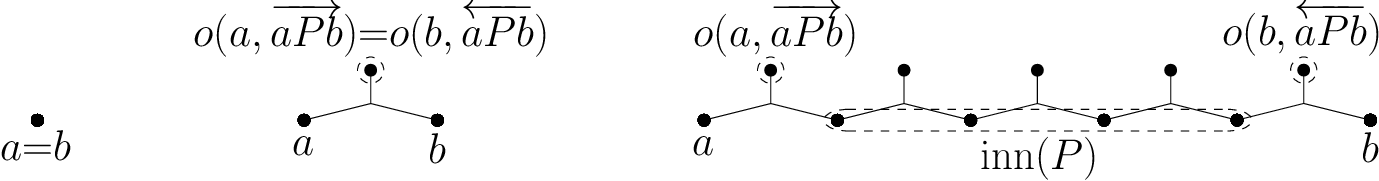}
	\caption{An $ab$-chain $P$ of length 0 (left), length 1 (middle), length 5 (right). In all figures of this paper, edges of size 3 will be represented using a "claw" shape joining their three vertices.}\label{Example_Paths}
\end{figure}

\begin{notation}\label{not:walk_chain}
	Let $P$ be an $ab$-chain. For fixed $a$ and $b$, there is a unique walk in $P$ satisfying the definition: we denote it by $\ora{aPb}=(a,e_1,\ldots,e_L,b)$. Similarly, for fixed $a$ (resp. fixed $b$), the walk $\ora{aP}=(a,e_1,\ldots,e_L)$ (resp. $\ora{bP}=(b,e_L,\ldots,e_1)$) is well defined. Note that the walks $\ora{aPb}$, $\ora{bPa}=\ola{aPb}$, $\ora{aP}$, $\ora{bP}$ all represent $P$.
\end{notation}

\begin{definition}\label{def:chain_inner}
	Let $P$ be an $ab$-chain. We define $\inn(P)=\bigcup_{e,e' \in E(P), e \neq e'}(e \cap e')$, which corresponds to the set of vertices of degree 2 in $P$. An element of $\inn(P)$ is called an \textit{inner vertex} of $P$. See Figure \ref{Example_Paths}.
\end{definition}

\begin{notation}\label{not:o_neighbor}
	Let $P$ be an $ab$-chain of positive length. We denote by $o(a,\ora{aPb})$ the only vertex in $\Start(\ora{aPb}) \setminus (\inn(P) \cup \{a,b\})$. See Figure \ref{Example_Paths}.
\end{notation}

\noindent We now introduce cycles. Our definition coincides with that from \cite{Kut04} (for cycles of length at least 3) and \cite{RW20}.

\begin{definition}\label{def_cycle}
	Let $C$ be a marked hypergraph, and let $a \in V(C)$. We say $C$ is an \textit{$a$-cycle} if there exists a walk inducing $C$ of the form $\ora{W}=(a,e_1,\ldots,e_L,a)$ where:
	\begin{itemize}[noitemsep,nolistsep]
		\item $e_1,\ldots,e_L$ are each of size 3;
		\item $L \geq 2$;
		\item if $L \geq 3$, then $\ora{W}$ is linear, and if $L=2$, then $|e_1 \cap e_2|=2$;
		\item $a$ is the only repeated vertex in $\ora{W}$, and $\{1 \leq i \leq L \mid a \in e_i\}=\{1,L\}$.
	\end{itemize}
	Any walk $\ora{W}$ that satisfies this definition or is equivalent to one that does is then said to \textit{represent} $C$. We say $L=|E(C)|$ is the \textit{length} of $C$. An $a$-cycle may simply be referred to as a \textit{cycle}. See Figure \ref{Example_Cycles}.
\end{definition}

\begin{figure}[h]
	\centering
	\includegraphics[scale=.58]{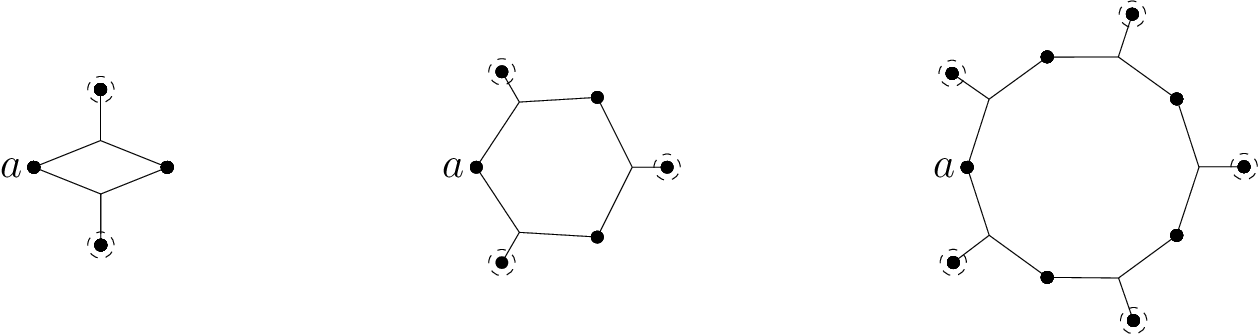}
	\caption{An $a$-cycle $C$ of length 2 (left), length 3 (middle), length 5 (right). The outer vertices are highlighted, the others are inner vertices.}\label{Example_Cycles}
\end{figure}

\begin{remarque}
	Note that a cycle is linear except if it is of length 2.
\end{remarque}

\begin{notation}\label{not:walk_cycle}
	Let $C$ be an $a$-cycle. For fixed $a$, there are exactly two walks satisfying Definition \ref{def_cycle}: if the first one is written as $(a,e_1,\ldots,e_L,a)$, then the second one is $(a,e_L,\ldots,e_1,a)$. We denote the former by $\ora{(a-e_1)C}$ and the latter by $\ora{(a-e_L)C}$. When wishing to consider one of the two arbitrarily, we may use the notation $\ora{aC}$.
\end{notation}

\begin{definition}\label{def:cycle_inner}
	Let $C$ be an $a$-cycle.
	\begin{itemize}[noitemsep,nolistsep]
		\item We define $\inn(C)=\bigcup_{e,e' \in E(C), e \neq e'}(e \cap e')$, which corresponds to the set of vertices of degree 2 in $C$. An element of $\inn(C)$ is called an \textit{inner vertex} of $C$.
		\item We define $\out(C)=V(C) \setminus \inn(C)$, which corresponds to the set of vertices of degree 1 in $C$. An element of $\out(C)$ is called an \textit{outer vertex} of $C$.
	\end{itemize}
	See Figure \ref{Example_Cycles}.
\end{definition}

\begin{remarque}
	An $a$-cycle $C$ is also a $b$-cycle for any $b \in \inn(C)$ (note that $a \in \inn(C)$ for instance), however it is not a $b$-cycle if $b \in \out(C)$.
\end{remarque}

\begin{definition}\label{def:hyperforest}
    A 3-uniform marked hypergraph is called a \textit{hyperforest} if it contains no cycles.
\end{definition}

\noindent Finally, we introduce tadpoles, a less standard hypergraph structure which will also play a prominent part in our structural studies. This terminology is inspired from graph theory, in which a \textit{tadpole graph} is defined as the union of a path and a cycle whose only shared vertex is one of the extremities of the path.

\begin{definition}\label{def_tadpole}
	Let $T$ be a marked hypergraph, and let $a \in V(T)$. We say $T$ is an \textit{$a$-tadpole} if there exists a walk inducing $T$ of the form $\ora{W}=(a,e_1,\ldots,e_s,b,e_{s+1},\ldots,e_t,b)$ where:
	\begin{itemize}[noitemsep,nolistsep]
		\item $a$ and $b$ are the only singletons;
		\item $e_1,\ldots,e_t$ are each of size 3;
		\item $(a,e_1,\ldots,e_s,b)$ represents an $ab$-chain $P_T$;
		\item $(b,e_{s+1},\ldots,e_t,b)$ represents a $b$-cycle $C_T$;
		\item $V(P_T) \cap V(C_T)=\{b\}$.
	\end{itemize}
	Any walk $\ora{W}$ that satisfies this definition or is equivalent to one that does is then said to \textit{represent} $T$. We may simply say $T$ is a \textit{tadpole}. The $ab$-chain $P_T$ and the $b$-cycle $C_T$ are clearly unique for a given $T$ (they do not depend on the choice of $\ora{W}$), so we may keep these notations. It is important to note that an $a$-cycle is a particular case of an $a$-tadpole, where $s=0$, i.e., $a=b$. See Figure \ref{Example_Tadpoles}.
\end{definition}

\begin{remarque}
	In other words, an $a$-tadpole is the union, for some vertex $b$, of an $ab$-chain and a $b$-cycle whose only common vertex is $b$. We emphasize that, by definition of a $b$-cycle, $b$ must be an inner vertex of the cycle: in particular, the examples from Figure \ref{Not_Tadpoles} are not $a$-tadpoles. Also note that a tadpole $T$ is linear except if $C_T$ is of length 2.
\end{remarque}

\begin{figure}[h]
	\centering
	\includegraphics[scale=.58]{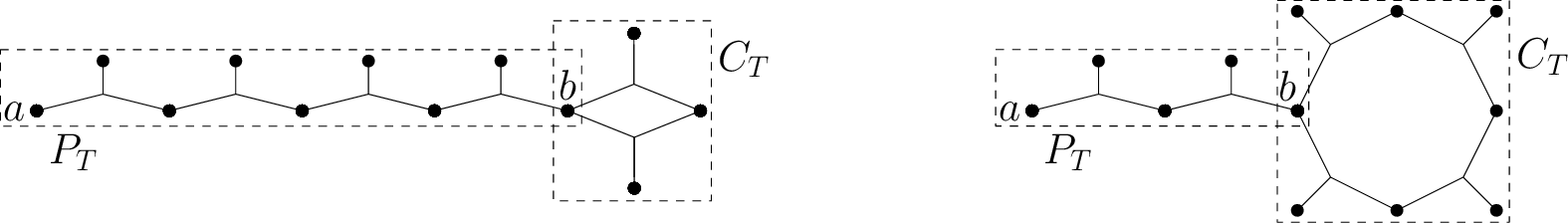}
	\caption{An $a$-tadpole $T$ (that is not an $a$-cycle), two examples.}\label{Example_Tadpoles}
\end{figure}

\begin{figure}[h]
	\centering
	\includegraphics[scale=.58]{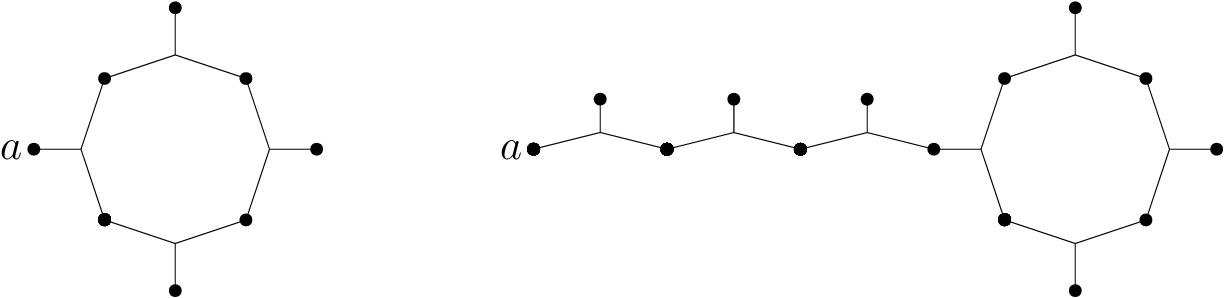}
	\caption{These two hypergraphs are not $a$-tadpoles.}\label{Not_Tadpoles}
\end{figure}

\begin{notation}\label{not:walk_tadpole}
	Let $T$ be an $a$-tadpole. For fixed $a$, there are exactly two walks satisfying Definition \ref{def_tadpole}: if the first one is written as $(a,e_1,\ldots,e_s,b,e_{s+1},\ldots,e_t,b)$, then the second one is $(a,e_1,\ldots,e_s,b,e_t,e_{t-1},\ldots,e_{s+1},b)$. The notation $\ora{aT}$ refers to either of the two arbitrarily.
\end{notation}

\subsubsection{Specific marked structures}

\noindent Chains, cycles and tadpoles are purely defined by their hypergraph structure, and may have any number of marked vertices. We now introduce special cases where some particular vertices are marked (see Figure \ref{Example_Marked} for some examples), which will be very relevant to our study of the Maker-Breaker game on 3-uniform marked hypergraphs.

\begin{definition}\label{def:snake}
	An \textit{$a$-snake} is an $ab$-chain $S$ of positive length such that $b \in M(S)$. We may also refer to $S$ as an \textit{$ab$-snake} or simply a \textit{snake}.
\end{definition}

\begin{definition}\label{def:nunchaku}
	An \textit{$ab$-nunchaku} is an $ab$-chain $N$ of positive length such that $M(N)=\{a,b\}$. We may also refer to $N$ as an \textit{$a$-nunchaku} or simply a \textit{nunchaku}.
\end{definition}

\begin{definition}\label{def:necklace}
	An \textit{$a$-necklace} is an $a$-cycle $C$ such that $M(C)=\{a\}$. An $a$-necklace may simply be referred to as a \textit{necklace}.
\end{definition}

\begin{remarque}
	Note that nunchakus and necklaces have an exact required number (and location) of marked vertices, whereas a snake might have more than the one prescribed marked vertex. For example, a nunchaku is technically a snake.
\end{remarque}

\begin{figure}[h]
	\centering
	\includegraphics[scale=.58]{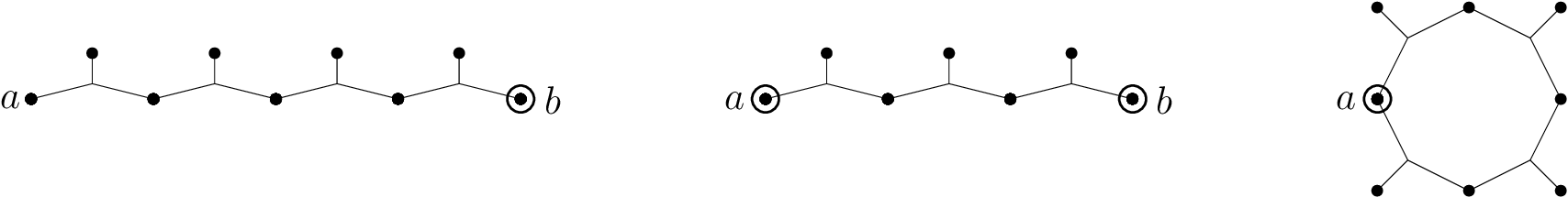}
	\caption{From left to right: an $ab$-snake, an $ab$-nunchaku, an $a$-necklace.}\label{Example_Marked}
\end{figure}

\section{Dangers}\label{Section3}

\subsection{Dangers and forks}

\noindent This subsection presents notions that are valid in all marked hypergraphs, regardless of rank. In this paper, we approach the Maker-Breaker game from Breaker's point of view. Informally, a danger at $x$ is a local threat that will be activated if Maker picks $x$, in the sense that Breaker would then be forced to neutralize the danger with his very next move. Therefore, a danger is naturally defined as a subhypergraph:

\begin{definition}\label{def:dangeratx}
	Let $H$ be a marked hypergraph and $x \in V(H) \setminus M(H)$. A \textit{danger at $x$ in $H$} is a subhypergraph $D$ of $H$ containing $x$ such that $D^{+x}$ is a Maker win.
\end{definition}

\noindent Suppose that $D$ is a danger at $x$ in $H$ and Maker picks $x$. Now, the resulting marked hypergraph $H^{+x}$ contains $D^{+x}$ as a subhypergraph, which is a Maker win by definition. However, the term "Maker win" assumes that Maker plays first, whereas in our scenario it is now Breaker's turn to pick some vertex $y$. If Breaker leaves $D^{+x}$ intact, then he will lose by Proposition \ref{prop_subwin} (Subhypergraph Monotonicity). Therefore, Breaker is forced to immediately "destroy" the danger, i.e., pick some $y \in V(D) \setminus \{x\}$. Note that it is unclear in general which choices of $y \in V(D) \setminus \{x\}$ are winning or losing for Breaker, but all choices of $y \not\in V(D) \setminus \{x\}$ are for sure losing, so it is indeed necessary for Breaker to destroy the danger. \newline
It is also useful to define a notion of danger as a "type" of object, independently of any ambient marked hypergraph, hence the following definitions.

\begin{definition}\label{def:pointed}
	A \textit{pointed marked hypergraph} is a pair $(H,x)$ where $H$ is a marked hypergraph and $x \in V(H) \setminus M(H)$.
\end{definition}

\begin{definition}\label{def:isomorphic}
	We say two pointed marked hypergraphs $(H,x)$ and $(H',x')$ are \textit{isomorphic} if there exists a bijection $\varphi : V(H) \to V(H')$ such that:
	\begin{itemize}[noitemsep,nolistsep]
		\item For all $e \subseteq V(H)$: $\, e \in E(H) \iff \varphi(e) \in E(H')$.
		\item For all $v \in V(H)$: $\, v \in M(H) \iff \varphi(v) \in M(H')$.		
		\item $\varphi(x)=x'$.
	\end{itemize}
\end{definition}

\begin{definition}\label{def:danger}
	A \textit{danger} is a pointed marked hypergraph $(D,x)$ such that $D^{+x}$ is a Maker win.
\end{definition}

\noindent The most simple example of a danger is the following:

\begin{definition}\label{def:trivialdanger}
	Let $k \geq 2$ be an integer. The \textit{trivial danger of size $k$} is the (unique, up to isomorphism) danger $(D,x)$ consisting of a single edge of size $k$ in which all vertices are marked except for $x$ and a single other vertex. See Figure \ref{Trivial_Danger}.
\end{definition}

\begin{figure}[h]
	\centering
	\includegraphics[scale=.5]{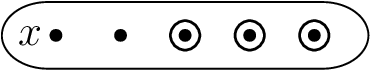}
	\caption{A trivial danger $(D,x)$ of size 5. It is indeed a danger because $D^{+x}$ is a trivial Maker win.}\label{Trivial_Danger}
\end{figure}

\noindent What if Maker picks a vertex $x$ at which there is, not just one danger, but a whole collection $\X$ of dangers? Then Breaker would be forced to answer by picking some vertex $y$ that destroys all of these dangers at once. We will thus consider intersections of collections of subhypergraphs. As vertices that are already marked cannot be picked, we want to exclude them from the intersection, hence the following definition.

\begin{definition}\label{def:intersection}
	Let $\X$ be a collection of marked hypergraphs and let $H$ be a marked hypergraph. We define the \textit{intersection of $\X$ in $H$} as:
	$$ \I{H}{\X} = \{y \in V(H) \setminus M(H) \mid \,y \in V(X) \text{ for all } X \in \X\}.$$
\end{definition}

\begin{remarque}
	Note that we have the following properties:
	\begin{itemize}[noitemsep,nolistsep]
		\item If $\X=\varnothing$, then $\I{H}{\X} = V(H) \setminus M(H)$.
		\item If $\X \subseteq \Y$, then $\I{H}{\Y} \subseteq \I{H}{\X}$.
	\end{itemize}
\end{remarque}

\noindent Since Maker plays first, we always study dangers before Maker's turn. This means that, when considering dangers at some $x$, that vertex $x$ has not been picked yet, however it makes sense to exclude $x$ from the intersection since Breaker will not be able to pick $x$ if Maker picks $x$ herself. The above definition is convenient in that regard, as we can exclude a non-marked vertex $x$ by simply taking the intersection in $H^{+x}$ instead of $H$. This is what we do in the following definition for instance, which describes a situation where $x$ is a winning move for Maker because Breaker cannot destroy all dangers at once.

\begin{definition}\label{def:fork}
	Let $H$ be a marked hypergraph and $x \in V(H) \setminus M(H)$. A \textit{fork at $x$ in $H$} is a collection $\F$ of dangers at $x$ in $H$ such that $\I{H^{+x}}{\F}=\varnothing$, i.e., $\I{H}{\F}=\{x\}$.
\end{definition}

\noindent When considering danger intersections, we will consider the same "types" of dangers at all vertices, given by some family of dangers $\D$.

\begin{notation}\label{notation_xF}
	Let $\D$ be a family of dangers. Let $H$ be a marked hypergraph and $x \in V(H) \setminus M(H)$. We denote by $x\D(H)$ the collection of all subhypergraphs $D$ of $H$ such that $x \in V(D)$ and $(D,x)$ is isomorphic to an element of $\D$ (in particular, all elements of $x\D(H)$ are dangers at $x$ in $H$).
\end{notation}

\begin{definition}\label{def:D-danger}
	Let $\D$ be a family of dangers. Let $H$ be a marked hypergraph and $x \in V(H) \setminus M(H)$.
	\begin{itemize}[noitemsep,nolistsep]
		\item A danger $(D,x) \in \D$ will be called a \textit{$\D$-danger}.
		\item An element of $x\D(H)$ will be called a \textit{$\D$-danger at $x$ in $H$}.
		\item A fork $\F$ at $x$ in $H$ such that $\F \subseteq x\D(H)$ will be called a \textit{$\D$-fork at $x$ in $H$}.
	\end{itemize}
\end{definition}

\noindent We now introduce a key property which states whether Breaker can destroy all $\D$-dangers at whichever vertex Maker picks at the start of the game. In other words, this property means that there are no $\D$-forks anywhere. This is a necessary condition for a Breaker win, whatever the considered family of dangers $\D$.

\begin{notation}\label{not:J}
	Let $\D$ be a family of dangers. Let $H$ be a marked hypergraph with $|V(H)\setminus M(H)| \geq 2$. We say that the property $J(\D,H)$ holds if:
	$$ \forall x \in V(H) \setminus M(H)\, : \,\,\I{H^{+x}}{x\D(H)}\neq\varnothing.$$
\end{notation}

\begin{remarque}
	Dangers are not relevant when there is less than one full round of play left, hence the assumption that $|V(H)\setminus M(H)| \geq 2$. This also avoids some dull cases where the property would fail on a technicality, by ensuring that if $x\D(H)=\varnothing$, then $\I{H^{+x}}{x\D(H)}=V(H^{+x})\setminus M(H^{+x}) \neq \varnothing$.
\end{remarque}

\begin{proposition}\label{prop_cn}
	Let $\D$ be a family of dangers. Let $H$ be a marked hypergraph with $|V(H)\setminus M(H)| \geq 2$. If $H$ is a Breaker win, then $J(\D,H)$ holds.  More precisely, if $J(\D,H)$ does not hold, then any $x \in V(H) \setminus M(H)$ such that $\I{H^{+x}}{x\D(H)} = \varnothing$ is a winning first pick for Maker.
\end{proposition}

\begin{proof}
	Suppose that Maker picks $x \in V(H) \setminus M(H)$ such that $\I{H^{+x}}{x\D(H)} = \varnothing$. We must show that, no matter what vertex $y \in V(H^{+x}) \setminus M(H^{+x})$ Breaker picks as an answer, $H^{+x-y}$ is a Maker win. Since $y \not\in \I{H^{+x}}{x\D(H)}$, there exists $D \in x\D(H)$ such that $y \not\in V(D)$. Therefore, $D^{+x}$ is a subhypergraph of $H^{+x-y}$, and it is a Maker win by definition of a danger at $x$. By Proposition \ref{prop_subwin} (Subhypergraph Monotonicity), $H^{+x-y}$ is a Maker win, which concludes.
\end{proof}

\noindent Observe that the property $J(\,\cdot\,,\,\cdot\,)$ is monotone in both its arguments:

\begin{proposition}
	Let $\D$ be a family of dangers. Let $H$ be a marked hypergraph with $|V(H)\setminus M(H)| \geq 2$.
	\begin{itemize}[noitemsep,nolistsep]
		\item For any family of dangers $\D' \subseteq \D$: $\,J(\D,H) \implies J(\D',H)$.
		\item For any subhypergraph $X \subseteq H$ such that $|V(X)\setminus M(X)| \geq 2$: $\,J(\D,H) \implies J(\D,X)$.
	\end{itemize}
\end{proposition}

\begin{proof}
	The first property comes from the fact that $x\D'(H) \subseteq x\D(H)$ hence $\I{H^{+x}}{x\D(H)} \subseteq \I{H^{+x}}{x\D'(H)}$. Let us now prove the second property. Suppose $J(\D,H)$ holds. Let $x \in V(X)\setminus M(X)$: we want to show that there exists $y \in \I{X^{+x}}{x\D(X)}$. By $J(\D,H)$, there exists $y' \in \I{H^{+x}}{x\D(H)}$. If $y' \in V(X)$, then $y = y'$ is suitable since $x\D(X) \subseteq x\D(H)$. If $y' \not\in V(X)$, then in particular $x\D(X)=\varnothing$ (indeed, if there existed $D_0 \in x\D(X) \subseteq x\D(H)$, then we would have $y' \in V(D_0) \subseteq V(X)$), therefore any $y \in V(X^{+x})\setminus M(X^{+x})$ is suitable.
\end{proof}

\noindent When considering all possible dangers at each non-marked vertex, the converse of Proposition \ref{prop_cn} actually holds:

\begin{theoreme}\label{theo_cns}
	Let $\D_{\textup{all}}$ be the family of all dangers. Let $H$ be a marked hypergraph with $|V(H)\setminus M(H)| \geq 2$. Then $H$ is a Breaker win if and only if $J(\D_{\textup{all}},H)$ holds.
\end{theoreme}

\begin{proof}
	The "only if" direction is given by Proposition \ref{prop_cn}, so we show the "if" direction. Suppose $J(\D_{\textup{all}},H)$ holds. Maker picks some $x \in V(H) \setminus M(H)$, and we must show that there exists an answer $y \in V(H^{+x}) \setminus M(H^{+x})$ by Breaker such that $H^{+x-y}$ is a Breaker win. Breaker may pick any $y \in \I{H^{+x}}{x\D_{\textup{all}}(H)}$. Since $y \not\in V(H^{-y})$, we have $H^{-y} \not\in x\D_{\textup{all}}(H)$, i.e., $H^{-y}$ is not a danger at $x$ in $H$. By definition of a danger at $x$, this means $(H^{-y})^{+x}$ is a Breaker win, which concludes since $(H^{-y})^{+x}=H^{+x-y}$.
\end{proof}

\noindent However, Theorem \ref{theo_cns} is useless from an algorithmic point of view, since checking property $J(\D_{\text{all}},\cdot\,)$ is not practical: indeed, identifying general dangers at a given $x$ is as difficult as identifying Maker wins. Recall that \makerbreaker is \PSPACE-complete, even when restricted to 4-uniform hypergraphs \cite{Gal25}. However, given some class of marked hypergraphs $\mathcal{H}$, if we find that an equivalence like that of Theorem \ref{theo_cns} holds (at least for non-trivial Maker wins, which is enough) for some family of dangers $\D$ that are identifiable in polynomial time, then we have proved that \makerbreaker is tractable on the class $\mathcal{H}$. 
For instance, this is possible for the class of 2-uniform marked hypergraphs:

\begin{theoreme}
	Let $\D_{\textup{triv}}$ be the singleton family containing the trivial danger of size 2. Let $H$ be a 2-uniform marked hypergraph that is not a trivial Maker win, with $|V(H) \setminus M(H)| \geq 2$. Then $H$ is a Breaker win if and only if $J(\D_{\textup{triv}},H)$ holds.
\end{theoreme}

\begin{proof}
	Since $H$ is 2-uniform, the fact that $H$ is not a trivial Maker win exactly means that $M(H)=\varnothing$, i.e., $H$ is a graph. Also note that a $\D_{\textup{triv}}$-danger at a vertex $x$ is nothing but an edge of the graph that is incident to $x$. The property $J(\D_{\textup{triv}},H)$, which by definition signifies the absence of any $\D_{\textup{triv}}$-fork in the graph $H$, is therefore equivalent to $H$ having maximum degree at most 1. This obviously characterizes graphs that are a Maker win: if $H$ has a vertex $x$ of degree at least 2, then Maker starts by picking $x$ and wins in the next round, otherwise Breaker wins by always picking the only neighbor of the vertex Maker has just picked (or an arbitrary vertex if that neighbor is already marked or does not exist).
\end{proof}

\noindent For a more complex class of marked hypergraphs $\mathcal{H}$, the trivial dangers are not enough, but one way to ensure that a given family of dangers $\D$ satisfies the desired equivalence is if $\D$ satisfies the following two properties:

\begin{enumerate}[noitemsep,nolistsep,label={\textup{(\roman*)}}]
	\item $\D$ contains the trivial dangers.
	\item $J(\D,\cdot\,)$ is hereditary in the sense that, for all $H \in \mathcal{H}$ with $|V(H) \setminus M(H)| \geq 2$: if $J(\D,H)$ holds, then, for all $x \in V(H) \setminus M(H)$, there exists $y \in V(H^{+x}) \setminus M(H^{+x})$ such that $J(\D,H^{+x-y})$ holds.
\end{enumerate}

\noindent Indeed, if (i) is satisfied, then the property $J(\D,\cdot\,)$ ensures that Breaker survives the next round, so if (ii) is also satisfied, then this property is maintained throughout which means Breaker survives until the end and wins. The issue is that (ii) seems hard to obtain, because dangers keep changing throughout the game: Maker's moves may create some, while Breaker's moves may destroy some. Therefore, the absence of any $\D$-fork before a certain round does not mean that there will be no $\D$-forks at the end of that round. \newline
We are now going to face these difficulties in a very concrete manner, for the class of 3-uniform marked hypergraphs. We will construct a family of dangers in two steps, starting with the most elementary ones, and then augmenting that family so that (ii) is satisfied.

\subsection{Forcing strategies: the families $\D_0$ and $\D_1$}\label{Subsection3-2}

\noindent In a 3-uniform marked hypergraph, consider an $x_0$-chain with edges $\{x_0,y_1,x_1\},\{x_1,y_2,x_2\}$, etc. where $x_0$ is the only marked vertex (see Figure \ref{Forcing}, left). If Maker is next to play, she can pick $x_1$, which forces Breaker to answer by picking $y_1$ because of the threat of the edge $\{x_0,y_1,x_1\}$. Maker can pick $x_2$ next, which forces Breaker to pick $y_2$, and so on. Where this "forcing strategy" becomes interesting for Maker is if she can reach an edge of the form $\{x_i,u,m\}$, where $u \neq x_i,y_i$ for all $i$, and $m$ is marked at this point. Indeed, after Maker picks $x_i$, she can win the game in the next round by picking either $y_i$ or $u$, and Breaker cannot stop both threats. There are two possibilities for this vertex $m$ (see Figure \ref{Forcing}, right): either $m \neq x_j$ for all $j<i$, i.e., we have an $x_0$-snake (nunchaku), or $m=x_j$ for some $j<i$, i.e., we have an $x_0$-tadpole.

\begin{figure}[h]
	\centering
	\includegraphics[scale=.55]{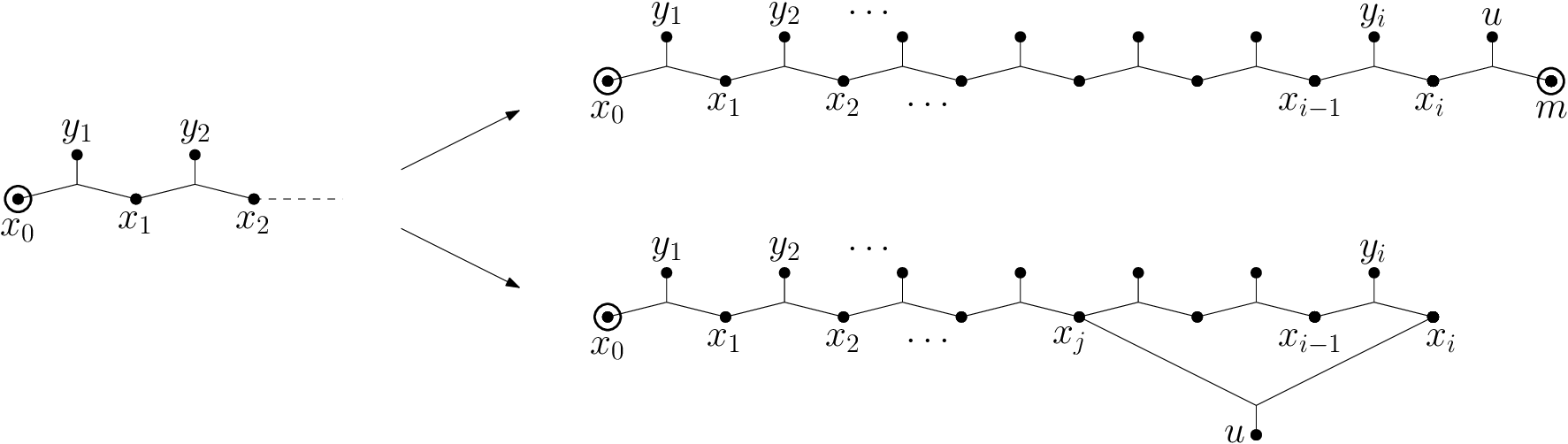}
	\caption{Two cases where Maker can win using a forcing strategy.}\label{Forcing}
\end{figure}

\noindent From this observation, we get our first elementary dangers in 3-uniform marked hypergraphs.

\begin{notation}\label{not:D0_D1}
    We define the family $\D_0$ of all pointed marked hypergraphs $(S,x)$ such that $S$ is an $x$-snake and $|M(S)|=1$. We also define the family $\D_1 \supseteq \D_0$ obtained from $\D_0$ by adding all pointed marked hypergraphs $(T,x)$ such that $T$ is an $x$-tadpole and $M(T)=\varnothing$. See Figure \ref{D1-dangers} for some examples.
\end{notation}

\begin{proposition}\label{prop_forcing}
	$\D_0$ and $\D_1$ are families of dangers.
\end{proposition}

\begin{proof}
    Let $(D,x_0) \in \D_1$. We want to show that $D^{+x_0}$ is a Maker win. This comes from the fact that we are in one of the two situations pictured on the right of Figure \ref{Forcing}. As we have just seen, when playing on $D^{+x_0}$, Maker can pick $x_1,\ldots,x_{i-1}$ successively (which forces Breaker to pick $y_1,\ldots,y_{i-1}$ successively in the meantime), then pick $x_i$ and win in the next round by picking whichever of $y_i$ or $u$ has not been picked by Breaker.
\end{proof}

\begin{figure}[h]
	\centering
	\includegraphics[scale=.55]{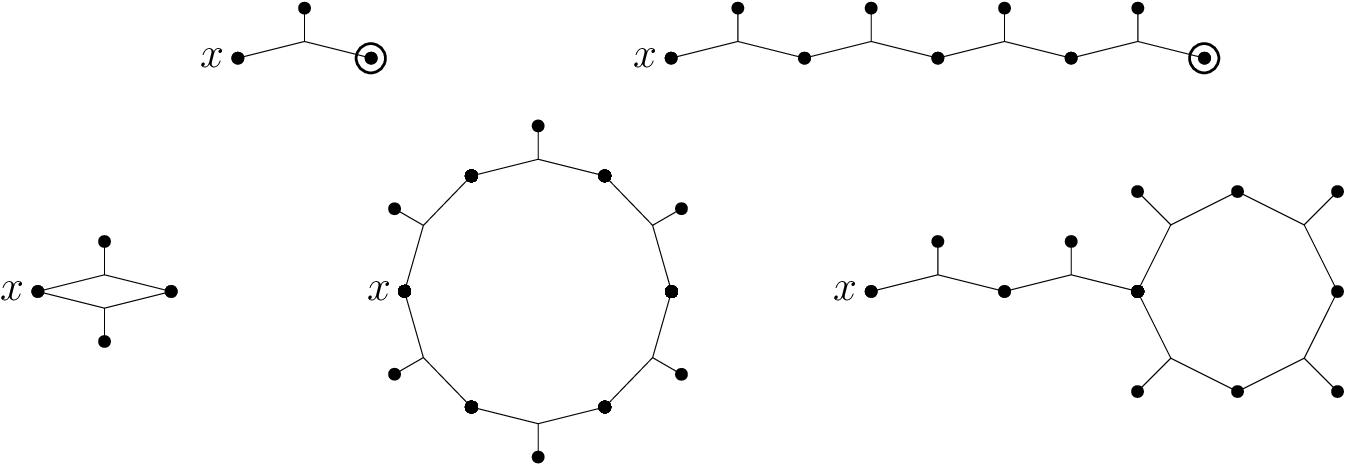}
	\caption{Five examples of a $\D_1$-danger at $x$. The topmost two are $\D_0$-dangers.}\label{D1-dangers}
\end{figure}

\noindent As an illustration of how the family $\D_1$ helps us understand the Maker-Breaker game, consider the tic-tac-toe hypergraph $H$ (Figure \ref{Example_Game}, left). It can easily be shown that $H$ is a Maker win, and that the center vertex or any corner vertex is a winning first move for Maker. This can be verified by brute force, but there is a simple proof using $\D_1$-dangers. For instance, let $x$ be the top-left vertex. Four $x$-cycles $C_1, C_2, C_3, C_4$ are highlighted in Figure \ref{Example_TicTacToe2}. Apart from $x$ itself, no vertex is in all four of these cycles, so by definition $\{C_1,C_2,C_3,C_4\}$ is a $\D_1$-fork at $x$ in $H$. This means $J(\D_1,H)$ does not hold, so $H$ is a Maker win and $x$ is a winning first move for Maker. This is a positive example for us in the sense that, just by considering the $\D_1$-dangers, Breaker can predict his loss by looking at $H$ statically (i.e., without having to imagine moves being played).
	
\begin{figure}[h]
	\centering
	\includegraphics[scale=.5]{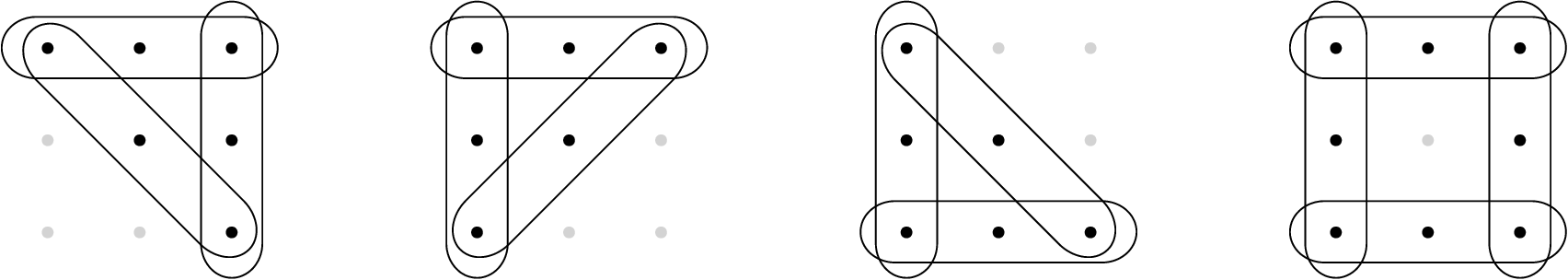}
	\caption{Four $x$-cycles in the tic-tac-toe hypergraph, where $x$ is the top-left vertex.}\label{Example_TicTacToe2}
\end{figure}

\subsection{Thinking one round ahead: the family $\D_2$}\label{Section3-3}

\noindent Unfortunately, property $J(\D_1,\cdot\,)$ is not enough to characterize Breaker wins among general 3-uniform marked hypergraphs. As an example, consider the hypergraph $H$ from Figure \ref{MakerWin5}. It is not too difficult to check that $J(\D_1,H)$ holds (recalling Figure \ref{Not_Tadpoles} should help here: for instance, there are no $x$-tadpoles in $H$). Now, say Maker starts by picking $x$, and assume by symmetry that Breaker picks some $y$ in the left half of $H$. After that first round, what can we say about the $\D_1$-dangers in $H^{+x-y}$? Consider the vertex $z$: Figure \ref{MakerWin6} highlights three specific $\D_1$-dangers at $z$ in $D^{+x} \subseteq H^{+x-y}$. The first two are a $z$-cycle $C$ and a $z$-tadpole $T$, that existed in $H$ from the beginning. The third one, on the other hand, is a $z$-snake $S$ that Maker has created from a $zx$-chain $P$ by marking $x$. Apart from $z$ itself, no vertex is in all three of these subhypergraphs, so by definition $\{C,T,S\}$ is a $\D_1$-fork at $z$ in $D^{+x}$. This means $J(\D_1,D^{+x})$ does not hold, so $D^{+x}$ is a Maker win and so is $H^{+x-y}$ by Proposition \ref{prop_subwin} (Subhypergraph Monotonicity). Since Breaker's pick $y$ was arbitrary, this shows $H$ is a Maker win. What we see here is a manifestation of the remark made just before Section \ref{Subsection3-2} about dangers in all generality: as $J(\,\cdot\,,H)$ is a static property of $H$, it might not tell anything about forks any further than the first round of play. As we have just seen in the case $\D=\D_1$, it can happen that there are no $\D$-forks anywhere at the beginning of a round, but Maker can pick a vertex that creates too many $\D$-forks for Breaker to destroy with his next pick. Contrary to the previous example (tic-tac-toe hypergraph), considering only the $\D_1$-dangers did not allow Breaker to predict his loss before the game began: only after the first round did he realize his fate. How could Breaker have foreseen the threat? The answer is that the subhypergraph $D$ and its symmetric counterpart $D'$ needed to be identified as dangers at $x$ in $H$. Indeed, we have just shown that $D^{+x}$ is a Maker win. By considering a bigger family $\D_2$ containing these, we would have identified $\{D,D'\}$ as a $\D_2$-fork at $x$ and concluded immediately that $H$ was a Maker win.

\begin{figure}[h] 
	\centering
	\includegraphics[scale=.55]{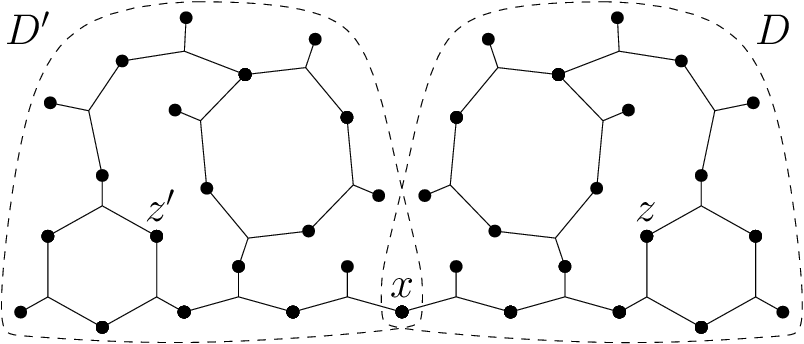}
	\caption{A Maker win $H$ such that $J(\D_1,H)$ holds.}\label{MakerWin5}
\end{figure}

\begin{figure}[h]
	\centering
	\includegraphics[scale=.55]{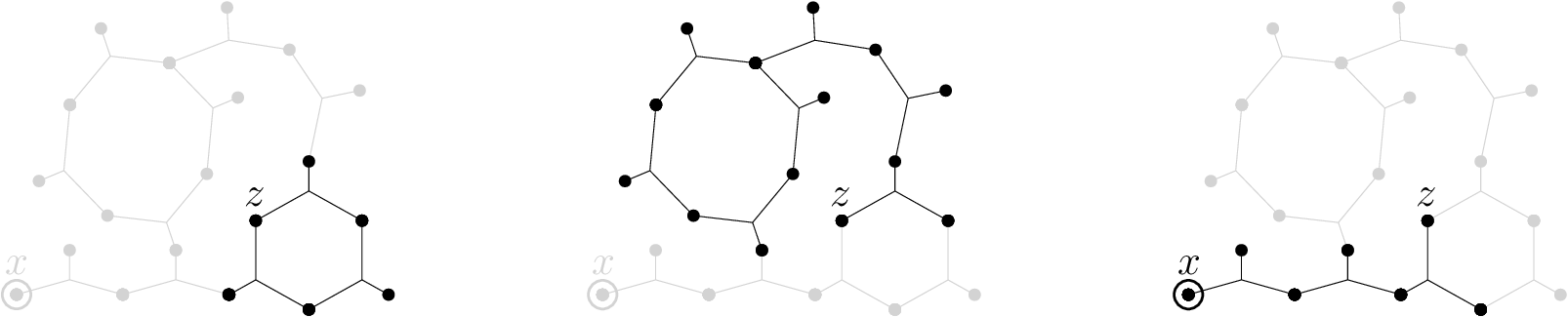}
	\caption{From left to right: a $z$-cycle, a $z$-tadpole and a $z$-snake in $D^{+x}$.}\label{MakerWin6}
\end{figure}

\noindent This example inspires us to also consider dangers $(D,x)$ such that $D^{+x}$ is the union of a $\D_1$-fork at some vertex $z$. This means $D$ is the union of a collection $\F_D$ that will become a $\D_1$-fork at $z$ if Maker marks $x$ (in a way, $\F_D$ is a "potential $\D_1$-fork" at $z$). That collection $\F_D$ can contain objects that already are $\D_1$-dangers at $z$ (think of $C$ and $T$ in the previous example), but also objects that will become $\D_1$-dangers at $z$ once $x$ is marked, i.e., $zx$-chains that will become snakes (think of $P$ in the previous example). This motivates the following definition, where item \ref{item_chapeau2} is there to avoid redundancies with dangers that are already in the family $\D_1$.

\begin{notation}\label{not:D1-hat}
	We define the family $\widehat{\D_1}$ of all pointed marked hypergraphs $(D,x)$ satisfying the following two properties:
	\begin{enumerate}[noitemsep,nolistsep,label={\textup{(\roman*)}}]
		\item There exists a \textit{decomposition} $(z , \F_D = \S_z \cup \T_z \cup \P_{zx})$ such that:
			\begin{itemize}[noitemsep,nolistsep]
				\item $z \in V(D) \setminus (M(D) \cup \{x\})$;
				\item $D = \union{\F_D}$;
				\item Each $S \in \S_z$ is a $z$-snake with $|M(S)|=1$ and $x \not\in V(S)$;
				\item Each $T \in \T_z$ is a $z$-tadpole with $M(T)=\varnothing$ and $x \not\in V(T)$;
				\item Each $P \in \P_{zx}$ is a $zx$-chain with $M(P)=\varnothing$;
				\item $\F_D^{+x}$ is a $\D_1$-fork at $z$ in $D^{+x}$.
			\end{itemize} \label{item_chapeau1}
		\item There are no $\D_1$-dangers at $x$ in $D$. \label{item_chapeau2}
	\end{enumerate}
	Finally, we also define $\D_2 = \D_1 \cup \widehat{\D_1}$.
\end{notation}

\noindent Using our notations from the previous example, where $D$ is the subhypergraph on the right of Figure \ref{MakerWin5}, a possible decomposition would be defined by $\S_z=\varnothing$, $\T_z=\{C,T\}$ and $\P_{zx}=\{P\}$. We then have $\F_D=\{C,T,P\}$, and $\F_D^{+x}=\{C,T,S\}$ is indeed a $\D_1$-fork at $z$ in $D^{+x}$.

\begin{proposition}\label{prop_D1chapeau}
    For all $(D,x) \in \widehat{\D_1}$, we have $|V(D^{+x}) \setminus M(D^{+x})| \geq 2$, and the property $J(\D_1,D^{+x})$ does not hold. In particular, $\widehat{\D_1}$ and $\D_2$ are families of dangers.
\end{proposition}

\begin{proof}
    Let $(D,x) \in \widehat{\D_1}$, with a decomposition $(z , \F_D = \S_z \cup \T_z \cup \P_{zx})$. Since $D = \union{\F_D}$, we know $\F_D \neq \varnothing$: let $X \in \F_D$, and let $e \in E(X)$. By definition, every element of $\F_D$ has at most one marked vertex, so $|M(e)|\leq 1$. Moreover, by item \ref{item_chapeau2} of the definition of $\widehat{\D_1}$, $e$ cannot contain both $x$ and a marked vertex, because $e$ would otherwise constitute a $\D_1$-danger at $x$ ($x$-snake of length 1). Therefore, $|e \setminus (M(e) \cup \{x\})| \geq 2$, from which $|V(D^{+x}) \setminus M(D^{+x})| \geq 2$. Now, since $\F_D^{+x}$ is a $\D_1$-fork at $z$ in $D^{+x}$, property $J(\D_1,D^{+x})$ does not hold. Therefore, $D^{+x}$ is a Maker win, i.e., $(D,x)$ is a danger.
\end{proof}

\noindent By construction, destroying the $\D_2$-dangers at Maker's pick $x$ is equivalent to destroying the $\D_1$-dangers at $x$ as well as all $\D_1$-forks that the marking of $x$ may have created anywhere.

\begin{proposition}\label{prop_equiv_dangers2}
	Let $H$ be a marked hypergraph with $|V(H)\setminus M(H)| \geq 4$, and suppose that $J(\D_1,H)$ holds. Let $x \in V(H) \setminus M(H)$ and $y \in V(H^{+x}) \setminus M(H^{+x})$. The following two assertions are equivalent:
	\begin{enumerate}[noitemsep,nolistsep,label={\textup{(\alph*)}}]
		\item $y \in \I{H^{+x}}{x\D_2(H)}$.
		\item $y \in \I{H^{+x}}{x\D_1(H)}$ and $J(\D_1,H^{+x-y})$ holds.
	\end{enumerate}
\end{proposition}

\begin{proof}
	Let us first show that (b) $\implies$ (a). Suppose for a contradiction that $y \in \I{H^{+x}}{x\D_1(H)}$ and that $J(\D_1,H^{+x-y})$ holds, but $y \not\in \I{H^{+x}}{x\D_2(H)}$. Then there exists $D \in x\widehat{\D_1}(H)$ such that $y \not\in V(D)$. This means $D^{+x}$ is a subhypergraph of $H^{+x-y}$. As a consequence, since $J(\D_1,H^{+x-y})$ holds, $J(\D_1,D^{+x})$ must hold as well. This contradicts Proposition \ref{prop_D1chapeau}. \newline
	Now, let us show that (a) $\implies$ (b). Suppose that $y \in \I{H^{+x}}{x\D_2(H)}$. Since $\D_1 \subseteq \D_2$, this obviously implies $y \in \I{H^{+x}}{x\D_1(H)}$, but suppose for a contradiction that $J(\D_1,H^{+x-y})$ does not hold. This means there exists a $\D_1$-fork $\F$ at some $z \in V(H^{+x-y}) \setminus M(H^{+x-y})$ in $H^{+x-y}$.
	We define the partition $\F = \S_z \cup \T_z \cup \S_{zx}$ as follows: 
	\begin{itemize}[noitemsep,nolistsep]
		\item The collection $\S_z$ (resp. $\T_z$) contains all $X \in \F$ such that $X$ is a $z$-snake (resp. a $z$-tadpole) and $x \not\in V(X)$. The elements of $\S_z \cup \T_z$ are "old" $\D_1$-dangers at $z$, i.e., ones that were already present in $H$.
		\item The collection $\S_{zx}$ contains all $X \in \F$ such that $x \in V(X)$. Since $x$ is marked in $H^{+x-y}$ (from which the collection $\F$ is taken), every $X \in \S_{zx}$ is necessarily a $zx$-snake with $M(X)=\{x\}$. The elements of $\S_{zx}$ are "new" $\D_1$-dangers at $z$, i.e., ones that did not exist in $H$. We assume that $\S_{zx} \neq \varnothing$, as $\F$ would otherwise be a $\D_1$-fork at $z$ in $H$, contradicting the fact that $J(\D_1,H)$ holds (this is the only moment where this assumption is used in this proof).
	\end{itemize}
	Let $\P_{zx}$ be the collection of all $zx$-chains in $H$ that the elements of $\S_{zx}$ come from, i.e., the collection with the same elements as $\S_{zx}$ except that $x$ is non-marked in all of them. Let $D=\union{\F_D}$, where $\F_D =  \S_z \cup \T_z \cup \P_{zx}$. Note that $x \in V(D)$ since $\P_{zx} \neq \varnothing$. On the other hand, since all elements of $\F$ are subhypergraphs of $H^{+x-y}$, we know all elements of $\F_D$ are subhypergraphs of $H^{-y}$, so $y \not\in V(D)$. Let us verify that $(D,x) \in \widehat{\D_1}$, with decomposition $(z,\F_D)$. For item \ref{item_chapeau1}, all that is left to check is that $\F_D^{+x}$ is a $\D_1$-fork at $z$ in $D^{+x}$: since $\F_D^{+x} = \F$ and $D^{+x} \subseteq  H^{+x-y}$, this immediately follows from the fact that $\F$ is a $\D_1$-fork at $z$ in $H^{+x-y}$. As for item \ref{item_chapeau2}, it is impossible that $D$ contains a $\D_1$-danger at $x$ since $y \in \I{H^{+x}}{x\D_1(H)}$ and $y \not\in V(D)$. All in all, items \ref{item_chapeau1} and \ref{item_chapeau2} are both satisfied, so $(D,x) \in \widehat{\D_1}$. This contradicts the fact that $y \in \I{H^{+x}}{x\D_2(H)}$ and $y \not\in V(D)$.
\end{proof}

\noindent It may seem that the addition of the $\widehat{D_1}$-dangers simply postpones the issue by one round: property $J(\D_2,H)$ ensures that there will be no $\D_1$-forks during the first two rounds of play on $H$, but what about the third round? As we will see however, property $J(\D_2,H)$ actually turns out to be sufficient for Breaker to prevent any $\D_1$-fork from appearing during the entire game, which means it fully characterizes Breaker wins. \newline From Maker's point of view, marked hypergraphs $H$ such that property $J(\D_2,H)$ does not hold are ones where Maker can force the appearance of a nunchaku or a necklace after at most three rounds of play. We are talking about full rounds of play, meaning that there is a nunchaku or a necklace in the updated marked hypergraph after Breaker's move. Figure \ref{MakerWin4} shows an example where three rounds are needed (it can be checked that Maker cannot force the appearance of a nunchaku or a necklace in less than three rounds).

\begin{proposition}\label{prop_interpretation}
	Let $H$ be a marked hypergraph with $|V(H) \setminus M(H)| \geq 2$. If property $J(\D_1,H)$ (resp. $J(\D_2,H)$) does not hold, then Maker has a strategy ensuring that there is a nunchaku or a necklace in the updated marked hypergraph obtained after at most two (resp. three) full rounds of play on $H$.
\end{proposition}

\begin{proof}
	Let $\alpha \in \{1,2\}$, and suppose $J(\D_{\alpha},H)$ does not hold. Maker picks $x_1$ such that $\I{H^{+x_1}}{x_1\D_{\alpha}(H)}=\varnothing$, and Breaker picks $y_1$, so that the updated marked hypergraph after one round of play is $H^{+x_1-y_1}$. Since $\I{H^{+x_1}}{x_1\D_{\alpha}(H)}=\varnothing$, there exists some $D \in x_1\D_{\alpha}(H)$ such that $y_1 \not\in V(D)$.
	\begin{itemize}[wide,noitemsep,nolistsep]
		\item Case 1: $D$ is an $x_1$-snake (resp. an $x_1$-cycle). Then, $D^{+x_1}$ is a nunchaku (resp. a necklace) in $H^{+x_1-y_1}$, which concludes.
		\item Case 2: $D=T$ is an $x_1$-tadpole that is not a cycle. Then, let $x_2$ be the only vertex in $V(P_T) \cap V(C_T)$. Maker picks $x_2$, and Breaker picks some $y_2$, so that the updated marked hypergraph after two rounds of play is $H^{+x_1-y_1+x_2-y_2}$. Either $y_2 \not\in V(P_T)$, in which case $P_T^{+x_1+x_2}$ is a nunchaku in $H^{+x_1-y_1+x_2-y_2}$, or $y_2 \not\in V(C_T)$, in which case $C_T^{+x_2}$ is a necklace in $H^{+x_1-y_1+x_2-y_2}$. This concludes.
		\item Case 3 (which can only happen if $\alpha=2$): $(D,x_1) \in \widehat{D_1}$ with a maximal decomposition $(x_2,\F_D)$. Then $\F_D^{+x_1}$ is a $\D_1$-fork at $x_2$ in $H^{+x_1-y_1}$. Maker picks $x_2$, and Breaker picks some $y_2$. By definition of a $\D_1$-fork, there exists $D' \in x_2\D_1(H^{+x_1-y_1})$ such that $y_2 \not\in V(D')$. If $D'$ is an $x_2$-snake (resp. an $x_2$-cycle), then we get a nunchaku (resp. a necklace) in $H^{+x_1-y_1+x_2-y_2}$ like in Case 1. If $D'=T$ is an $x_2$-tadpole, then Maker picks the only vertex $x_3$ in $V(P_T) \cap V(C_T)$ like in Case 2, and we get a nunchaku or a necklace in $H^{+x_1-y_1+x_2-y_2+x_3-y_3}$ whatever Breaker's pick $y_3$. \qedhere
	\end{itemize}
\end{proof}

\begin{figure}[h]
	\centering
	\includegraphics[scale=.55]{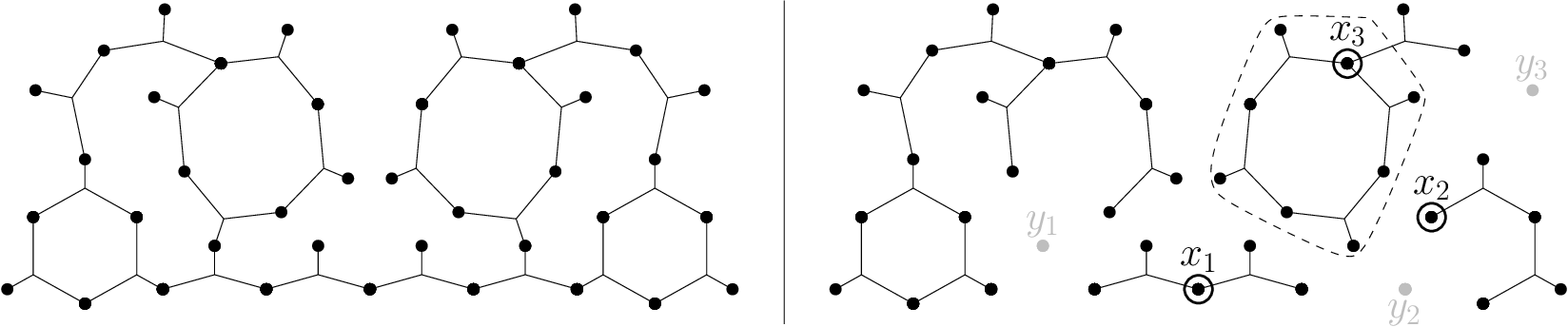}
	\caption{Left: the hypergraph $H$ from Figure \ref{MakerWin5}. Right: one possibility for $H^{+x_1-y_1+x_2-y_2+x_3-y_3}$, in which a necklace is highlighted.}\label{MakerWin4}
\end{figure}

\noindent Proposition \ref{prop_interpretation} is crucially relevant to Conjecture \ref{conjecture}. Indeed, using the reduction from Proposition \ref{prop_reduction}, any nunchaku or necklace is equivalent for the game to the hypergraph pictured in Figure \ref{Nunchaku0}, which describes exactly the structure of the manriki formula from \cite{RW20}. Therefore, Conjecture \ref{conjecture} implies that there exists a constant $r$ such that Maker wins on a 3-uniform marked hypergraph if and only if she can force the appearance of a nunchaku or a necklace within at most $r$ rounds of play. In this paper, we prove that this is true for $r=3$.

\begin{figure}[h]
	\centering
	\includegraphics[scale=.55]{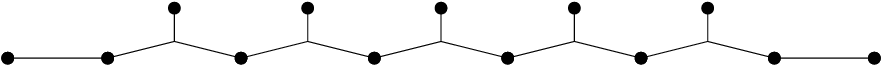}
	\caption{The hypergraph version of the manriki formula: two edges of size 2 linked by a chain.}\label{Nunchaku0}
\end{figure}

\subsection{Statement of the main results}

\noindent We now state the five main results of this paper, which will be proved in Section \ref{Section5}. We start by giving three classes of 3-uniform marked hypergraphs in which Breaker wins can be characterized in terms of $\D_0$-forks, $\D_1$-forks and $\D_2$-forks respectively. The first result, which addresses 3-uniform marked hyperforests (i.e., there are no cycles), is much less involved than the other two.

\begin{theoreme}\label{theo_main_structure3}
	Let $H$ be a 3-uniform marked hyperforest that is not a trivial Maker win, with $|V(H) \setminus M(H)| \geq 2$. Then $H$ is a Breaker win if and only if $J(\D_0,H)$ holds.
\end{theoreme}

\noindent On the bigger class of linear hypergraphs of rank 3, previous work had been made by Kutz. His main result \cite[Theorem 38]{Kut04} characterizes the special structure of Breaker wins for a subclass that he reduces to, namely, the class of connected linear hypergraphs of rank 3 with no articulation vertices that have exactly one edge of size 2 (in our setting of 3-uniform marked hypergraphs, an edge of size 2 would be modeled as an edge of size 3 with a marked vertex). 
Interestingly, by carefully reading Kutz's proof, it can be seen that the absence of that special structure implies the existence of some simple subhypergraph $X \subseteq H$ which is a Maker win. It can actually be checked in all cases that, not only is $X$ a Maker win, but in fact $J(\D_1,X)$ does not hold. All in all, it can be derived from Kutz's proof that, for all $H$ in the considered class (apart from some trivial cases), $H$ is a Breaker win if and only if $J(\D_1,H)$ holds. We give a new independent proof of this through our second main result, whose statement is actually slightly stronger.

\begin{theoreme}\label{theo_main_structure2}
	Let $H$ be a 3-uniform marked hypergraph that is not a trivial Maker win, with $|V(H) \setminus M(H)| \geq 2$. Suppose that, for any $x \in V(H) \setminus M(H)$, there exists an $x$-snake in $H$. Then $H$ is a Breaker win if and only if $J(\D_1,H)$ holds. Moreover, $H$ is a Maker win if and only if Maker has a strategy ensuring that there is a nunchaku or a necklace in the updated marked hypergraph obtained after at most two full rounds of play on $H$.
\end{theoreme}

\noindent We now state the most central result of this paper, which solves the general 3-uniform case.

\begin{theoreme}\label{theo_main_structure1}
	Let $H$ be a 3-uniform marked hypergraph that is not a trivial Maker win, with $|V(H) \setminus M(H)| \geq 2$. Then $H$ is a Breaker win if and only if $J(\D_2,H)$ holds. More precisely:
	\begin{enumerate}[noitemsep,nolistsep,label={\textup{(\roman*)}}]
		\item If $J(\D_2,H)$ does not hold, then $H$ is a Maker win and: any $x_1 \in V(H) \setminus M(H)$ such that $\I{H^{+x_1}}{x_1\D_2(H)}=\varnothing$ is a winning first pick for Maker. 
		\item If $J(\D_2,H)$ holds, then $H$ is a Breaker win and: for any first pick $x_1 \in V(H) \setminus M(H)$ of Maker, any $y_1 \in \I{H^{+x_1}}{x_1\D_2(H)}$ is a winning answer for Breaker. 
	\end{enumerate}
	Moreover, $H$ is a Maker win if and only if Maker has a strategy ensuring that there is a nunchaku or a necklace in the updated marked hypergraph obtained after at most three full rounds of play on $H$.
\end{theoreme}

\noindent Our fourth main result is algorithmic. We show that there exists a polynomial-time algorithm that decides whether a 3-uniform marked hypergraph $H$ is a Maker win. This improves on the result of Kutz \cite{Kut04}, which showed the same in the linear case only. Together with Theorem \ref{theo_main_structure1}, and recalling the observation made at the end of Section \ref{Section3-3} about nunchakus and necklaces corresponding to Rahman and Watson's manriki, this validates Conjecture \ref{conjecture} with $r=3$ for positive 3-CNF formulas. The proof relies on the fact that Theorem \ref{theo_main_structure1} yields an immediate reduction to the chain existence problem, which is tractable \cite{GGS22}.

\begin{theoreme}\label{theo_main_algo}
	\makerbreaker is solved in polynomial time on 3-uniform marked hypergraphs, i.e., hypergraphs of rank 3. More precisely, there exists an algorithm which, given a 3-uniform marked hypergraph with $n$ vertices, $m$ edges and maximum degree $\Delta$, decides whether $H$ is a Maker win in time $O(\max(n^5m^2,n^6\Delta))$.
\end{theoreme}

\noindent Finally, our fifth main result will come as an easy consequence of Theorem \ref{theo_main_structure1}. It states that, if Maker has a winning strategy on a 3-uniform marked hypergraph, then she can get a fully marked edge in a logarithmic number of rounds. The idea is that, once there is a nunchaku or necklace, Maker can pick a vertex in the middle of it to guarantee a nunchaku of half the length in the next round: Maker may iterate this process until the nunchaku is of length 1, at which point she wins on the spot.

\begin{theoreme}\label{theo_main_duration}
	Let $H$ be a 3-uniform marked hypergraph with $|V(H)\setminus M(H)|\geq 6$. If $H$ is a Maker win, then $\tau_M(H) \leq 3+\lceil \log_2(|V(H)\setminus M(H)|-5)\rceil$.
\end{theoreme}

\section{Structural preliminaries}\label{Section4}

\noindent In this section, we set the game aside to perform a structural study of all objects that are relevant to our problem. We start with the most simple ones: chains, cycles and tadpoles. We then go up a level to look at $\widehat{\D_1}$-dangers. All established properties will be used to prove our main results in Section \ref{Section5}.

\subsection{Structural properties of chains, cycles and tadpoles}

\subsubsection{Substructure lemmas}

\noindent We address the existence, and sometimes uniqueness, of chains and tadpoles inside other chains and tadpoles. These results are intuitively clear just by looking at a drawing, but we give rigorous proofs using walks.

\begin{lemme}\label{lemma_subchain1}
	Let $P$ be a chain and let $u,v \in V(P)$. Then there exists a unique $uv$-chain in $P$.
\end{lemme}

\begin{proof}
	Let $a,b$ be such that $P$ is an $ab$-chain, and write $\ora{aPb}=(a,e_1,\ldots,e_L,b)$.
	\begin{itemize}[wide]
		\item Firstly, suppose $u=v$. Then that single vertex forms the only $uv$-chain in $P$.
		\item Secondly, suppose $u \neq v$ and there exists some $1 \leq i \leq L$ such that $\{u,v\}\subseteq e_i$ (note that $i$ is unique since two distinct edges of a chain cannot intersect on two vertices). Then $(u,e_i,v)$ represents a $uv$-chain. Moreover, if some walk $\ora{W'}$ represents a $uv$-chain in $P$, then we have $u \in \Start(\ora{W'})$ and $v \in \End(\ora{W'})$, so $\Start(\ora{W'})=\End(\ora{W'})=e_i$ hence the uniqueness.
		\item Finally, suppose $u \neq v$ and no edge of $P$ contains both $u$ and $v$. For $x \in \{u,v\}$, define $j(x)=\min\{1 \leq i \leq L \mid x  \in e_i\}$ and $j'(x)=\max\{1 \leq i \leq L \mid x \in e_i\}$: note that $j'(x)=j(x)+1$ if $x \in \inn(P)$ and $j'(x)=j(x)$ otherwise. Up to swapping the roles of $u$ and $v$, assume $j(u) \leq j(v)$: we actually have $j(u) < j(v)$, otherwise $e_{j(u)}=e_{j(v)}$ would contain both $u$ and $v$. Since $j'(u) \in \{j(u),j(u)+1\}$, this yields $j'(u) \leq j(v)$ hence $j'(u) < j(v)$ for the same reason. We define $\ora{W}= (u,e_{j'(u)},e_{j'(u)+1},\ldots,e_{j(v)},v)$, and we claim that $\ora{W}$ is a walk that represents a $uv$-chain. Indeed:
			\begin{itemize}[noitemsep,nolistsep]
				\item The fact that $\ora{aPb}$ is a linear walk by definition of a chain, coupled with the fact that $u \in e_{j'(u)}$ and $v \in e_{j(v)}$, implies that $\ora{W}$ is a linear walk.
				\item The fact that the walk $\ora{aPb}$ is simple by definition of a chain, coupled with the maximality of $j'(u)$ and the minimality of $j(v)$, implies that $\ora{W}$ is also simple.
			\end{itemize}
			Let us now address uniqueness. Let $\ora{W'}= (u,e_{i_1},e_{i_2},\ldots,e_{i_t},v)$ be a walk representing a $uv$-chain in $P$, where $i_1,\ldots,i_t$ are pairwise distinct indices in $\{1,\ldots,L\}$. Since $u \in e_{i_1}$ and $v \in e_{i_t}$, we have $i_1 \in \{j(u),j'(u)\}$ and $i_t \in \{j(v),j'(v)\}$. We have seen that $j'(u)<j(v)$, so $i_1<i_t$. For all $1 \leq s \leq t-1$, we have $|e_{i_s} \cap e_{i_{s+1}}|=1$ by definition of a chain hence $|i_s-i_{s+1}|=1$. Since $i_1<i_t$ and the indices $i_1,\ldots,i_t$ are pairwise distinct, this implies $i_{s+1}=i_s+1$ for all $1 \leq s \leq t-1$. To conclude that $\ora{W'}=\ora{W}$, it only remains to show that $i_1=j'(u)$ and $i_t=j(v)$. We have mentioned that $i_1 \in \{j(u),j'(u)\}$: if $i_1=j(u)=j'(u)-1$, then $e_{i_2}=e_{j'(u)} \ni u$, hence a repetition in $\ora{W'}$ which contradicts the definition of a chain. Therefore $i_1=j'(u)$, and an analogous reasoning yields $i_t=j(v)$.\qedhere
	\end{itemize}
\end{proof}

\noindent We are also interested in the existence of chains inside cycles. First of all, we need to describe what happens when we delete a vertex from a cycle:

\begin{lemme}\label{lemma_subchain2}
	Let $C$ be a cycle and let $w \in V(C)$. Let $w_1,w_2$ be the two inner vertices of $C$ that are adjacent to $w$ in $C$ (if $C$ is of length 2 and $w \in \inn(C)$, then $w_1=w_2$).
	\begin{itemize}[noitemsep,nolistsep]
		\item If $w \in \out(C)$, then $C^{-w}$ is a $w_1w_2$-chain.
		\item If $w \in \inn(C)$, then $C^{-w}$ is the union of a $w_1w_2$-chain and two isolated vertices which are the two outer vertices of $C$ that are adjacent to $w$ in $C$.
	\end{itemize}
\end{lemme}

\begin{proof}
	Let us first address the case where $C$ is of length 2. If $w \in \out(C)$, then write $E(C)=\{\{w_1,w,w_2\},\{w_1,u,w_2\}\}$: $C^{-w}$ consists of the edge $\{w_1,u,w_2\}$, which forms a $w_1w_2$-chain. If $w \in \inn(C)$, then write $E(C)=\{\{w,u_1,w_1\},\{w,u_2,w_1\}\}$: $C^{-w}$ consist of the three isolated vertices $w_1=w_2$, $u_1$ and $u_2$.
	\\ Now assume that $C$ is of length at least 3. Let $e$ be the edge of $C$ containing both $w$ and $w_1$, and write $\ora{(w_1-e)C}=(w_1,e=e_1,e_2\ldots,e_L,w_1)$. We have $e_1 \cap e_L = \{w_1\}$. If $w \in \out(C)$, then $e_1=\{w_1,w,w_2\}$ so $e_1 \cap e_2=\{w_2\}$. If $w \in \inn(C)$, then $e_1 \cap e_2=\{w\}$ hence $e_2 \cap e_3 = \{w_2\}$ since $w_2$ is adjacent to $w$. Therefore, defining $i=2$ if $w \in \out(C)$ and $i=3$ if $w \in \inn(C)$, the only edges of $C$ containing $w_2$ are $e_{i-1}$ and $e_i$. We define $\ora{W}=(w_2,e_i,\ldots,e_L,w_1)$, and we claim that $\ora{W}$ is a linear simple walk. Indeed:
	\begin{itemize}[noitemsep,nolistsep]
		\item[--] By definition of a cycle, the walk $\ora{(w_1-e)C}$ is linear since $C$ is of length at least 3, and $w_1$ is its only repeated vertex with $\{1 \leq i \leq L, w_1 \in e_i\}=\{1,L\}$. Therefore, its subsequence $(e_i,\ldots,e_L,w_1)$ is also a linear walk, and has no repeated vertices since it does not contain the edge $e_1$.
		\item[--] The addition of $w_2$ at the start of $(e_i,\ldots,e_L,w_1)$ preserves the fact that it is a linear walk since $w_2 \in e_i$, and also preserves the absence of any repeated vertex since $w_2 \not\in e_j$ for all $j>i$.
	\end{itemize}
	Therefore, by definition, $\ora{W}$ represents a $w_2w_1$-chain. We can now conclude:
	\begin{itemize}[noitemsep,nolistsep]
		\item If $w \in \out(C)$, then $V(C^{-w}) = V(C) \setminus \{w\} = e_2 \cup \ldots \cup e_L = V(\ora{W})$ and $E(C^{-w}) = E(C) \setminus \{e_1\} =\{e_2,\ldots,e_L\}=E(\ora{W})$, so $C^{-w}$ is the $w_1w_2$-chain represented by $\ora{W}$.
		\item If $w \in \inn(C)$, then let $u_1$ and $u_2$ be the outer vertices of $C$ in $e_1$ and $e_2$ respectively: we have $V(C^{-w}) = V(C) \setminus \{w\} = (e_3 \cup \ldots \cup e_L) \cup \{u_1,u_2\} = V(\ora{W}) \cup \{u_1,u_2\}$ and $E(C^{-w}) = E(C) \setminus \{e_1,e_2\} =\{e_3,\ldots,e_L\}=E(\ora{W})$, so $C^{-w}$ is the union of the $w_1w_2$-chain represented by $\ora{W}$ and the two isolated vertices $u_1$ and $u_2$. \qedhere
	\end{itemize}
\end{proof}

\noindent We can now conclude about the existence of chains between two given vertices of a cycle, first when trying to avoid a third vertex, then in general.

\begin{lemme}\label{lemma_subchain3}
	Let $C$ be a cycle and let $u,v,w \in V(C)$ with $w \neq u,v$. Then there exists a unique $uv$-chain in $C$ that does not contain $w$, unless all the following hold: $w \in \inn(C)$, $u \neq v$, and $u$ or $v$ is an outer vertex of $C$ that is adjacent to $w$ (in which case there exist none).
\end{lemme}

\begin{proof}
	First of all, note that a $uv$-chain in $C$ that does not contain $w$ is exactly a $uv$-chain in $C^{-w}$. Assume $u \neq v$, otherwise the result is trivial. If $w \in \out(C)$, then $C^{-w}$ is a chain according to Substructure Lemma \ref{lemma_subchain2}, which contains a unique $uv$-chain by Substructure Lemma \ref{lemma_subchain1}. Now assume $w \in \inn(C)$: then $C^{-w}$ is the union of a chain $P$ and two isolated vertices $u_1,u_2$ that are the two outer vertices of $C$ adjacent to $w$ according to Substructure Lemma \ref{lemma_subchain2}. If $u \in \{u_1,u_2\}$ or $v \in \{u_1,u_2\}$, then there obviously cannot exist a $uv$-chain in $C^{-w}$. Otherwise $u,v \in V(P)$, so there exists a unique $uv$-chain in $P$ (and in $C^{-w}$ as a result) by Substructure Lemma \ref{lemma_subchain1}.
\end{proof}

\begin{lemme}\label{lemma_subchain3bis}
	Let $C$ be a cycle and let $u,v \in V(C)$. Then there exists a $uv$-chain in $C$, unless $C$ is of length 2 and $\out(C)=\{u,v\}$.
\end{lemme}

\begin{proof}
	If $C$ is of length 2 and $\out(C)=\{u,v\}$, then there are no $uv$-chains in $C$, because $|e_u \cap e_v|=2$ where $e_u$ (resp. $e_v$) denotes the only edge of $C$ containing $u$ (resp. $v$). Otherwise, there exists $w \in \out(C) \setminus \{u,v\}$: by Substructure Lemma \ref{lemma_subchain3}, there exists a unique $uv$-chain in $C$ that does not contain $w$, so in particular $C$ contains a $uv$-chain.
\end{proof}

\noindent We now give analogous results for tadpoles.

\begin{lemme}\label{lemma_subchain4}
	Let $T$ be a tadpole and let $u,v,w \in V(T)$. If $w \in \out(C_T) \setminus \{u,v\}$, then there exists a $uv$-chain in $T$ that does not contain $w$.
\end{lemme}

\begin{proof}
	Note that $w \not\in V(P_T)$, so that Substructure Lemma \ref{lemma_subchain1} concludes if $u,v \in V(P_T)$. If $u,v \in V(C_T)$, then Substructure Lemma \ref{lemma_subchain3} concludes. Therefore, assume $u \in V(P_T)$ and $v \in V(C_T)$. Let $b$ be the only vertex in $V(P_T) \cap V(C_T)$. By Substructure Lemma \ref{lemma_subchain1}, there exists a $ub$-chain $P_{ub}$ in $P_T$, that does not contain $w$ since $w \not\in V(P_T)$. By Substructure Lemma \ref{lemma_subchain3}, there exists a $bv$-chain $P_{bv}$ in $C_T$ that does not contain $w$. Since $V(P_{ub}) \cap V(P_{bv})=\{b\}$, it is clear that $\ora{uP_{ub}b} \oplus \ora{bP_{bv}v}$ represents a $uv$-chain in $T$ that does not contain $w$.
\end{proof}

\begin{lemme}\label{lemma_subchain4bis}
	Let $T$ be a tadpole and let $u,v \in V(T)$. Then there exists a $uv$-chain in $T$, unless $C_T$ is of length 2 and $\out(C_T)=\{u,v\}$.
\end{lemme}

\begin{proof}
	If $C_T$ is of length 2 and $\out(C_T)=\{u,v\}$, then there are no $uv$-chains in $T$, because $|e_u \cap e_v|=2$ where $e_u$ (resp. $e_v$) denotes the only edge of $T$ containing $u$ (resp. $v$). Otherwise, there exists $w \in \out(C) \setminus \{u,v\}$: by Substructure Lemma \ref{lemma_subchain4}, there exists a $uv$-chain in $T$ that does not contain $w$, so in particular $T$ contains a $uv$-chain.
\end{proof}

\begin{lemme}\label{lemma_subtadpole}
	Let $T$ be a tadpole and let $u \in V(T) \setminus \out(C_T)$. Then $T$ contains a $u$-tadpole.
\end{lemme}

\begin{proof}
	Let $b$ be the only vertex in $V(P_T) \cap V(C_T)$. Since $u \not\in \out(C_T)$, we have $u \in \inn(C_T)$ or $u \in V(P_T)$. If $u \in \inn(C_T)$, then $C_T$ is a $u$-cycle (recall that a $u$-cycle is a particular case of a $u$-tadpole). If $u \in V(P_T)$, then there exists a $ub$-chain $P_{ub}$ in $P_T$ by Substructure Lemma \ref{lemma_subchain1}, so $\ora{uP_{ub}b} \oplus \ora{bC_T}$ represents a $u$-tadpole.
\end{proof}

\subsubsection{Projections}

\noindent One of the most common tools that we will use is, inside a chain or a tadpole, to follow a subchain starting from some vertex $u$ until reaching some set of vertices $Z$, as made possible by the previous results:

\begin{proposition}\label{prop_projection}
	Let $H$ be a hypergraph. Let $X$ be a chain or a tadpole in $H$, let $u \in V(X)$, and let $Z \subseteq V(H)$ be such that $Z \cap V(X) \neq \varnothing$. In the case where $X$ is a tadpole with $C_X$ of length 2 and $u \in \out(C_X)$, also suppose that $Z \cap V(X) \neq \out(C_X) \setminus \{u\}$. Then there exists a $u$-chain $\proj{Z}{u}{X}$ in $X$ such that:
	\begin{itemize}[noitemsep,nolistsep]
		\item If $u \in Z$, then $\proj{Z}{u}{X}$ is of length 0.
		\item If $u \not\in Z$, then $\,\proj{Z}{u}{X}$ is of positive length and its only edge intersecting $Z$ is $\End(\ora{u\proj{Z}{u}{X}})$, with $|\End(\ora{u\proj{Z}{u}{X}}) \cap Z| \in \{1,2\}$.
	\end{itemize}
\end{proposition}

\begin{proof}
	Let us start by showing the existence of $z \in Z \cap V(X)$ such that there exists a $uz$-chain in $X$. If $X$ is a chain, then any $z \in Z \cap V(X)$ is suitable by Substructure Lemma \ref{lemma_subchain1}. If $X$ is a tadpole, then any $z \in Z \cap V(X)$ is suitable by Substructure Lemma \ref{lemma_subchain4bis}, unless $C_X$ is of length 2 and $u \in \out(C_X)$ in which case we choose $z \in Z \cap V(X) \setminus (\out(C_X) \setminus \{u\})$ as allowed by the assumption.
	\\ Let $z \in Z \cap V(X)$ minimizing the length of a shortest $uz$-chain in $X$, and let $P$ be a shortest $uz$-chain in $X$. We define $\proj{Z}{u}{X} = P$, which we now prove has the desired properties. Clearly, $P$ is of positive length if and only if $u \not\in Z$. Assume $u \not\in Z$. Consider the walk $\ora{uP}\vert_Z$: recall that, as per Notation \ref{notation_walk}, this is the walk $\ora{uP}$ ``cut'' at the first edge that intersects $Z$. In particular, the only edge of $\ora{uP}\vert_Z$ that intersects $Z$ is $\End(\ora{uP}\vert_Z)$, so $|\End(\ora{uP}\vert_Z) \cap Z| \in \{1,2\}$. Therefore, it suffices to show that $\ora{uP}\vert_Z=\ora{uP}$ to finish the proof. Let $z' \in \End(\ora{uP}\vert_Z)$. The walk $\ora{uP}\vert_Z$ induces a $uz'$-chain, which cannot be shorter than $P$ by minimality of $z$, hence why $\ora{uP}\vert_Z=\ora{uP}$.
\end{proof}

\begin{remarque}
	There is not necessarily uniqueness, even if $X$ is a chain: indeed, it is possible that there are vertices of $Z$ on both sides of $u$ in the chain.
\end{remarque}

\begin{definition}\label{def:projection}
	For $X,u,Z$ satisfying the required conditions, a $u$-chain $\proj{Z}{u}{X}$ from Proposition \ref{prop_projection} is called a \textit{projection of $u$ onto $Z$ in $X$}. Figure \ref{Projections} features some examples. As there is no uniqueness in general, we will consider that the notation $\proj{Z}{u}{X}$ always refers to the same chain for given $X,u,Z$. We normally use the notation once anyway, to give ourselves one arbitrary such projection and then work with that one. We also emphasize that, despite the argument used in the proof of Proposition \ref{prop_projection} to guarantee their existence, projections are not defined to be shortest possible, as this will not be needed. The example on the right of Figure \ref{Projections} illustrates this.
\end{definition}

\begin{figure}[h]
	\centering
	\includegraphics[scale=.58]{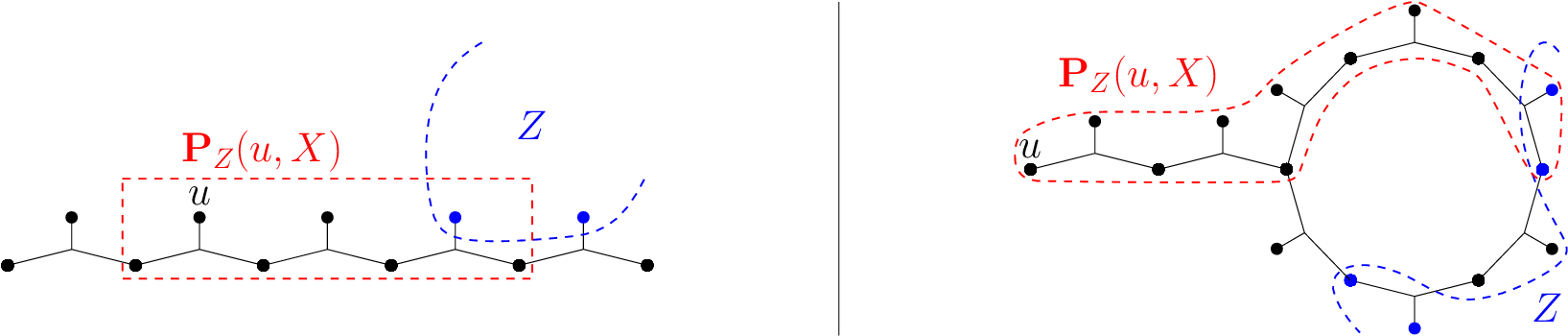}
	\caption{Examples of projections. Left: $X$ is a chain and $|\End(\ora{u\proj{Z}{u}{X}}) \cap Z|=1$. Right: $X$ is a tadpole and $|\End(\ora{u\proj{Z}{u}{X}}) \cap Z|=2$.}
	\label{Projections}
\end{figure}

\subsubsection{Union lemmas}

\noindent We now look at some structures that appear in unions of chains and tadpoles. The following three lemmas are immediately deduced from the concatenation of the walks representing the chains, cycles and tadpoles involved in their statements. We will use them often without necessarily referencing them.

\begin{lemme}
	If $P$ is an $ab$-chain and $P'$ is a $bc$-chain such that $V(P) \cap V(P')=\{b\}$, then $P \cup P'$ is an $ac$-chain.
\end{lemme}

\begin{lemme}
	If $P$ and $P'$ are $ab$-chains such that $V(P) \cap V(P')=\{a,b\}$, then $P \cup P'$ is an $a$-cycle and a $b$-cycle.
\end{lemme}

\begin{lemme}
	If $P$ is an $ab$-chain and $T$ is a $b$-tadpole such that $V(P) \cap V(T)=\{b\}$, then $P \cup T$ is an $a$-tadpole.
\end{lemme}

\noindent However, when the intersection of the two objects is more complex, it is less clear what their union contains. For example, Figure \ref{NonTransitive} illustrates the fact that the union of an $ab$-chain and a $bc$-chain does not necessarily contain an $ac$-chain.

\begin{figure}[h]
	\centering
	\includegraphics[scale=.58]{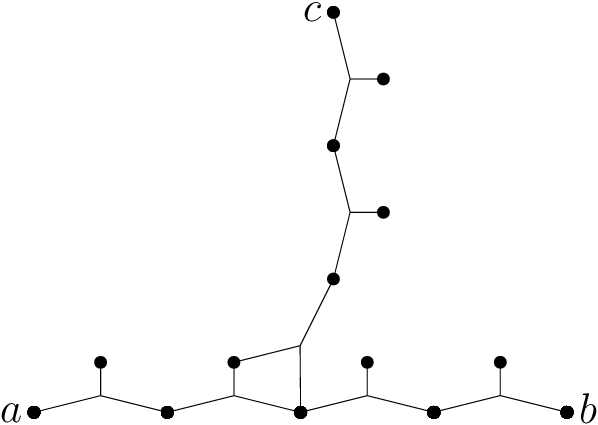}
	\caption{There are no $ac$-chains, because of the edge intersection of size 2.}
	\label{NonTransitive}
\end{figure}

\noindent Let us first consider the union of an $ab$-chain $P$ of positive length and an edge $e^*$ such that $e^* \cap V(P) \neq \varnothing$ and there exists $u \in e^* \setminus V(P)$. When is it possible to prolong a subchain of $P$ with the edge $e^*$ to get an $au$-chain and/or a $bu$-chain?
\\If $|e^* \cap V(P)|=1$, then we get both an $au$-chain and a $bu$-chain, represented by the walks $\ora{aPb}\vert_{e^*} \oplus (e^*,u)$ and $\ola{aPb}\vert_{e^*} \oplus (e^*,u)$ respectively, as illustrated in Table \ref{table1}.

\begin{table}[h] \centering \begin{tabular}{|ll|}
\hline
\multirow{8}{*}{\includegraphics[scale=.58]{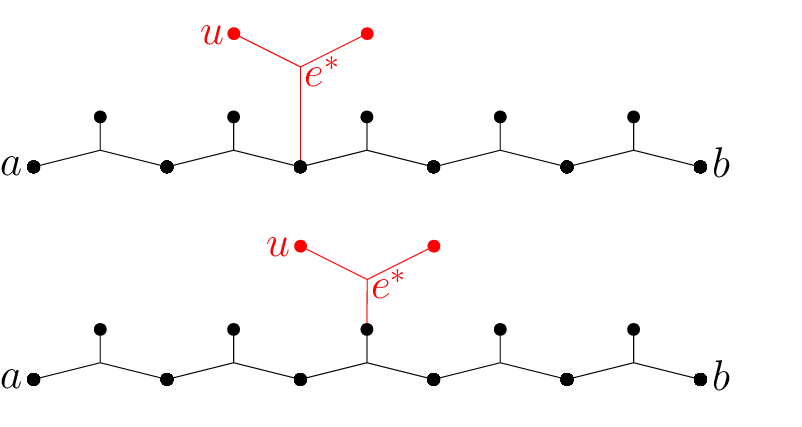}} & \\ & \\ & \\
& $\triangleright$ $au$-chain $\HH{\ora{aPb}\vert_{e^*} \oplus (e^*,u)}$ \\
& $\triangleright$ $bu$-chain $\HH{\ola{aPb}\vert_{e^*} \oplus (e^*,u)}$ \\
& \\ & \\ & \\
\hline
\end{tabular} \caption{An edge $e^*$ intersecting an $ab$-chain $P$ on one vertex, depending on whether that vertex is an inner vertex of $P$ (top) or not (bottom).}\label{table1}
\end{table}

\noindent If $|e^* \cap V(P)|=2$ though, then the walk $\ora{aPb}\vert_{e^*} \oplus (e^*,u)$ does not necessarily represent an $au$-chain (same for $b$). If $a \in e^*$, i.e., $\ora{aPb}\vert_{e^*}=(a)$, then it obviously does. But if $a \not\in e^*$, i.e., $\ora{aPb}\vert_{e^*}$ represents a chain of positive length, then it does if and only if $|e^* \cap \End(\ora{aPb}\vert_{e^*})|=1$. We see a key notion appearing here:

\begin{notation}\label{not:perp}
	Let $P$ be an $ab$-chain of positive length and let $e^*$ be an edge. Write $\ora{aPb}=(a,e_1,\ldots,e_L,b)$. The notation $e^* \perp \ora{aPb}$ (or $e^* \perp \ora{aP}$ equivalently) means that either $e_1 \setminus \{a\} \subseteq e^*$ or $e_i \setminus e_{i-1} \subseteq e^*$ for some $2 \leq i \leq L$. See Figure \ref{Perpendicular}.
\end{notation}

\begin{remarque}
	Note that it is technically possible to have both $e^* \perp \ora{aPb}$ and $e^* \perp \ola{aPb}$. This is the case if, for some $j$, we have $e^*=e_j$ or $e^*=\{o_j,o_{j+1}\} \cup (e_j \cap e_{j+1})$ where $o_i$ denotes the only vertex in $e_i \setminus (\{a,b\} \cup \inn(P))$. However, this will never happen for us, as in practice we will always have either $e^* \not\subseteq V(P)$ or $a \in e^*$.
\end{remarque}

\begin{figure}[h]
	\centering
	\includegraphics[scale=.58]{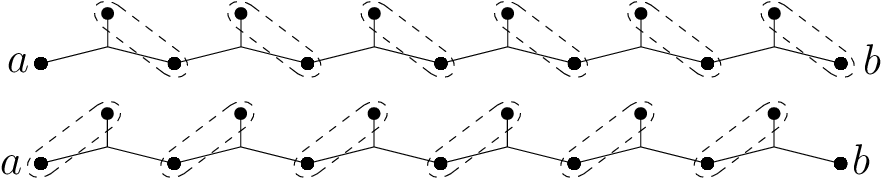}
	\caption{We have $e^* \perp \ora{aPb}$ (resp. $e^* \perp \ola{aPb}$) if and only if $e^*$ contains one of the pairs of vertices highlighted at the top (resp. at the bottom).}
	\label{Perpendicular}
\end{figure}

\noindent In the case at hand $|e^* \cap V(P)|=2$, we can see that $e^* \perp \ora{aPb}$ if and only if $a \not\in e^*$ and $|e^* \cap \End(\ora{aPb}\vert_{e^*})|=2$. Therefore, the walk $\ora{aPb}\vert_{e^*} \oplus (e^*,u)$ represents an $au$-chain if and only if $e^* \not\perp \ora{aPb}$, and similarly the walk $\ola{aPb}\vert_{e^*} \oplus (e^*,u)$ represents a $bu$-chain if and only if $e^* \not\perp \ola{aPb}$. All of this is illustrated in Table \ref{table2}: note that, if no $au$-chain (resp. no $bu$-chain) appears, then we get a $b$-tadpole (resp. an $a$-tadpole).

\begin{table}[h] \begin{tabularx}{\textwidth}{|c|X|X|}
\cline{2-3} \multicolumn{1}{c|}{} & \begin{center} $e^* \not\perp \ola{aPb}$ \end{center} & \begin{center}$e^* \perp \ola{aPb}$\end{center}\\\hline
\multirow{3}{*}{$e^* \not\perp \ora{aPb}$} & \includegraphics[scale=.58]{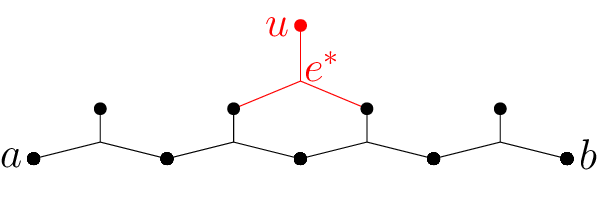} & \includegraphics[scale=.58]{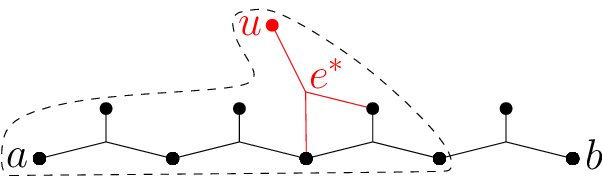} \\
& $\triangleright$ $au$-chain $\HH{\ora{aPb}\vert_{e^*} \oplus (e^*,u)}$ & $\triangleright$ $au$-chain $\HH{\ora{aPb}\vert_{e^*} \oplus (e^*,u)}$ \\
& $\triangleright$ $bu$-chain $\HH{\ola{aPb}\vert_{e^*} \oplus (e^*,u)}$\vphantom{$\displaystyle\int$} & $\triangleright$ $a$-tadpole $\HH{\ora{aPb}\vert_{e^*} \oplus (e^*,\End(\ola{aPb}\vert_{e^*}))}$\vphantom{$\displaystyle\int$} \\\hline
\multirow{3}{*}{$e^* \perp \ora{aPb}$} & \includegraphics[scale=.58]{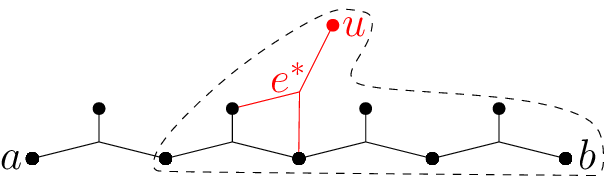} & \hfil\multirow{3}{*}{\hfil impossible} \\
& $\triangleright$ $bu$-chain $\HH{\ola{aPb}\vert_{e^*} \oplus (e^*,u)}$ & \\
& $\triangleright$ $b$-tadpole $\HH{\ola{aPb}\vert_{e^*} \oplus (e^*,\End(\ora{aPb}\vert_{e^*}))}$\vphantom{$\displaystyle\int$} & \\\hline
\end{tabularx} \caption{An edge $e^*$ intersecting an $ab$-chain $P$ on two vertices: all cases. The $a$-tadpole or $b$-tadpole, when one appears, is highlighted.}\label{table2}
\end{table}

\noindent Let us now consider the union of an $ab$-chain $P$ of positive length and some edge $e^* \neq \Start(\ora{aPb})$ that intersects $P$ on at least two vertices including $a$: do we get an $a$-cycle? If $|e^* \cap V(P)|=2$, then the answer is yes, as illustrated in Table \ref{table3}. If $|e^* \cap V(P)|=3$, then Table \ref{table4} shows that it is possible that no $a$-cycle appears, in which case we get a $b$-tadpole.

\begin{table}[h] \centering \begin{tabular}{|ll|}
\hline
\multirow{5}{*}{\includegraphics[scale=.58]{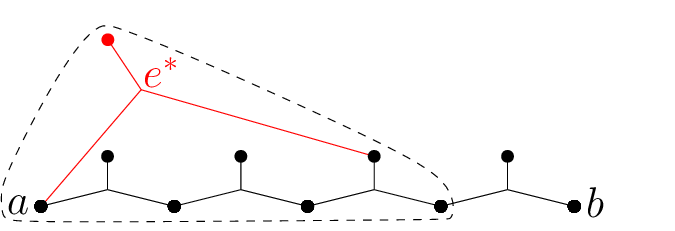}} & \\
& \\
& $\triangleright$ $a$-cycle $\HH{\ora{aPb}\vert_{e^*\setminus \{a\}} \oplus (e^*,a)}$ \\
& \\
& \\
\hline
\end{tabular} \caption{An edge $e^*$ intersecting an $ab$-chain $P$ on two vertices including $a$. The $a$-cycle is highlighted.}\label{table3}
\end{table}
\begin{table}[h] \begin{tabularx}{\textwidth}{|c|X|X|}
\cline{2-3} \multicolumn{1}{c|}{} & \begin{center}$e^* \not\perp \ola{aPb}$\end{center} & \begin{center}$e^* \perp \ola{aPb}$\end{center}\\\hline
\multirow{2}{*}{$e^* \not\perp \ora{aPb}$} & \includegraphics[scale=.58]{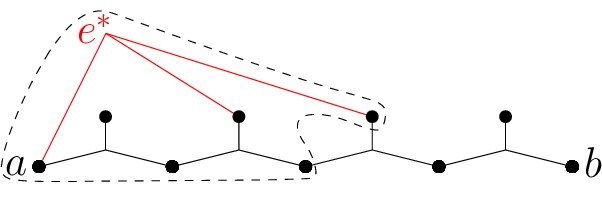} & \includegraphics[scale=.58]{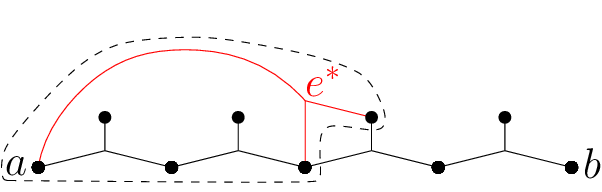} \\
& $\triangleright$ $a$-cycle $\HH{\ora{aPb}\vert_{e^*\setminus \{a\}} \oplus (e^*,a)}$\vphantom{$\displaystyle\int$} & $\triangleright$ $a$-cycle $\HH{\ora{aPb}\vert_{e^*\setminus \{a\}} \oplus (e^*,a)}$\vphantom{$\displaystyle\int$} \\\hline
\multirow{2}{*}{$e^* \perp \ora{aPb}$} & \includegraphics[scale=.58]{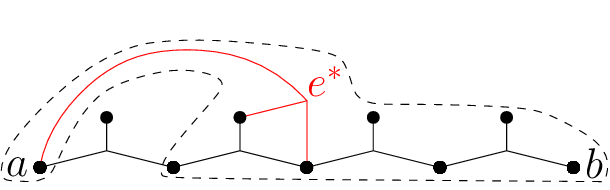} & \hfil\multirow{2}{*}{\hfil impossible} \\
& $\triangleright$ $b$-tadpole $\HH{\ola{aPb}\vert_{e^*} \oplus (e^*,\End(\ora{aPb}\vert_{e^* \setminus \{a\}}))}$\vphantom{$\displaystyle\int$} & \\\hline
\end{tabularx} \caption{An edge $e^*$ intersecting an $ab$-chain $P$ on three vertices including $a$ (and $e^* \neq \Start(\ora{aPb})$): all cases. The $a$-cycle or $b$-tadpole is highlighted.}\label{table4}
\end{table}

\noindent Using these tables, we get the following four union lemmas, which are fundamental in our structural study of 3-uniform hypergraphs. They give us some basic information about the union of two chains or the union of a chain and a tadpole.

\begin{lemme}\label{Lemma1}
	Let $a,b,c$ be distinct vertices. Let $P_{ab}$ be an $ab$-chain, and let $P_c$ be a $c$-chain such that $c \not\in V(P_{ab})$ and $V(P_c) \cap V(P_{ab}) \neq \varnothing$. In particular, $e^* = \End(\ora{c\proj{V(P_{ab})}{c}{P_c}})$ is well defined. Suppose there are no $ca$-chains in $P_{ab} \cup P_c$. Then $|e^* \cap V(P_{ab})|=2$ and $e^* \perp \ora{aPb}$, moreover there is a $cb$-chain in $P_{ab} \cup P_c$ and a $b$-tadpole in $P_{ab} \cup e^* \subseteq P_{ab} \cup P_c$. See Figure \ref{Lemma_PathPath}.
\end{lemme}

\begin{proof}
	By definition of a projection, we have $|e^* \cap V(P_{ab})|\in \{1,2\}$. Let $u \in e^* \setminus V(P_{ab})$. All ways that $e^*$ might intersect $P_{ab}$ are summarized in Tables \ref{table1} and \ref{table2}. There are no $au$-chains $P_{au}$ in $P_{ab} \cup e^*$, otherwise the walk $\ora{cP_c}\vert_{\{u\}} \oplus \ora{uP_{au}a}$ would represent a $ca$-chain in $P_{ab} \cup P_c$, contradicting the assumption of the lemma. Therefore, we are necessarily in the bottom-left case of Table \ref{table2}, which means that: $|e^* \cap V(P_{ab})|=2$, $e^* \perp \ora{aPb}$, there is a $b$-tadpole in $P_{ab} \cup e^*$, and there is a $bu$-chain $P_{bu}$ in $P_{ab} \cup e^*$. The walk $\ora{cP_c}\vert_{\{u\}} \oplus \ora{uP_{bu}b}$ represents a $cb$-chain in $P_{ab} \cup P_c$.
\end{proof}

\begin{figure}[h]
	\centering
	\includegraphics[scale=.58]{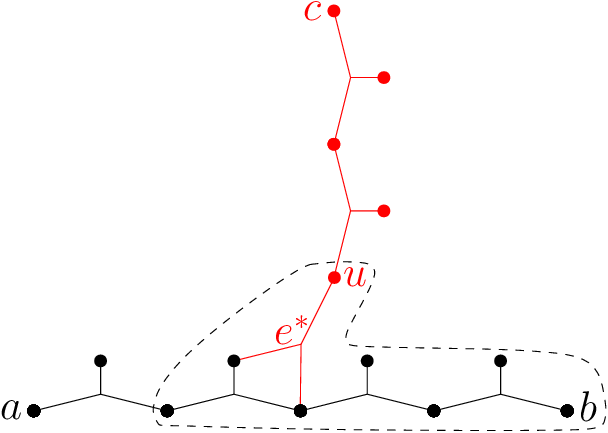}
	\caption{Illustration of Union Lemma \ref{Lemma1}. The represented chains are $P_{ab}$ and $\proj{V(P_{ab})}{c}{P_c}$. The $b$-tadpole is highlighted.}
	\label{Lemma_PathPath}
\end{figure}

\begin{lemme}\label{Lemma4}
Let $a,b,c$ be distinct vertices, where $b$ is marked. Let $S_{ab}$ be an $ab$-snake, and let $P_c$ be a $c$-chain such that $c \not\in V(S_{ab})$ and $V(P_c) \cap V(S_{ab}) \neq \varnothing$.
	\begin{itemize}[noitemsep,nolistsep]
		\item Suppose there are no $c$-snakes in $S_{ab} \cup P_c$. Then there is both a $ca$-chain and an $a$-tadpole in $S_{ab} \cup P_c$.
		\item Suppose there are no $ca$-chains in $S_{ab} \cup P_c$. Then there is both a $cb$-snake and a $b$-tadpole in $S_{ab} \cup P_c$.
	\end{itemize}
	See Figure \ref{Lemma_PathSnake}.
\end{lemme}

\begin{proof}
	The second item is exactly Union Lemma \ref{Lemma1}. The first item is Union Lemma \ref{Lemma1} where the roles of $a$ and $b$ are reversed.
\end{proof}

\begin{figure}[h]
	\centering
	\includegraphics[scale=.58]{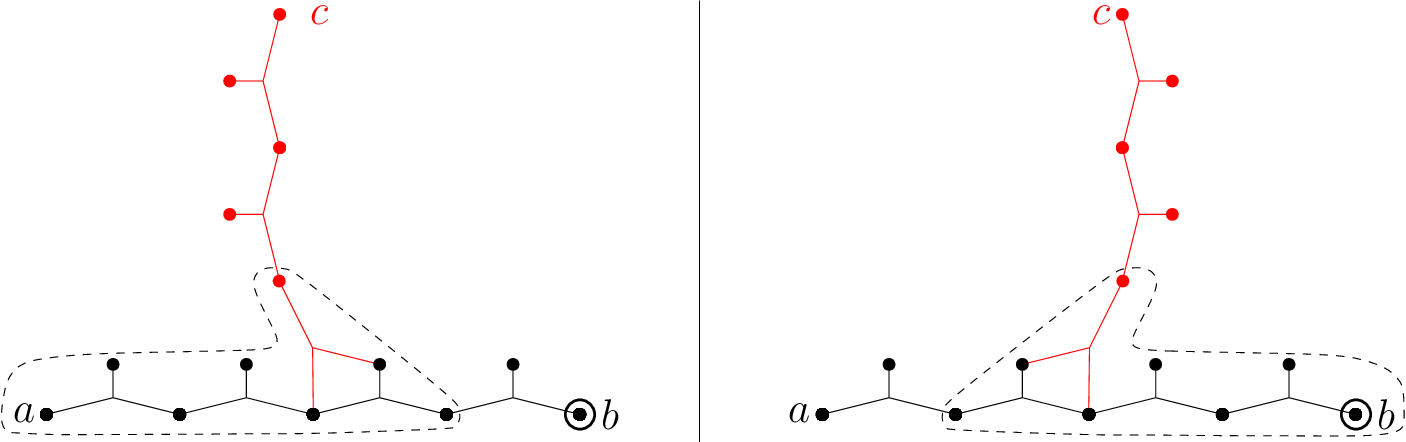}
	\caption{Illustration of Union Lemma \ref{Lemma4} (first item on the left, second item on the right). The represented chains are $S_{ab}$ and $\proj{V(S_{ab})}{c}{P_c}$.}
	\label{Lemma_PathSnake}
\end{figure}

\begin{lemme}\label{Lemma2}
	Let $a,b$ be distinct vertices. Let $P_{ab}$ be an $ab$-chain, and let $P_a$ be an $a$-chain such that $\Start(\ora{aP_a})\neq\Start(\ora{aP_{ab}b})$ and $V(P_a) \cap (V(P_{ab})\setminus\{a\}) \neq \varnothing$. In particular, $e^* = \End(\ora{a\proj{V(P_{ab}) \setminus \{a\}}{a}{P_a}})$ is well defined. Suppose there are no $a$-cycles in $P_{ab} \cup P_a$. Then $e^* \perp \ora{aPb}$ and there is a $b$-tadpole in $P_{ab} \cup e^* \subseteq P_{ab} \cup P_a$. See Figure \ref{Lemma_PathCycle}.
\end{lemme}

\begin{proof}
We distinguish between two cases:
	\begin{itemize}[wide,noitemsep,nolistsep]
		\item First suppose $a \in e^*$. Since $e^* \neq \Start(\ora{aP_{ab}b})$, all ways that $e^*$ might intersect $P_{ab}$ are summarized in Tables \ref{table3} and \ref{table4}. Since there are no $a$-cycles in $P_{ab} \cup P_a \supseteq P_{ab} \cup e^*$ by assumption, we are necessarily in the bottom-left case of Table \ref{table4}, so $e^* \perp \ora{aPb}$ and there is a $b$-tadpole in $P_{ab} \cup e^*$.
		\item Now suppose $a \not\in e^*$, meaning the projection $\proj{V(P_{ab}) \setminus \{a\}}{a}{P_a}$ is of length at least 2. Write $\Start(\ora{aP_a})=\{a,c,c'\}$ where $c \in \inn(P_a)$, as in Figure \ref{Lemma_PathCycle}. Define the $c$-chain $P_c = P_a^{-a-c'}$: we have $c \not\in V(P_{ab})$ and $e^* = \End(\ora{c\proj{V(P_{ab})}{c}{P_c}})$, so the idea is to apply Union Lemma \ref{Lemma1} to $P_{ab}$ and $P_c$. If there was a $ca$-chain $P_{ca}$ in $P_{ab} \cup P_c$, then $(a,\Start(\ora{aP_a}),c) \oplus \ora{cP_{ca}a}$ would represent an $a$-cycle in $P_{ab} \cup P_a$, contradicting the assumption of the lemma. Therefore, there are no $ca$-chains in $P_{ab} \cup P_c$, so Union Lemma \ref{Lemma1} ensures that $e^* \perp \ora{aPb}$ and that there is a $b$-tadpole in $P_{ab} \cup e^*$. \qedhere
	\end{itemize}
\end{proof}

\begin{figure}[h]
	\centering
	\includegraphics[scale=.58]{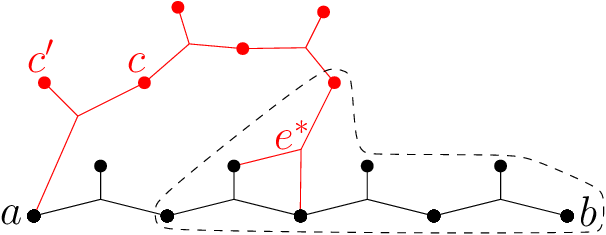}
	\caption{Illustration of Union Lemma \ref{Lemma2}. The represented chains are $P_{ab}$ and $\proj{V(P_{ab})\setminus\{a\}}{a}{P_a}$. The $b$-tadpole is highlighted.}
	\label{Lemma_PathCycle}
\end{figure}

\begin{lemme}\label{Lemma3}
	Let $a,c$ be distinct vertices. Let $T$ be an $a$-tadpole, and let $P_c$ be a $c$-chain such that $c \not\in V(T)$ and $V(P_c) \cap V(T) \neq \varnothing$. In particular, $e^* = \End(\ora{c\proj{V(T)}{c}{P_c}})$ is well defined. Suppose there are no $ca$-chains in $T \cup P_c$. Then $T$ is not a cycle, $|e^* \cap V(T)|=2$ and $e^* \perp \ora{aP_T}$, moreover there is a $c$-tadpole in $T \cup P_c$. See Figure \ref{Lemma_PathTadpole}.
\end{lemme}

\begin{proof}
	Up to replacing $P_c$ by the projection $\proj{V(T)}{c}{P_c}$, assume that $e^*$ is the only edge of $P_c$ intersecting $T$. Let $b$ be the only vertex in $V(P_T) \cap V(C_T)$.
	\begin{claim}\label{claim_Lemma3}
		$e^* \cap (V(P_T)\setminus \{a\}) \neq \varnothing$.
	\end{claim}
	\begin{proofclaim}[Proof of Claim \ref{claim_Lemma3}]
		We already know that $a \not\in e^*$, otherwise $P_c$ would be a $ca$-chain, contradicting the assumption of the lemma. Therefore we must show that $e^* \cap V(P_T) \neq \varnothing$. Suppose for a contradiction that $e^* \cap V(P_T) = \varnothing$. There are two possibilities:
		\begin{itemize}[wide,noitemsep,nolistsep]
			\item Suppose $|e^* \cap V(C_T)|=1$, and write $e^* \cap V(C_T)=\{v\}$. Note that $P_c$ is a $cv$-chain. By Substructure Lemma \ref{lemma_subchain3bis} (with $u=b$), there exists a $bv$-chain $P_{bv}$ in $C_T$. The walk $\ora{cP_cv} \oplus \ora{vP_{bv}b} \oplus \ora{bP_Ta}$ represents a $ca$-chain in $T \cup P_c$, contradicting the assumption of the lemma.
			\item Suppose $|e^* \cap V(C_T)|=2$, and write $e^* \cap V(C_T)=\{v,w\}$. Note that $P_c$ is both a $cv$-chain and a $cw$-chain. Up to swapping the roles of $v$ and $w$, we can assume that $w \in \out(C_T)$ or $v \in \inn(C_T)$. Since $b \in \inn(C_T)$ and $b \neq w$ (indeed $b \in V(P_T)$ whereas $e^* \cap V(P_T)=\varnothing$), Substructure Lemma \ref{lemma_subchain3} (with $u=b$) thus ensures that there exists a $bv$-chain $P_{bv}$ in $C_T$ that does not contain $w$. The fact that $w \not\in V(P_{bv})$ implies that $V(P_c) \cap V(P_{bv})=\{v\}$. The walk $\ora{cP_cv} \oplus \ora{vP_{bv}b} \oplus \ora{bP_Ta}$ thus represents a $ca$-chain in $T \cup P_c$, contradicting the assumption of the lemma. \qedhere
		\end{itemize}
	\end{proofclaim}
	\noindent Claim \ref{claim_Lemma3} implies that $V(P_T)\setminus \{a\} \neq \varnothing$, i.e., $P_T$ is of positive length, i.e., $T$ is not a cycle. It also implies that $V(P_c) \cap V(P_T) \neq \varnothing$, so we can apply Union Lemma \ref{Lemma1} with $P_c$ and the $ab$-chain $P_{ab}=P_T$. Since there are no $ca$-chains in $T \cup P_c \supseteq P_T \cup P_c$ by assumption, Union Lemma \ref{Lemma1} tells us that: $|e^* \cap V(P_T)|=2$, $e^* \perp \ora{aP_Tb}$, and there is a $cb$-chain $P_{cb}$ in $P_T \cup P_c$. Since $|e^* \cap V(P_T)|=2$, we have $e^* \cap (V(C_T) \setminus \{b\})=\varnothing$, hence $V(P_{cb}) \cap V(C_T)=\{b\}$. Therefore $P_{cb} \cup C_T$ is a $c$-tadpole in $T \cup P_c$, which concludes.
\end{proof}

\begin{figure}[h]
	\centering
	\includegraphics[scale=.58]{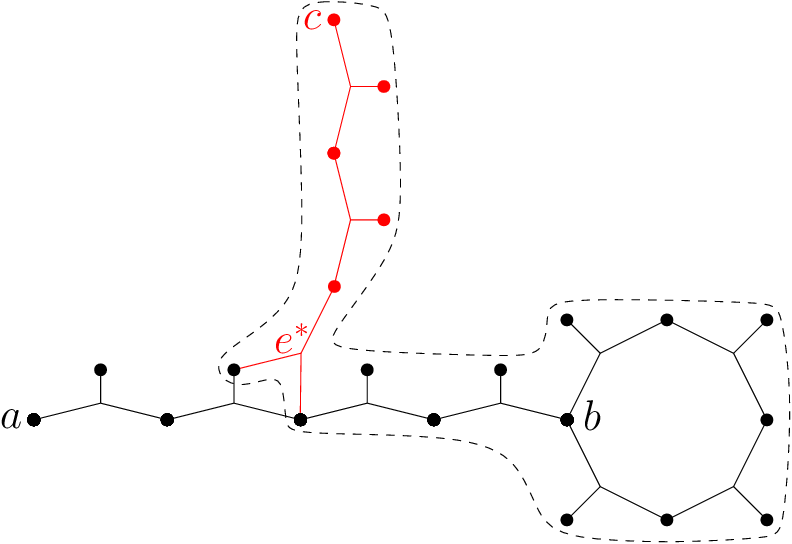}
	\caption{Illustration of Union Lemma \ref{Lemma3}. The represented objects are $T$ and $\proj{V(T)}{c}{P_c}$. The $c$-tadpole is highlighted.}
	\label{Lemma_PathTadpole}
\end{figure}

\subsubsection{Two useful results about $\D_1$-dangers}

\noindent The family $\D_1$ only features dangers that have the minimum number of marked vertices (one for snakes, none for the tadpoles). Indeed, even though snakes and tadpoles with extra marked vertices are also dangers by Proposition \ref{prop_extra} (Marking Monotonicity), the following result shows that it would be redundant to include them, hence why we chose not to. This will slightly simplify matters later in the paper.

\begin{proposition}\label{prop_snakeinside}
	Let $H$ be a marked hypergraph that is not a trivial Maker win, and let $u \in V(H) \setminus M(H)$. Then any $u$-snake or $u$-tadpole in $H$ contains a $\D_1$-danger at $u$. In particular, if $u' \in \I{H^{+u}}{u\D_1(H)}$, then any $u$-snake or $u$-tadpole in $H$ contains $u'$. 
\end{proposition}

\begin{proof}
	Let $X$ be a $u$-snake or $u$-tadpole in $H$. If $M(X)=\varnothing$, then $X$ is necessarily a $u$-tadpole, and $X \in u\D_1(H)$. Therefore, assume that $M(X) \neq \varnothing$, so that the $u$-snake $S = \proj{M(X)}{u}{X} \subseteq X$ is well defined. Note that we may have $S \neq X$ even in the case where $X$ is a snake, because the definition of a snake allows for extra marked vertices. By definition of a projection, the only edge of $S$ that intersects $M(X)$ is $\End(\ora{uS})$. Moreover, since $H$ is not a trivial Maker win, that edge contains exactly one marked vertex, hence $|M(S)|=1$, i.e., $S \in u\D_1(H)$.
	\\ The final assertion of this proposition ensues immediately: all $u$-snakes and $u$-tadpoles in $H$ contain some $D \in u\D_1(H)$, and $u' \in V(D)$ since $u' \in \I{H^{+u}}{u\D_1(H)}$.
\end{proof}

\noindent Finally, the following result will be used numerous times.

\begin{proposition}\label{prop_markedvertex}
	Let $H$ be a marked hypergraph that is not a trivial Maker win, with $|V(H) \setminus M(H)| \geq 2$. If $J(\D_1,H)$ holds, then, for any $m \in M(H)$, there are no $m$-snakes and no $m$-tadpoles in $H$.
\end{proposition}

\begin{proof}
    We prove the contrapositive. Suppose that there exists a subhypergraph $X$ of $H$ that is an $m$-snake or an $m$-tadpole for some $m \in M(H)$.
    \begin{itemize}[leftmargin=*,noitemsep,nolistsep]
        \item Case 1: $M(X)=\{m\}$. In particular, $X$ is necessarily a tadpole.
            \begin{itemize}[leftmargin=*]
                \item[--] First suppose that $X=C$ is a cycle. Let $x \in \inn(C) \setminus \{m\}$: it is clear that $C$ is the union of two $xm$-snakes $S$ and $S'$ such that $M(S)=M(S')=\{m\}$ and $V(S) \cap V(S')=\{x,m\}$ (the two "halves" of the cycle). Therefore, $\{S,S'\}$ is a $\D_0$-fork at $x$ in $C \subseteq H$, so $J(\D_0,H)$ does not hold and neither does $J(\D_1,H)$.
                \item[--] Now, suppose that $X=T$ is a tadpole which is not a cycle. Let $x$ be the only vertex in $V(P_T) \cap V(C_T)$. Since $P_T$ is an $xm$-snake with $M(P_T)=\{m\}$, and $C_T$ is an $x$-cycle with $M(C_T)=\varnothing$, we see that $\{P_T,C_T\}$ is a $\D_1$-fork at $x$ in $T \subseteq H$, so $J(\D_1,H)$ does not hold.
            \end{itemize}
        \item Case 2: $M(X) \setminus \{m\} \neq \varnothing$.  The $m$-snake $N = \proj{M(X) \setminus \{m\}}{m}{X} \subseteq X$ is well defined. By definition of a projection, the only edge of $N$ that intersects $M(X) \setminus \{m\}$ is $\End(\ora{mN})$. Moreover, since $H$ is not a trivial Maker win, that edge contains exactly one marked vertex (call it $m'$), so $N$ is an $mm'$-nunchaku of length at least 2. Let $x \in \inn(N)$: it is clear that $N$ is the union of an $xm$-snake $S$ and an $xm'$-snake $S'$ such that $M(S)=\{m\}$, $M(S')=\{m'\}$ and $V(S) \cap V(S')=\{x\}$ (the two "halves" of the nunchaku). Therefore, $\{S,S'\}$ is a $\D_0$-fork at $x$ in $N \subseteq H$, so $J(\D_0,H)$ does not hold and neither does $J(\D_1,H)$. \qedhere
    \end{itemize}
\end{proof}

\subsection{Structural properties of the $\widehat{\D_1}$-dangers}

\noindent Now that we have addressed chains and tadpoles, we move on to $\widehat{\D_1}$-dangers, which are special unions of such objects.

\subsubsection{First properties}

\noindent Recall that a $\widehat{\D_1}$-danger $D$ at $x$ is basically the union of a "potential $\D_1$-fork" $\F_D=\S_z \cup \T_z \cup \P_{zx}$ at some vertex $z$, which will become an actual $\D_1$-fork at $z$ if Maker picks $x$. The elements of $\S_z \cup \T_z$ are $\D_1$-dangers at $z$ already, while the elements of $\P_{zx}$ would be one if $x$ was marked ($zx$-chains that become $zx$-snakes). Figure \ref{Example_Dangers} features some examples.

\begin{figure}[h]
    \centering
    \includegraphics[scale=.55]{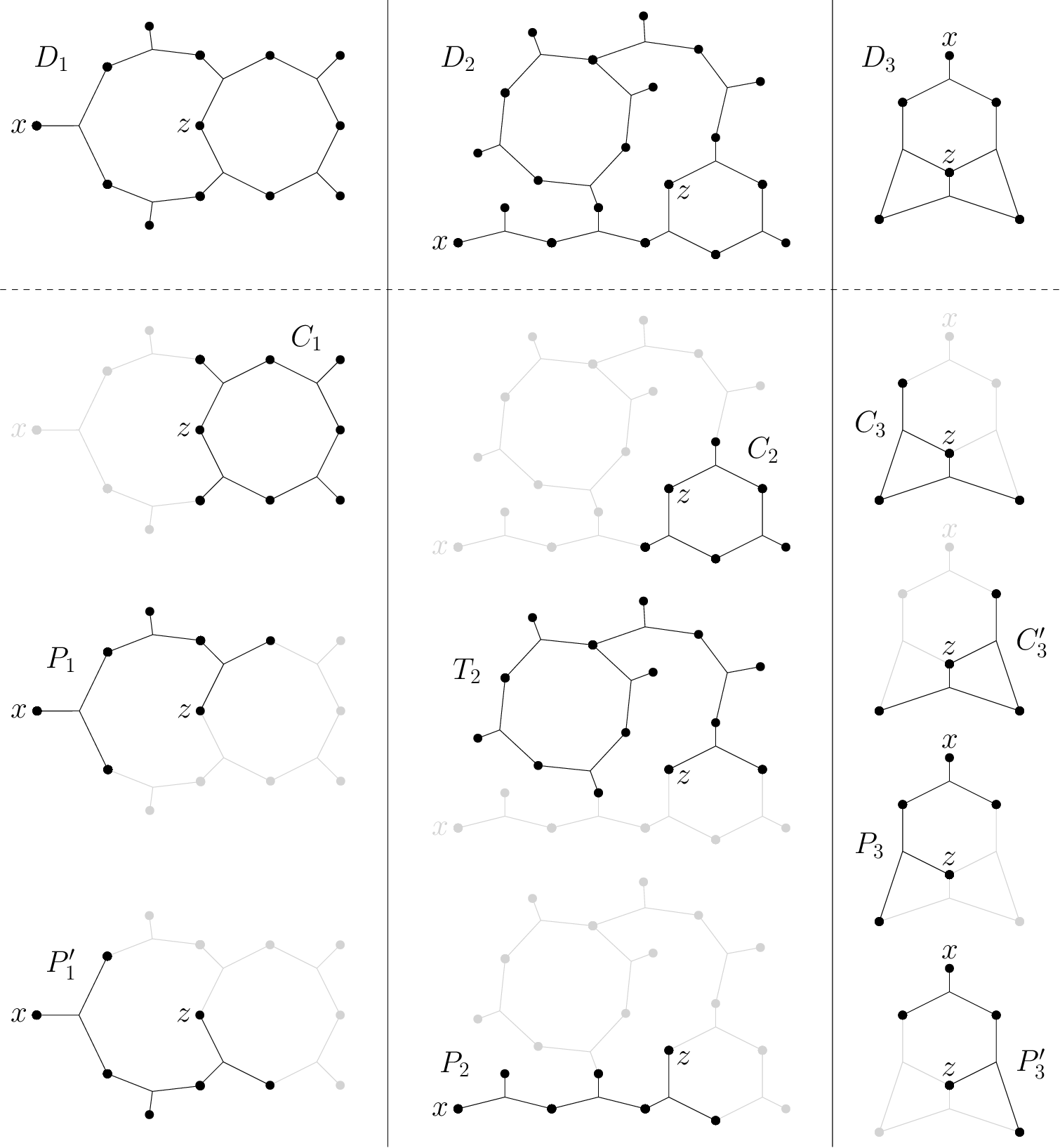}
    \caption{Three examples of a $\widehat{\D_1}$-danger at $x$, with a possible decomposition $(z,\S_z \cup \T_z \cup \P_{zx})$ below. For $D_1$: $\S_z=\varnothing$, $\T_z=\{C_1\}$, $P_{zx}=\{P_1,P'_1\}$. For $D_2$ (seen before in Figure \ref{MakerWin5}): $\S_z=\varnothing$, $\T_z=\{C_2,T_2\}$, $P_{zx}=\{P_2\}$. For $D_3$: $\S_z=\varnothing$, $\T_z=\{C_3,C'_3\}$, $P_{zx}=\{P_3,P'_3\}$.}\label{Example_Dangers}
\end{figure}

\noindent For fixed $(D,x) \in \widehat{\D_1}$ and fixed $z$, there are two natural options for a preferred decomposition. One would be to choose a minimal $\F_D$, so that every element is essential to the fact that $\F_D^{+x}$ is a fork at $z$ (this is the choice we have made in Figure \ref{Example_Dangers}, as it makes it easier to understand the family $\widehat{\D_1}$). The other would be to choose a maximal $\F_D$, so that we have all relevant objects that exist inside $D$ at our disposal, even though some might be redundant for the intersection (for example, $D_2$ from Figure \ref{Example_Dangers} contains two $zx$-chains going through the tadpole $T_2$, additionally to the $zx$-chain $P_2$). The latter is best suited for our structural studies, so we will always make the choice of maximality hence the following definition.

\begin{definition}\label{def:maximal_decomposition}
        Let $(D,x) \in \widehat{\D_1}$. A decomposition $(z , \F_D = \S_z \cup \T_z \cup \P_{zx})$ of $(D,x)$ is said to be \textit{maximal} if $\S_z$, $\T_z$ and $\P_{zx}$ are maximal collections with respect to set inclusion. In other words, this means that:
	\begin{itemize}[nolistsep,noitemsep]
		\item $\S_z$ is the collection of all $z$-snakes in $D$ such that $|M(S)|=1$ and $x \not\in V(S)$;
		\item $\T_z$ is the collection of all $z$-tadpoles in $D$ such that $M(T)=\varnothing$ and $x \not\in V(T)$;
		\item $\P_{zx}$ is the collection of all $zx$-chains in $D$ such that $M(P)=\varnothing$.
	\end{itemize}
\end{definition}

\begin{remarque}
	Because of the choice of $z$, there might not be a unique maximal decomposition, but we do not mind. Given some $(D,x) \in \widehat{D_1}$, we will always just fix some maximal decomposition, with arbitrary $z$.
\end{remarque}

\noindent Let us start by gathering a few properties of the $\widehat{\D_1}$-dangers which we will use often.

\begin{proposition}\label{prop_structure}
	Let $(D,x) \in \widehat{\D_1}$, with a maximal decomposition $(z,\F_D=\S_z \cup \T_z \cup \P_{zx})$. We have the following properties:
	\begin{enumerate}[noitemsep,nolistsep,label={\textup{(\alph*)}}]
		\item For all $u \in V(D) \setminus (M(D) \cup \{x,z\})$, there exists $X \in \F_D$ such that $u \not\in V(X)$. \label{item2}
            \item $D$ is not a trivial Maker win. \label{item5}
		\item There are no $x$-snakes and no $x$-tadpoles in $D$. \label{item6}
		\item $\P_{zx} \neq \varnothing$. \label{item4}
		\item $\S_z \cup \T_z \neq \varnothing$, i.e., there exists a $\D_1$-danger at $z$ in $D$ that does not contain $x$. \label{item3}
		\item There exists a $z$-cycle in $D$. \label{item7}
	\end{enumerate}
\end{proposition}

\begin{proof}

	Let us prove properties (a) through (f) in that order.
	
	\begin{enumerate}[label={(\alph*)},leftmargin=22pt]
	
		\item By definition of a decomposition, $\F_D^{+x}$ is a $\D_1$-fork at $z$ in $D^{+x}$, i.e., $\I{D^{+x}}{\F_D^{+x}}=\{z\}$. This implies that $\I{D}{\F_D} \subseteq \{z,x\}$: in other words, apart from $z$ and maybe $x$, no non-marked vertex of $D$ is contained in all elements of $\F_D$.

        \item By definition, each element of $\F_D$ has at most one marked vertex, so $D$ is not a trivial Maker win.

		\item Thanks to property \ref{item5}, we can apply Proposition \ref{prop_snakeinside} with $H=D$ and $u=x$. It ensures that $D$ contains no $x$-snakes and no $x$-tadpoles, as it would otherwise contain a $\D_1$-danger at $x$, contradicting the definition of $\widehat{\D_1}$.

		\item It is impossible that $\P_{zx} = \varnothing$, because we would get $D=\union{\S_z \cup \T_z}$, contradicting the fact that $D$ contains $x$ while the subhypergraphs in the collection $\S_z \cup \T_z$ do not.
	
	\end{enumerate}
		
	 \noindent Using property \ref{item4}, let $P_{zx}$ be a shortest chain in $\P_{zx}$, and define $v = o(x,\ora{xP_{zx}z})$ and $w = o(z,\ola{xP_{zx}z})$. This chain will help us prove properties \ref{item3} and \ref{item7}. Note that $M(P_{zx})=\varnothing$ by definition of the collection $\P_{zx}$.
	
	\begin{enumerate}[label={(\alph*)},leftmargin=22pt]
	
	\setcounter{enumi}{4}
		
		\item Suppose for a contradiction that $\S_z=\T_z=\varnothing$. By maximality of the decomposition, this exactly means that all $\D_1$-dangers at $z$ in $D$ contain $x$. Also note that $D=\union{\P_{zx}}$, so $M(D)=\varnothing$. We are going to use the chain $P_{zx}$. By property \ref{item2}, there exists $P^{\overline{v}} \in \F_D=\P_{zx}$ such that $v \not\in V(P^{\overline{v}})$. We have $\Start(\ora{xP^{\overline{v}}z}) \neq \Start(\ora{xP_{zx}z}) \ni v$ and $V(P^{\overline{v}}) \cap (V(P_{zx}) \setminus \{x\}) \supseteq \{z\} \neq \varnothing$, so we can apply Union Lemma \ref{Lemma2}: since $P_{zx} \cup P^{\overline{v}} \subseteq D$ contains no $x$-cycles by property \ref{item6}, it contains a $z$-tadpole $T$. If $x \not\in V(T)$ as on the left of Figure \ref{Item3}, then $T$ contradicts the fact that all $\D_1$-dangers at $z$ in $D$ contain $x$. If $x \in V(T)$, then the only possibility is that the projection $\proj{V(P_{zx}) \setminus \{x\}}{x}{P^{\overline{v}}}$ consists of a single edge $e$ as illustrated on the right of Figure \ref{Item3}: since $v \not\in e$, we get a $zx$-chain $P'_{zx}$ that is strictly shorter than $P_{zx}$, also a contradiction.
		
		\begin{figure}[h]
			\centering
			\includegraphics[scale=.55]{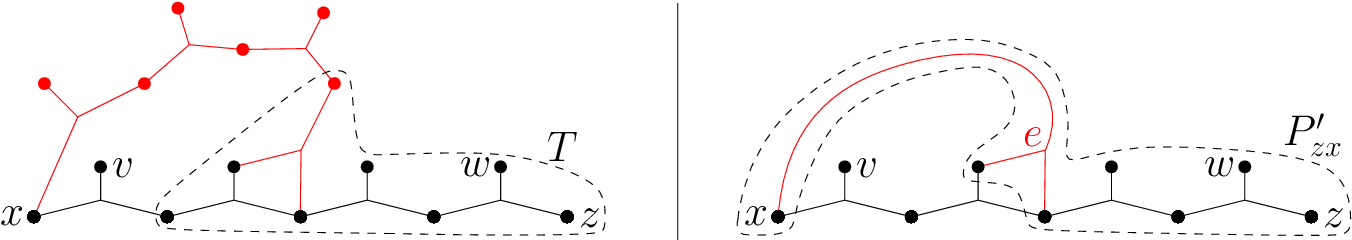}
			\caption{The contradiction that yields property \ref{item3}, if $x \not\in V(T)$ (left) or if $x \in V(T)$ (right). The represented chains are $P_{zx}$ (bottom) and $\proj{V(P_{zx}) \setminus \{x\}}{x}{P^{\overline{v}}}$.}\label{Item3}
		\end{figure}
		
		\item We are going to use the chain $P_{zx}$ again. By property \ref{item2}, there exists $X^{\overline{w}} \in \F_D$ such that $w \not\in V(X^{\overline{w}})$. We claim that $V(X^{\overline{w}}) \cap (V(P_{zx}) \setminus \{z\}) \neq \varnothing$ in all cases, indeed:
			\begin{itemize}[noitemsep,nolistsep]
				\item If $X^{\overline{w}} \in \P_{zx}$, then this is obvious because $x \in V(X^{\overline{w}})$.
				\item If $X^{\overline{w}} \in \T_z$, then this is true because otherwise $P_{zx} \cup X^{\overline{w}}$ would be an $x$-tadpole in $D$, contradicting property \ref{item6}.
				\item If $X^{\overline{w}} \in \S_z$, then this is true because otherwise $P_{zx} \cup X^{\overline{w}}$ would be an $x$-snake in $D$, contradicting property \ref{item6}.
			\end{itemize}
			Therefore, the projection $P = \proj{V(P_{zx}) \setminus \{z\}}{z}{X^{\overline{w}}}$ is well defined. Since $w \not\in V(P)$ and $w \in \Start(\ola{xP_{zx}z})$, we have $\Start(\ora{zP}) \neq \Start(\ola{xP_{zx}z})$ so we can apply Union Lemma \ref{Lemma2}: since $P_{zx} \cup P \subseteq D$ contains no $x$-tadpoles by property \ref{item6}, it contains a $z$-cycle. \qedhere
		
	\end{enumerate}
	
\end{proof}

\noindent The proofs of properties \ref{item3} and \ref{item7} are typical of the methods that we will use extensively. We see that property \ref{item2} is a key existence tool, providing us with subhypergraphs of $D$ which we can use to partially reconstruct $D$ and establish structural properties.

\

\noindent Beyond these basic characteristics, it is difficult to say much about the structure of $\widehat{\D_1}$-dangers in general. However, in practice, we will always consider $\widehat{\D_1}$-dangers inside marked hypergraphs $H$ such that $J(\D_1,H)$ holds. Given some $\widehat{\D_1}$-danger $D$ at $x$ in $H$, with a maximal decomposition $(z,\F_D=\S_z \cup \T_z \cup \P_{zx})$, this implies that $\I{H^{+z}}{z\D_1(H)} \neq \varnothing$. In other words, even though the intersection in $H^{+z}$ of $z\D_1(H) \cup \P_{zx}$ is empty, the intersection in $H^{+z}$ of $z\D_1(H)$ alone is not: it contains some vertex $s$. This vertex $s$ will often be useful.

\begin{proposition}\label{prop_s}
	Let $H$ be a marked hypergraph that is not a trivial Maker win. Let $D$ be a $\widehat{\D_1}$-danger at some $x$ in $H$, with a maximal decomposition $(z,\F_D=\S_z \cup \T_z \cup \P_{zx})$. Suppose $\I{H^{+z}}{z\D_1(H)} \neq \varnothing$. Then, for all $s \in \I{H^{+z}}{z\D_1(H)}$:
	\begin{itemize}[noitemsep,nolistsep]
		\item Any $z$-tadpole or $z$-snake in $H$ contains $s$.
		\item $s \in V(D) \setminus (M(D)\cup\{x,z\})$.
		\item There exists $\Ps \in \P_{zx}$ such that $s \not\in V(\Ps)$. Moreover, the edges $\Start(\ora{x\Ps z})$ and $\End(\ora{x\Ps z})$ are the same for any choice of $\Ps$. 
	\end{itemize}
\end{proposition}

\begin{proof}
	We prove all three assertions separately:
	\begin{itemize}[wide,noitemsep,nolistsep]
		\item Since $s \in \I{H^{+z}}{z\D_1(H)}$ and $H$ is not a trivial Maker win by assumption, Proposition \ref{prop_snakeinside} applies with $u=z$ and $u'=s$, hence the first assertion.
		\item By definition of $\I{H^{+z}}{\cdot}$, we have $s \not\in M(H^{+z})=M(H) \cup \{z\}$. Moreover, we know $x \not\in \I{H^{+z}}{z\D_1(H)}$ by Proposition \ref{prop_structure}\ref{item3}, so $s \neq x$ hence the second assertion.
		\item Since $s \in V(D) \setminus (M(D)\cup\{x,z\})$, there exists some $X^{\overline{s}} \in \F_D$ such that $s \not\in V(X^{\overline{s}})$ by Proposition \ref{prop_structure}\ref{item2}. Since $\S_z \cup \T_z \subseteq z\D_1(H)$, all elements of $\S_z \cup \T_z$ must contain $s$, so necessarily $X^{\overline{s}} \in \P_{zx}$. Finally, let $\Ps_1,\Ps_2 \in \P_{zx}$ be such that $s \not\in V(\Ps_1)$ and $s \not\in V(\Ps_2)$. Suppose for a contradiction that $\Start(\ora{x\Ps_1z}) \neq \Start(\ora{x\Ps_2z})$: by Union Lemma \ref{Lemma2}, $\Ps_1 \cup \Ps_2 \subseteq D$ contains an $x$-cycle (contradicting Proposition \ref{prop_structure}\ref{item6}) or a $z$-tadpole (which does not contain $s$, also a contradiction). Similarly, suppose for a contradiction that $\End(\ora{x\Ps_1 z}) \neq \End(\ora{x\Ps_2 z})$, i.e., $\Start(\ola{x\Ps_1 z}) \neq \Start(\ola{x\Ps_2 z})$: by Union Lemma \ref{Lemma2}, $\Ps_1 \cup \Ps_2 \subseteq D$ contains an $x$-tadpole (contradicting Proposition \ref{prop_structure}\ref{item6}) or a $z$-cycle (which does not contain $s$, also a contradiction). \qedhere
	\end{itemize}
\end{proof}

\noindent We now establish some important properties of the $\widehat{\D_1}$-dangers in an ambient hypergraph $H$ where $\I{H^{+z}}{z\D_1(H)} \neq \varnothing$, or sometimes under the stronger assumption that $J(\D_1,H)$ holds.

\subsubsection{Union lemmas}

The next two lemmas are the analog for $\widehat{\D_1}$-dangers of the union lemmas that we have established for chains and tadpoles. We look at what happens inside the union of a $\widehat{\D_1}$-danger and a chain.

\begin{lemme}\label{Lemma5}

	Let $H$ be a marked hypergraph that is not a trivial Maker win, and let $x \in V(H) \setminus M(H)$. Let $D$ be a  $\widehat{\D_1}$-danger at $x$ in $H$, with a maximal decomposition $(z , \F_D = \S_z \cup \T_z \cup \P_{zx})$. Let $c \in V(H) \setminus V(D)$, and let $P_c$ be a $c$-chain such that $V(P_c) \cap V(D) \neq \varnothing$.
	\begin{enumerate}[noitemsep,nolistsep,label=\textup{(\roman*)}]
		\item If $\I{H^{+z}}{z\D_1(H)} \neq \varnothing$, then there is a $c$-tadpole, a $c$-snake or a $cx$-chain in $D \cup P_c$. \label{item1_Lemma5}
		\item If $J(\D_1,H)$ holds, then there is a $c$-tadpole or a $cx$-chain in $D \cup P_c$. \label{item2_Lemma5}
	\end{enumerate}
\end{lemme}

\begin{proof}
	First of all, we can assume that $P_c$ consists of a single edge $e$. Indeed, suppose that is not the case, and define $e = \End(\ora{c\proj{V(D)}{c}{P_c}})$ and $c'$ as the unique vertex in $(e \cap \inn(P_c)) \setminus V(D)$:
	\begin{itemize}[noitemsep,nolistsep]
		\item If there is a $c'x$-chain $P$ in $D \cup e$, then $\ora{cP_c}\vert_{\{c'\}} \oplus \ora{c'Px}$ represents a $cx$-chain in $D \cup P_c$.
		\item If there is a $c'$-snake $S$ in $D \cup e$, then $\ora{cP_c}\vert_{\{c'\}} \oplus \ora{c'S}$ represents a $c$-snake in $D \cup P_c$.
		\item If there is a $c'$-tadpole $T$ in $D \cup e$, then $\ora{cP_c}\vert_{\{c'\}} \oplus \ora{c'T}$ represents a $c$-tadpole in $D \cup P_c$.
	\end{itemize}
	Therefore, we are working in $D \cup P_c=D \cup e$. Now suppose for a contradiction that:
	\begin{equation}\label{cont1}
		\text{$D \cup e$ contains no $c$-tadpoles, no $cx$-chains, and also no $c$-snakes in the case of item \ref{item1_Lemma5}.}\tag{$\ast$}
	\end{equation}
	Since $\I{H^{+z}}{z\D_1(H)} \neq \varnothing$ by assumption, let $s \in \I{H^{+z}}{z\D_1(H)}$, and let $\Ps \in \P_{zx}$ be such that $s \not\in V(\Ps)$ as per Proposition \ref{prop_s}. Define $w = o(z,\ola{x\Ps z})$. These notations are summed up in Figure \ref{Lemma5-4}.
	
	\begin{figure}[h]
		\centering
		\includegraphics[scale=.55]{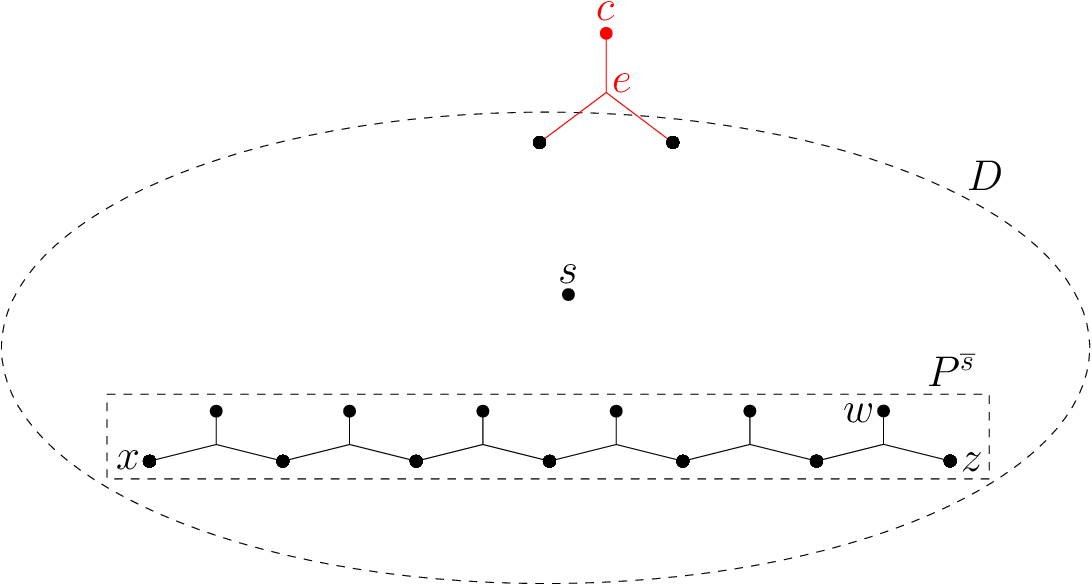}
		\caption{$D$ is only partially represented. In this picture we have $|e \cap V(D)|=2$, but it is also possible that $|e \cap V(D)|=1$.}\label{Lemma5-4}
	\end{figure}
	
	\noindent The key to the proof is the fact that every $z$-tadpole contains $s$ (by definition of $s$), whereas $\Ps$ does not. This will eventually lead to a contradiction. Before going into details, let us explain the idea of the proof. We want to show that there exists a $cz$-chain $P_{cz}^{\overline{w}}$ in $D \cup e$ that does not contain $w$. Indeed, suppose we manage to exhibit one. On the one hand, following $P_{cz}^{\overline{w}}$ starting from $c$ until touching $\Ps$, we know by (\ref{cont1}) that we cannot get a $cx$-chain. Therefore, we necessarily get a $z$-tadpole $T$, containing $s$ at a uniquely determined location, as pictured on the left of Figure \ref{Lemma5-5}. On the other hand, following $P_{cz}^{\overline{w}}$ starting from $z$ until touching $\Ps$ again, we know by Proposition \ref{prop_structure}\ref{item6} that we cannot get an $x$-tadpole. Therefore, we necessarily get a $z$-cycle, which must also contain $s$, as pictured on the right of Figure \ref{Lemma5-5}. All in all, we get two disjoint parts of $P_{cz}^{\overline{w}}$ which both contain $s$, a contradiction. We now proceed with the rigorous proof, in three steps. We prove items \ref{item1_Lemma5} and \ref{item2_Lemma5} jointly: there are only two times during the proof where we will have to differentiate the two very briefly to make separate arguments.
	
	\begin{figure}[h]
		\centering
		\includegraphics[scale=.55]{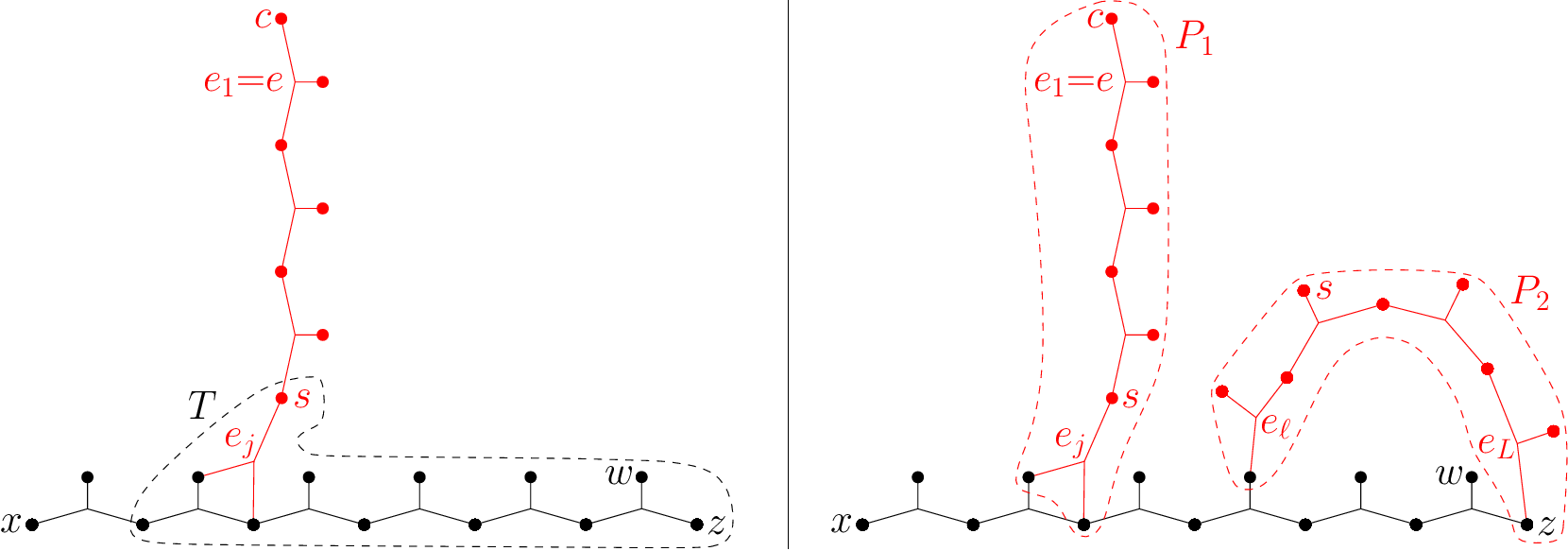}
		\caption{Left: illustration of Claim \ref{claim_Lemma5}. Right: the desired contradiction, with $s$ having two different locations at once.}\label{Lemma5-5}
	\end{figure}
	
	\begin{enumerate}[wide,label=\textbf{\arabic*)}]
	
		\item Firstly: we show that there exists a $cz$-chain $P_{cz}$ in $D \cup e$.
			\\ Since $e \cap V(D) \neq \varnothing$, there exists $X \in \F_D = \S_z \cup \T_z \cup \P_{zx}$ such that $e \cap V(X) \neq \varnothing$. For each of the three cases, we use an adequate union lemma from Section \ref{Section2}:
			\begin{itemize}[noitemsep,nolistsep]
				\item Suppose $X \in \T_z$, and write $X=T$. Since there are no $c$-tadpoles in $D \cup e \supseteq T \cup e$ by (\ref{cont1}), Union Lemma \ref{Lemma3} ensures that there is a $cz$-chain in $T \cup e$.
				\item Suppose $X \in \S_z$, and write $X=S$. We address items \ref{item1_Lemma5} and \ref{item2_Lemma5} separately. For \ref{item1_Lemma5}, there are no $c$-snakes in $D \cup e \supseteq S \cup e$ by (\ref{cont1}). For \ref{item2_Lemma5}, let $m$ be the marked vertex such that $S$ is a $zm$-snake: since $J(\D_1,H)$ holds, Proposition \ref{prop_markedvertex} tells us there are no $m$-tadpoles in $D \cup e \supseteq S \cup e$. In both cases, Union Lemma \ref{Lemma4} ensures that there is a $cz$-chain in $S \cup e$.
				\item Suppose $X \in \P_{zx}$, and write $X=P$. Since there are no $cx$-chains in $D \cup e \supseteq P \cup e$ by (\ref{cont1}), Union Lemma \ref{Lemma1} ensures that there is a $cz$-chain in $P \cup e$.
			\end{itemize}
			In all cases, we get a $cz$-chain $P_{cz}$ in $D \cup e$.

		\item Secondly: we show that there exists a $cz$-chain $P_{cz}^{\overline{w}}$ in $D \cup e$ that does not contain $w$.
			\\ We start with a useful observation about $cz$-chains:
		
			\begin{claim}\label{claim_Lemma5}
				Let $P_0$ be a $cz$-chain in $D \cup e$, and write $\ora{c\proj{V(\Ps)}{c}{P_0}}=(c,e_1,\ldots,e_j)$. Then: $j>1$, $e_j \perp \ora{x\Ps z}$, and $e_{j-1} \cap e_j = \{s\}$. In particular, the $cs$-chain in $P_0$ is disjoint from $\Ps$.
			\end{claim}
	
			\begin{proofclaim}[Proof of Claim \ref{claim_Lemma5}]
				Since there are no $cx$-chains in $D \cup e$ by (\ref{cont1}), we apply Union Lemma \ref{Lemma1} with $a=x$, $b=z$, $P_{ab}=\Ps$. Note that $e_j$ is precisely the edge $e^* = \End(\ora{c\proj{V(\Ps)}{c}{P_0}})$ from the statement of Union Lemma \ref{Lemma1}. We get that: $|e_j \cap V(\Ps)|=2$, $e_j \perp \ora{x\Ps z}$, and there is a $z$-tadpole $T$ in $\Ps \cup e_j$.
				\\ Since $|e_j \cap V(\Ps)|=2$, there is exactly one vertex of $T$ that is not in $\Ps$. That vertex is necessarily $s$, as pictured on the left of Figure \ref{Lemma5-5}: indeed, we know $s \in V(T)$ by definition of $s$, and $s \not\in V(\Ps)$ by definition of $\Ps$. In particular, since $c \neq s$ ($s \in V(D)$ whereas $c \not\in V(D)$), we get $j>1$. We know $(e_1 \cup \ldots \cup e_{j-1}) \cap V(\Ps) = \varnothing$ by definition of a projection, therefore $e_{j-1} \cap e_j = \{s\}$ and $(e_1,\ldots,e_{j-1})$ represents the unique $cs$-chain in $P_0$, which is disjoint from $\Ps$.
			\end{proofclaim}
			
			\noindent Applying Claim \ref{claim_Lemma5} with $P_0=P_{cz}$, we obtain a $cs$-chain $P_{cs}$ in $P_{cz}$ such that $V(P_{cs}) \cap V(\Ps)=\varnothing$: in particular, $w \not\in V(P_{cs})$. Moreover, since $w$ is non-marked (otherwise $\Ps$ would contain an $x$-snake, contradicting Proposition \ref{prop_structure}\ref{item6}), Proposition \ref{prop_structure}\ref{item2} ensures that there exists $X^{\overline{w}} \in \F_D$ such that $w \not\in V(X^{\overline{w}})$. We thus find $P_{cz}^{\overline{w}}$ inside $P_{cs} \cup X^{\overline{w}}$:
			\begin{itemize}[noitemsep,nolistsep]
				\item Suppose $X^{\overline{w}} \in \T_z$, and write $X^{\overline{w}} = T$. In particular $s \in V(T)$, so $V(P_{cs}) \cap V(T) \neq \varnothing$. Since $D \cup e \supseteq P_{cs} \cup T$ does not contain a $c$-tadpole by (\ref{cont1}), Union Lemma \ref{Lemma3} ensures that $P_{cs} \cup T$ contains a $cz$-chain.
				\item Suppose $X^{\overline{w}} \in \S_z$, and write $X^{\overline{w}} = S$. In particular $s \in V(S)$, so $V(P_{cs}) \cap V(S) \neq \varnothing$. For the second and last time in this proof, we address items \ref{item1_Lemma5} and \ref{item2_Lemma5} separately. For \ref{item1_Lemma5}, there are no $c$-snakes in $D \cup e \supseteq P_{cs} \cup S$ by (\ref{cont1}). For \ref{item2_Lemma5}, let $m$ be the marked vertex such that $S$ is a $zm$-snake: Proposition \ref{prop_markedvertex} tells us there are no $m$-tadpoles in $D \cup e \supseteq P_{cs} \cup S$. In both cases, Union Lemma \ref{Lemma4} ensures that there is a $cz$-chain in $P_{cs} \cup S$.
				\item Suppose $X^{\overline{w}} \in \P_{zx}$, and write $X^{\overline{w}} = P$. Since $w \not\in V(P)$, we have $w \not\in \Start(\ola{xPz})$ hence $\Start(\ola{xPz}) \neq \Start(\ola{x\Ps z})$. By the final assertion of Proposition \ref{prop_s}, this implies $s \in V(P)$, so $V(P_{cs}) \cap V(P) \neq \varnothing$. Since $D \cup e \supseteq P_{cs} \cup P$ does not contain a $cx$-chain by (\ref{cont1}), Union Lemma \ref{Lemma1} ensures that $P_{cs} \cup P$ contains a $cz$-chain.
			\end{itemize}
			\noindent In all cases, we get a $cz$-chain $P_{cz}^{\overline{w}}$ in $P_{cs} \cup X^{\overline{w}} \subseteq D \cup e$, that does not contain $w$ since neither $P_{cs}$ nor $X^{\overline{w}}$ does.
			
		\item Finally: we conclude by getting the desired contradiction illustrated on the right of Figure \ref{Lemma5-5}. We now work exclusively inside $P_{cz}^{\overline{w}} \cup \Ps$.
		\\ We start by defining the chains $P_1$ and $P_2$ pictured on the right of Figure \ref{Lemma5-5}. First, define the projection $P_1 = \proj{V(\Ps)}{c}{P_{cz}^{\overline{w}}}$. By (\ref{cont1}), it is impossible that $V(P_{cz}^{\overline{w}}) \cap V(\Ps)=\{z\}$, because $P_{cz}^{\overline{w}} \cup \Ps$ would then be a $cx$-chain. Therefore, the projection $P_2 = \proj{V(\Ps) \setminus\{z\}}{z}{P_{cz}^{\overline{w}}}$ is also well defined. Write $\ora{cP_{cz}^{\overline{w}}z}=(c,e_1,\ldots,e_L,z)$, $\ora{cP_1}=(c,e_1,\ldots,e_j)$, and $\ora{zP_2}=(z,e_L,e_{L-1},\ldots,e_{\ell})$, i.e., $j=\min\{1 \leq i \leq L \mid e_i \cap V(\Ps) \neq \varnothing\}$ and $\ell=\max\{1 \leq i \leq L \mid e_i \cap (V(\Ps)\setminus\{z\}) \neq \varnothing\}$. Note that necessarily $e_1=e$, since $e$ is the only edge incident to $c$.

			\begin{itemize}[noitemsep,nolistsep]
			
				\item First of all, we show that $1<j<\ell$ and that $s \in e_{j-1}$. By Claim \ref{claim_Lemma5} applied with $P_0=P_{cz}^{\overline{w}}$, we have: $j>1$, $e_j \perp \ora{x\Ps z}$, and $e_{j-1} \cap e_j = \{s\}$. Moreover, since $w \not \in V(P_{cz}^{\overline{w}})$, we have $w \not\in e_j$: since $e_j \perp \ora{x\Ps z}$, this implies $z \not\in e_j$. Therefore $j<L$, so we can consider the edge $e_{j+1}$. Since $s \in e_{j-1} \cap e_j$, we have $s \not\in e_{j+1}$, so $e_j \cap e_{j+1} \subseteq e_j \setminus \{s\} \subseteq V(\Ps) \setminus \{z\}$: in particular $j<\ell$ by maximality of $\ell$.
				
				\item Finally, we show that $s \in e_i$ for some $\ell \leq i \leq L$, i.e., $s \in V(P_2)$. Note that $P_2 \subseteq D$: indeed, we have $P_2 \subseteq P_{cz}^{\overline{w}} \subseteq D \cup e$, and $e=e_1$ is not an edge of $P_2$ because $\ell \geq 2$. Since $P_2 \subseteq P_{cz}^{\overline{w}}$ does not contain $w$, we have $\Start(\ora{zP_2})\neq \Start(\ora{z\Ps})$, so we can apply Union Lemma \ref{Lemma2}. There are no $x$-tadpoles in $P_2 \cup \Ps \subseteq D$ by Proposition \ref{prop_structure}\ref{item6}, so we get a $z$-cycle $C$ in $P_2 \cup \Ps$. Since $C$ must contain $s$, we have $s \in V(P_2) \cup V(\Ps)$ hence $s \in V(P_2)$.
				
			\end{itemize}
			
			\noindent Since $j<\ell$, $e_{j-1}$ is disjoint from $e_{\ell},\ldots,e_L$ by definition of a chain. However, we have just shown that $s \in e_{j-1}$ and $s \in e_i$ for some $l \leq i \leq L$. This is a contradiction. \qedhere

	\end{enumerate}
	
\end{proof}

\begin{lemme}\label{Lemma6}
	Let $H$ be a marked hypergraph that is not a trivial Maker win, and let $x \in V(H) \setminus M(H)$. Let $D$ be a $\widehat{\D_1}$-danger at $x$ in $H$, with a maximal decomposition $(z , \F_D = \S_z \cup \T_z \cup \P_{zx})$. Suppose that $\I{H^{+z}}{z\D_1(H)} \neq \varnothing$. Then there is a unique edge $e_x$ in $D$ that is incident to $x$. Moreover, let $P_x$ be an $x$-chain in $H$ such that $V(P_x) \cap (V(D) \setminus \{x\}) \neq \varnothing$ and $\Start(\ora{xP_x}) \neq e_x$: then $D \cup P_x$ contains an $x$-snake or an $x$-tadpole.
\end{lemme}

\begin{proof}

	Let $s \in \I{H^{+z}}{z\D_1(H)}$, and let $\Ps \in \P_{zx}$ be such that $s \not\in V(\Ps)$ as per Proposition \ref{prop_s}. We define $e_x = \Start(\ora{x\Ps z})$. We will show at the end of the proof that $e_x$ is the unique edge of $D$ containing $x$.
	\\ For now, let $P_x$ be an $x$-chain in $H$ such that $V(P_x) \cap (V(D) \setminus \{x\}) \neq \varnothing$ and $\Start(\ora{xP_x}) \neq e_x$. Up to replacing $P_x$ by the projection $\proj{V(D) \setminus \{x\}}{x}{P_x}$, assume that $\End(\ora{xP_x})$ is the only edge of $P_x$ that intersects $V(D) \setminus \{x\}$. Suppose for a contradiction that:
	\begin{equation}\label{cont2}
		\text{There are no $x$-snakes and no $x$-tadpoles in $D \cup P_x$.}\tag{$\ast$}
	\end{equation}
	\noindent Let $e = \Start(\ora{xP_x})$. We distinguish between two cases, pictured in Figure \ref{Lemma6-1}.
	
	\begin{figure}[h]
		\centering
		\includegraphics[scale=.55]{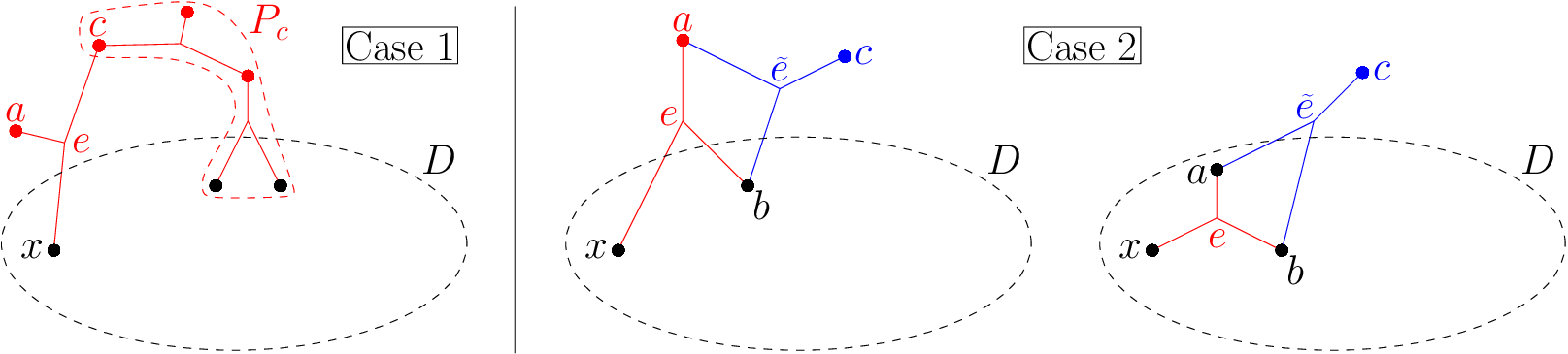}
		\caption{Case 1: $e \cap V(D)=\{x\}$. Case 2: $|e \cap V(D)|=2$ or $|e \cap V(D)|=3$.}
		\label{Lemma6-1}
	\end{figure}
	
	\begin{enumerate}[leftmargin=*,label=\textbf{\arabic*)}]

		\item Case 1: $P_x$ is of length at least 2, i.e., $e \cap V(D)=\{x\}$ (since we have just assumed that $\End(\ora{xP_x})$ is the only edge of $P_x$ that intersects $V(D) \setminus \{x\}$).
			\\ Write $e=\{x,a,c\}$ where $c$ is the only vertex in $\inn(P_x) \cap e$, and let $P_c$ be the $c$-chain defined as $P_c= P_x^{-x-a}$ (see Figure \ref{Lemma6-1}). Since $H$ is not a trivial Maker win and $\I{H^{+z}}{z\D_1(H)} \neq \varnothing$, we can apply item \ref{item1_Lemma5} of Union Lemma \ref{Lemma5} in $H$ to $D$ and $P_c$, which tells us that $D \cup P_c$ contains one of the following:
			\begin{itemize}[noitemsep,nolistsep]
				\item a $cx$-chain $P$. Then the walk $(x,e,c)\oplus \ora{cPx}$ represents an $x$-cycle in $D \cup P_x$.
				\item a $c$-tadpole $T$. If $x \in V(T)$, then $T$ contains a $cx$-chain by Substructure Lemma \ref{lemma_subchain4bis}, so we simply go back to that case. If $x \not\in V(T)$, then the walk $(x,e,c)\oplus \ora{cT}$ represents an $x$-tadpole in $D \cup P_x$.
				\item a $c$-snake $S$. If $x \in V(S)$, then $S$ contains a $cx$-chain by Substructure Lemma \ref{lemma_subchain1} so we simply go back to that case. If $x \not\in V(S)$, then the walk  $(x,e,c)\oplus \ora{cS}$ represents an $x$-snake in $D \cup P_x$.
			\end{itemize}
			All three possibilities contradict (\ref{cont2}). This ends Case 1.

		\item Case 2: $P_x=e$ is of length 1, i.e., $|e \cap V(D)| \geq 2$.
			\\ Write $e=\{x,a,b\}$. As a gadget, we create a new non-marked vertex $c$ and an edge $\widetilde{e}=\{a,b,c\}$ (see Figure \ref{Lemma6-1}).
			\begin{claim}\label{claim_Lemma6}
				Let $X$ be a subhypergraph of $D \cup \widetilde{e}$ such that $\widetilde{e} \in E(X)$ and $x \not\in V(X)$, and define the subhypergraph $\varphi(X)$ of $D \cup e$ obtained from $X$ by replacing $c$ by $x$ and $\widetilde{e}$ by $e$. Then we have the isomorphisms of pointed marked hypergraphs: $(X,c) \sim (\varphi(X),x)$ and $(X,v) \sim (\varphi(X),v)$ for all $v \in V(X) \setminus (M(X) \cup \{c\})$.
			\end{claim}
			\begin{proofclaim}[Proof of Claim \ref{claim_Lemma6}]
				This is straightforward.
			\end{proofclaim}
			\noindent The idea is to apply Union Lemma \ref{Lemma5} in $D \cup \widetilde{e}$ to $D$ and $P_c =\widetilde{e}$, and then contradict (\ref{cont2}) through replacing $\widetilde{e}$ by $e$ in the obtained subhypergraph as per Claim \ref{claim_Lemma6}. To do so, we need to check that $D \cup \widetilde{e}$ is not a trivial Maker win and that $\I{(D \cup \widetilde{e})^{+z}}{z\D_1(D \cup \widetilde{e})} \neq \varnothing$.
			\\ The former is clear: we know $D \subseteq H$ is not a trivial Maker win, moreover there are no $x$-snakes in $P_x$ by (\ref{cont2}) so $M(e)=\varnothing$ hence $M(\widetilde{e})=\varnothing$, so $D \cup \widetilde{e}$ is not a trivial Maker win either.
			\\ The latter is more difficult, because the addition of $\widetilde{e}$ may create new $\D_1$-dangers at $z$. However, we now show that they all contain $s$, i.e., $s \in \I{(D \cup \widetilde{e})^{+z}}{z\D_1(D \cup \widetilde{e})}$. Indeed, let $X$ be a $\D_1$-danger at $z$ in $D \cup \widetilde{e}$: we want to show that $s \in V(X)$.
			\begin{itemize}[noitemsep,nolistsep]
				\item Suppose $\widetilde{e} \not\in E(X)$. Then $X \in z\D_1(H)$, hence $s \in V(X)$ by definition of $s$.
				\item Suppose $\widetilde{e} \in E(X)$ and $x \not\in V(X)$. By Claim \ref{claim_Lemma6}, we have $(X,z) \sim (\varphi(X),z)$, therefore $\varphi(X)$ is a $\D_1$-danger at $z$ in $D \cup e$ hence $s \in V(\varphi(X))$. Since $s \neq x$ by Proposition \ref{prop_s}, this yields $s \in V(X)$.
				\item Finally, suppose $\widetilde{e} \in E(X)$ and $x \in V(X)$. In particular, we have $c,x \in V(X)$.
					\begin{itemize}[noitemsep,nolistsep]
						\item If there exists a $cx$-chain $P$ in $X$, then necessarily $\Start(\ora{cPx})=\widetilde{e}$ since $\widetilde{e}$ is the only edge incident to $c$ in $D \cup \widetilde{e}$. Either $a$ or $b$, say $b$, is an inner vertex of $P$, so that $P^{-c-a}$ is a $bx$-chain in $D$ that does not contain $a$. This means that $P^{-c-a} \cup e$ is an $x$-cycle in $D \cup e$, contradicting (\ref{cont2}).
						\item If there are no $cx$-chains in $X$, then the only possibility according to Substructure Lemmas \ref{lemma_subchain1} and \ref{lemma_subchain4bis} is that $X$ is a $z$-tadpole (we write $X=T$) such that $C_T$ is of length 2 and $\out(C_T)=\{c,x\}$ as in Figure \ref{Lemma6-2}. Therefore, the edges incident respectively to $c$ and $x$ in $T$ intersect on two vertices. Since the edge incident to $c$ in $T$ is necessarily $\widetilde{e}=\{a,b,c\}$, the edge incident to $x$ in $T$ is precisely $\{a,b,x\}=e$. Define the $zx$-chain $P = T^{-c}$, as in Figure \ref{Lemma6-2}. Since $T \subseteq D \cup \widetilde{e}$, we have $P \subseteq D$, so $P \in \P_{zx}$ by maximality of the decomposition. Moreover $\Start(\ora{xPz})=e\neq e_x = \Start(\ora{x\Ps z})$ by assumption, so the final assertion of Proposition \ref{prop_s} ensures that $s \in V(P) \subseteq V(X)$.
						\begin{figure}[h]
							\centering
							\includegraphics[scale=.55]{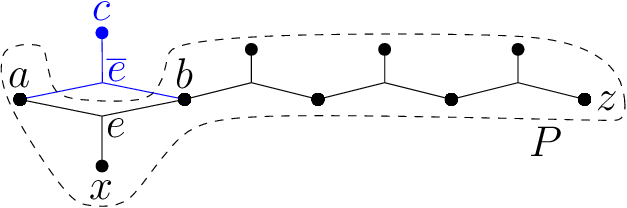}
							\caption{Illustration of $X=T$ if there are no $cx$-chains in $X$.}
							\label{Lemma6-2}
						\end{figure}
					\end{itemize}
			\end{itemize}
			Now that we have shown that $\I{(D \cup \widetilde{e})^{+z}}{z\D_1(D \cup \widetilde{e})} \neq \varnothing$, we can apply Union Lemma \ref{Lemma5} in $D \cup \widetilde{e}$ to $D$ and $P_c = \widetilde{e}$, which tells us that $D \cup \widetilde{e}$ contains one of the following:
			\begin{itemize}[noitemsep,nolistsep]
				\item a $cx$-chain $P$. In this case, taking $P$ and replacing $\widetilde{e}$ by $e$ yields an $x$-cycle. Indeed, write $\ora{cPx}=(c,e_1,\ldots,e_L,x)$: since $\widetilde{e}$ is the only edge incident to $c$, we have $e_1=\widetilde{e}$, so the walk $(x,e,e_2,\ldots,e_L,x)$ represents an $x$-cycle in $D \cup e$.
				\item a $c$-tadpole $T$. Since $\widetilde{e}$ is the only edge incident to $c$, we have $\widetilde{e} \in E(T)$. If $x \in V(T)$, then $T$ contains a $cx$-chain so we simply go back to that case. If $x \not\in V(T)$, then by Claim \ref{claim_Lemma6} we have $(T,c) \sim (\varphi(T),x)$, so $\varphi(T)$ is an $x$-tadpole in $D \cup e$.
				\item a $c$-snake $S$. Since $\widetilde{e}$ is the only edge incident to $c$, we have $\widetilde{e} \in E(S)$. If $x \in V(S)$, then $S$ contains a $cx$-chain so we simply go back to that case. If $x \not\in V(S)$, then by Claim \ref{claim_Lemma6} we have $(S,c) \sim (\varphi(S),x)$, so $\varphi(S)$ is an $x$-snake in $D \cup e$.
			\end{itemize}
			All three possibilities thus contradict (\ref{cont2}). This ends Case 2 and concludes the proof of the second assertion of the lemma.
			
	\end{enumerate}
	
	\noindent Finally, we prove the first assertion of the lemma, namely, that $e_x$ is the only edge of $D$ that is incident to $x$. Suppose for a contradiction that there exists $e'_x \in E(D)$ such that $x \in e'_x$ and $e'_x \neq e_x$. Define $P_x = e'_x$: we have $V(P_x) \cap (V(D) \setminus \{x\}) = e'_x \setminus \{x\} \neq \varnothing$ and $\Start(\ora{xP_x}) = e'_x \neq e_x$. Therefore, we can apply the second assertion of the lemma to the chain $P_x$: we get an $x$-snake or an $x$-tadpole in $D \cup P_x = D$, contradicting Proposition \ref{prop_structure}\ref{item6}.
\end{proof}

\subsubsection{Inside structure}

\noindent The two previous lemmas address the union of a $\widehat{\D_1}$-danger and a chain. We now look at a $\widehat{\D_1}$-danger alone. In Figure \ref{Example_Dangers}, all featured examples were unions of $z$-tadpoles and $zx$-chains only, no $z$-snakes. Also, $x$ was of degree 1 in all of them. We can now show these properties hold in all interesting cases:

\begin{proposition}\label{prop_nomarkedvertex}
	Let $(D,x) \in \widehat{\D_1}$, with a maximal decomposition $(z , \F_D = \S_z \cup \T_z \cup \P_{zx})$. If $J(\D_1,D)$ holds, then $M(D)=\varnothing$. In particular, we have $\F_D=\T_z \cup \P_{zx}$.
\end{proposition}

\begin{proof}
	Suppose for a contradiction that there exists some $m \in M(D)$. As a gadget, we add two new non-marked vertices $a$ and $c$ as well as a new edge $\widetilde{e}=\{a,c,m\}$. This does not create any new $\D_1$-danger at $z$: indeed, it is obvious that a $z$-snake or a $z$-tadpole cannot contain an edge with two non-marked vertices of degree 1 other than $z$. For that reason, the fact that $J(\D_1,D)$ holds implies that $J(\D_1,D \cup \widetilde{e})$ holds as well. Moreover, since $D$ is not a trivial Maker win by Proposition \ref{prop_structure}\ref{item5}, $D \cup \widetilde{e}$ is not either. Therefore, item \ref{item2_Lemma5} of Union Lemma \ref{Lemma5} applied to $D$ and $P_c = \widetilde{e}$ ensures that $D \cup \widetilde{e}$ contains a $cx$-chain or a $c$-tadpole. Since $\widetilde{e}$ is the only edge containing $c$, it is easy to see by removing $\widetilde{e}$ that $D$ contains an $mx$-chain or an $m$-tadpole respectively (see Figure \ref{NoMarkedVertex}). The former is impossible because an $mx$-chain in $D$ is an $x$-snake in $D$, which cannot exist by Proposition \ref{prop_structure}\ref{item6}. The latter is impossible by Proposition \ref{prop_markedvertex}. We can conclude that $M(D)=\varnothing$, which implies $\S_z=\varnothing$ hence the final assertion.
\end{proof}

\begin{figure}[h]
	\centering
	\includegraphics[scale=.55]{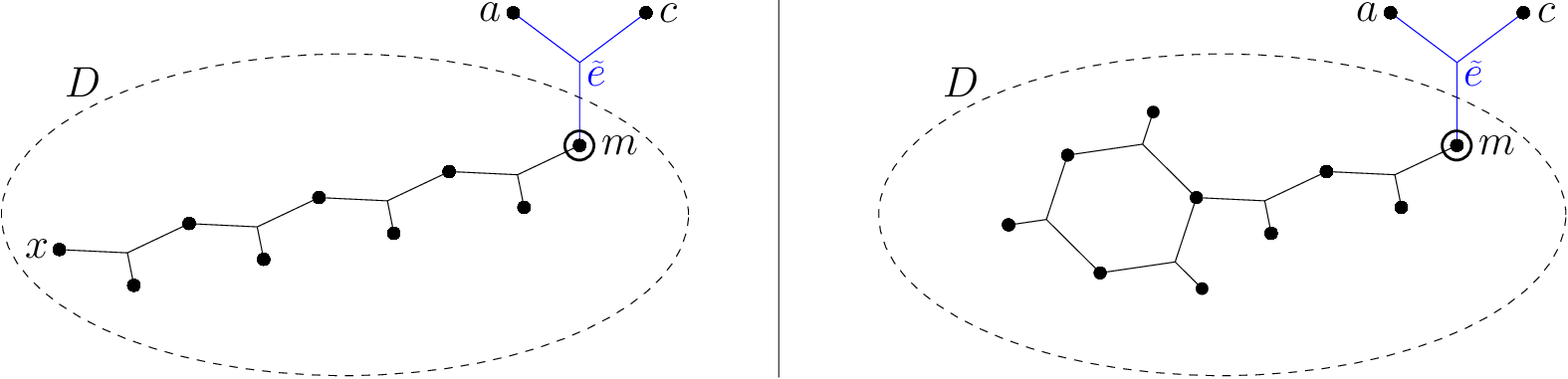}
	\caption{Left: a $cx$-chain yields an $xm$-snake. Right: a $c$-tadpole yields an $m$-tadpole.}
	\label{NoMarkedVertex}
\end{figure}

\begin{proposition}\label{prop_degree1}
	Let $(D,x) \in \widehat{\D_1}$, with a maximal decomposition $(z , \F_D = \S_z \cup \T_z \cup \P_{zx})$. If $\I{D^{+z}}{z\D_1(D)} \neq \varnothing$, then $x$ is of degree 1 in $D$.
\end{proposition}

\begin{proof}
	This is the first assertion of Union Lemma \ref{Lemma6} applied in $H=D$.
\end{proof}

\noindent The next result delves into the inside structure of the $\widehat{\D_1}$-dangers with much more precision.

\begin{proposition}\label{prop_insidestructure}
	Let $(D,x) \in \widehat{\D_1}$, with a maximal decomposition $(z , \F_D = \S_z \cup \T_z \cup \P_{zx})$. Suppose that $J(\D_1,D)$ holds. Then $D$ is of at least one of the two following types (see Figure \ref{TwoTypes}):
	\begin{enumerate}[noitemsep,nolistsep,label=\textup{(\arabic*)}]
		\item $D$ contains:
			\begin{itemize}[noitemsep,nolistsep]
				\item a $z$-cycle $C$ such that $x \not\in V(C)$;
				\item an $xw$-chain $P_{xw}$ for some $w \in \out(C)$ such that $V(P_{xw}) \cap V(C)=\{w\}$;
				\item some $X \in \F_D$ such that $V(X) \cap V(P_{xw}) \neq \varnothing$ and $e \setminus \{w\} \not\subseteq V(X)$ where $e$ denotes the unique edge of $C$ containing $w$.
			\end{itemize}
		\item $D$ contains:
			\begin{itemize}[noitemsep,nolistsep]
				\item a $z$-cycle $C$ such that $x \not\in V(C)$;
				\item an $xw$-chain $P_{xw}$ for some $w \in V(C)$ such that $V(P_{xw}) \cap V(C)=\{w,w'\}$ where we have defined $w' = o(w,\ola{xP_{xw}w})$.
			\end{itemize}
	\end{enumerate}
\end{proposition}

\begin{figure}[h]
	\centering
	\includegraphics[scale=.55]{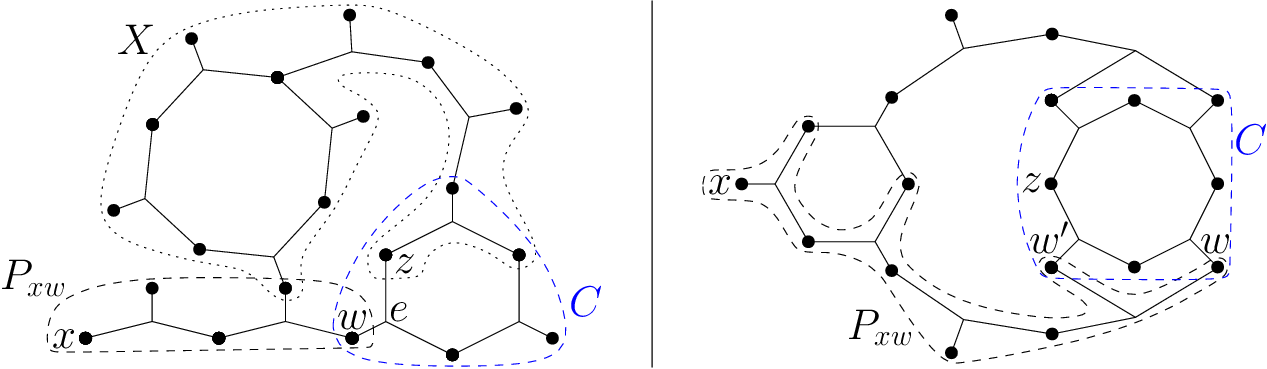}
	\caption{Two $\widehat{\D_1}$-dangers. The left one is of type (1) only (same for the other two from Figure \ref{Example_Dangers}). The right one is of type (2) only.}
	\label{TwoTypes}
\end{figure}

\begin{proof}

	Assume that $D$ is not of type (2): we show that $D$ is of type (1).
	\begin{claim}\label{claim_TwoTypes0}
		There exists a pair $(C,P_{xw})$ where $C$ is a $z$-cycle and $P_{xw}$ is an $xw$-chain for some $w \in \out(C)$ such that $V(P_{xw}) \cap V(C)=\{w\}$.
	\end{claim}
	\begin{proofclaim}[Proof of Claim \ref{claim_TwoTypes0}]
		The existence of $C$ is given by Proposition \ref{prop_structure}\ref{item7}. The existence of $P_{xw}$ is also straightforward:
 		\begin{itemize}[wide,noitemsep,nolistsep]
			\item Suppose $x \in V(C)$. Necessarily $x \in \out(C)$, otherwise $C$ would be an $x$-cycle, contradicting Proposition \ref{prop_structure}\ref{item6}. Therefore, simply take $w = x$ and $P_{xw}$ of length 0.
			\item Suppose $x \not\in V(C)$. Let $P \in \P_{zx}$ and define $P_x = \proj{V(C)}{x}{P}$. By definition of a projection: $|\End(\ora{xP_x}) \cap V(C)|\in\{1,2\}$. We cannot have $|\End(\ora{xP_x}) \cap V(C)|=2$ because $D$ would be of type (2), therefore $|\End(\ora{xP_x}) \cap V(C)|=1$. Let $w$ be the only vertex in $\End(\ora{xP_x}) \cap V(C)$. Necessarily $w \in \out(C)$, otherwise $P_x \cup C$ would be an $x$-tadpole, contradicting Proposition \ref{prop_structure}\ref{item6}. Take $P_{xw}= P_x$. \qedhere
		\end{itemize}
	\end{proofclaim}
	\noindent Of all pairs $(C,P_{xw})$ as in Claim \ref{claim_TwoTypes0}, we choose one where $P_{xw}$ is longest. This choice ensures that:
	\begin{claim}\label{claim_TwoTypes}
		For any $z$-cycle $C'$ in $D$, we have $V(C') \cap V(P_{xw}) \neq \varnothing$.
	\end{claim}
	\begin{proofclaim}[Proof of Claim \ref{claim_TwoTypes}]
		Suppose for a contradiction that there exists a $z$-cycle $C'$ such that $V(C') \cap V(P_{xw}) = \varnothing$. Since $z \in V(C') \cap (V(C) \setminus \out(C))$, the projection $P = \proj{V(C')}{w}{C}$ is well defined, and it is of positive length because $w \not\in V(C')$. Therefore, defining $P'_x = \HH{\ora{xP_{xw}w} \oplus \ora{wP}}$, the chain $P'_x$ is strictly longer than $P_{xw}$. For the same reason as $P_x$ in the proof of Claim \ref{claim_TwoTypes0} above, $P'_x$ satisfies $\End(\ora{xP'_x}) \cap V(C')=\{w'\}$ for some $w' \in \out(C')$, and $P'_x$ is an $xw'$-chain. The pair $(C',P'_x)$ thus contradicts the maximality of the length of $P_{xw}$.
	\end{proofclaim}
	\noindent We will show that $x \not\in V(C)$, i.e., $x \neq w$ at the end of the proof. For now, let $e$ be the only edge of $C$ containing $w$, and let us show the existence of $X \in \F_D$ such that $V(X) \cap V(P_{xw}) \neq \varnothing$ and $e \setminus \{w\} \not\subseteq V(X)$. 
	\\ Let us first address the case $z \in e$. Since $z \in \inn(C)$ and $w \in \out(C)$, we have $z \neq w$. Let $v$ be the third vertex of $e$, so that $e=\{w,z,v\}$. By Proposition \ref{prop_structure}\ref{item2}, there exists $X^{\overline{v}} \in \F_D$ such that $v \not\in V(X^{\overline{v}})$, which implies $e \setminus \{w\} \not\subseteq V(X^{\overline{v}})$. Suppose for a contradiction that $V(X^{\overline{v}}) \cap V(P_{xw}) = \varnothing$. In particular $X^{\overline{v}}$ is not a $zx$-chain. We also know $X^{\overline{v}}$ is not a $z$-snake by Proposition \ref{prop_nomarkedvertex}, so $X^{\overline{v}}$ is a $z$-tadpole. We write $X^{\overline{v}}=T$. Since $V(T) \cap (V(P_{xw}) \cup e) = \{z\}$, the walk $\ora{xP_{xw}w} \oplus (w,e,z) \oplus \ora{zT}$ represents an $x$-tadpole in $D$, contradicting Proposition \ref{prop_structure}\ref{item6}. In conclusion, we have $V(X^{\overline{v}}) \cap V(P_{xw}) \neq \varnothing$.
	\\ We can now assume that $z \not\in e$. Write $\ora{zC}=(z,e_1,\ldots,e_L,z)$: we have $e=e_i$ for some $1 \leq i \leq L$. Actually, since $z \not\in e$, we have $L \geq 3$ and $2 \leq i \leq L-1$. We can thus define $w_1$ (resp. $w_2$) as the only vertex in $e_{i-1} \cap e_i$ (resp. in $e_i \cap e_{i+1}$), and we have $e=\{w,w_1,w_2\}$. Therefore, defining $P_1 =\HH{(z,e_1,\ldots,e_{i-1},w_1)}$ and $P_2 = \HH{(z,e_L,e_{L-1},\ldots,e_{i+1},w_2)}$, $P_1$ is a $zw_1$-chain and $P_2$ is a $zw_2$-chain. These notations are summed up in Figure \ref{NotationsType1}.
	
	\begin{figure}[h]
		\centering
		\includegraphics[scale=.55]{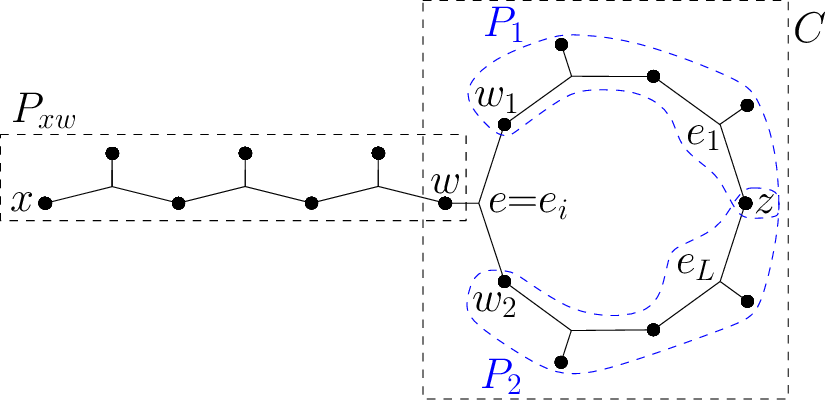}
		\caption{Summary of the notations used in the proof of Proposition \ref{prop_insidestructure}.}
		\label{NotationsType1}
	\end{figure}
	
	\noindent Since $z \not\in e$, we have $w_1, w_2 \neq z$. By Proposition \ref{prop_structure}\ref{item2}, for all $j \in \{1,2\}$, there exists $X^{\overline{w_j}} \in \F_D$ such that $w_j \not\in V(X^{\overline{w_j}})$, which implies $e \setminus \{w\} \not\subseteq V(X^{\overline{w_j}})$. We choose $X^{\overline{w_1}}=X^{\overline{w_2}}$ if possible, i.e., if there exists an element of $\F_D$ containing neither $w_1$ nor $w_2$. Suppose for a contradiction that $V(X^{\overline{w_1}}) \cap V(P_{xw}) = \varnothing$ and $V(X^{\overline{w_2}}) \cap V(P_{xw}) = \varnothing$.
	\\ In particular, $X^{\overline{w_1}}$ and $X^{\overline{w_2}}$ are not $zx$-chains, moreover they are not $z$-snakes by Proposition \ref{prop_nomarkedvertex} and they are not $z$-cycles by Claim \ref{claim_TwoTypes}. Therefore, $X^{\overline{w_1}}$ and $X^{\overline{w_2}}$ are $z$-tadpoles that are not cycles. We write $X^{\overline{w_1}} = T^{\,\overline{w_1}}$ and $X^{\overline{w_2}} = T^{\,\overline{w_2}}$. We distinguish between two cases, obtaining a contradiction for both.
	
	\begin{itemize}[leftmargin=*]
	
		\item First case: $e_L \not\in E(T^{\,\overline{w_1}})$ or $e_1 \not\in E(T^{\,\overline{w_2}})$. By symmetry, assume that $e_L \not\in E(T^{\,\overline{w_1}})$.
		\\ It is impossible that $V(T^{\,\overline{w_1}}) \cap V(P_2)=\{z\}$, otherwise we would have $V(T^{\,\overline{w_1}}) \cap (V(P_{xw}) \cup e \cup V(P_2))=\{z\}$ so the walk $\ora{xP_{xw}w} \oplus (w,e,w_2) \oplus \ora{w_2P_2z} \oplus \ora{zT^{\,\overline{w_1}}}$ would represent an $x$-tadpole in $D$, contradicting Proposition \ref{prop_structure}\ref{item6}. Therefore, the projection $P^{\,\overline{w_1}} = \proj{V(P_2) \setminus \{z\}}{z}{T^{\,\overline{w_1}}}$ is well defined. Since $P^{\,\overline{w_1}} \subseteq T^{\,\overline{w_1}}$, we have $V(P^{\,\overline{w_1}}) \cap (V(P_{xw}) \cup \{w_1\})=\varnothing$ and $e_L \not\in E(P^{\,\overline{w_1}})$. In particular $\Start(\ora{zP^{\,\overline{w_1}}}) \neq e_L=\Start(\ora{zP_2w_2})$, so we can apply Union Lemma \ref{Lemma2}. Since $P_2 \cup P^{\,\overline{w_1}}$ cannot contain a $z$-cycle by Claim \ref{claim_TwoTypes}, it contains a $w_2$-tadpole $T$. We have $V(T) \subseteq V(P_2) \cup V(P^{\,\overline{w_1}})$ hence $V(T) \cap (V(P_{xw}) \cup e) = \{w_2\}$, so the walk $\ora{xP_{xw}w} \oplus (w,e,w_2) \oplus \ora{w_2T}$ represents an $x$-tadpole in $D$, contradicting Proposition \ref{prop_structure}\ref{item6}.
		
		\item Second case: $e_L \in E(T^{\,\overline{w_1}})$ and $e_1 \in E(T^{\,\overline{w_2}})$.
		\\ Since $T^{\,\overline{w_1}}$ and $T^{\,\overline{w_2}}$ are not cycles, $z$ is of degree 1 in both of them, hence $e_1 \not\in E(T^{\,\overline{w_1}})$ and $e_L \not\in E(T^{\,\overline{w_2}})$. Since $e_1 \in E(T^{\,\overline{w_2}})$ and $e_1 \not\in E(T^{\,\overline{w_1}})$, we have $T^{\,\overline{w_1}} \neq T^{\,\overline{w_2}}$, so our initial choice of $T^{\,\overline{w_1}}$ and $T^{\,\overline{w_2}}$ ensures that $w_2 \in V(T^{\,\overline{w_1}})$ and $w_1 \in V(T^{\,\overline{w_2}})$.
		
			\begin{itemize}[nolistsep,noitemsep]
			
				\item Firstly, suppose that $w_2 \not\in \out(C_{T^{\,\overline{w_1}}})$ or $w_1 \not\in \out(C_{T^{\,\overline{w_2}}})$. By symmetry, assume that $w_2 \not\in \out(C_{T^{\,\overline{w_1}}})$. By Substructure Lemma \ref{lemma_subtadpole}, $T^{\,\overline{w_1}}$ contains a $w_2$-tadpole $T$. Since $V(T) \cap (V(P_{xw}) \cup e)=\{w_2\}$, the walk $\ora{xP_{xw}w} \oplus (w,e,w_2) \oplus \ora{w_2T}$ represents an $x$-tadpole in $D$, which contradicts Proposition \ref{prop_structure}\ref{item6}.
				
				\item Finally, suppose that $w_2 \in \out(C_{T^{\,\overline{w_1}}})$ and $w_1 \in \out(C_{T^{\,\overline{w_2}}})$. Since $J(\D_1,D)$ holds, there exists $s \in \I{D^{+z}}{z\D_1(D)}$. In particular, $s \in V(T^{\,\overline{w_1}}) \cap V(C)$. Since $w \not\in V(T^{\,\overline{w_1}})$, we have $s \neq w$, hence $s \in V(P_1)$ or $s \in V(P_2)$. By symmetry, assume $s \in V(P_1)$. In particular, $s \neq w_2$: by Substructure Lemma \ref{lemma_subchain4}, the fact that $w_2 \in \out(C_{T^{\,\overline{w_1}}})$ thus ensures the existence of a $zs$-chain $P_{zs}^{\,\overline{w_2}}$ in $T^{\,\overline{w_1}}$ that does not contain $w_2$. Since $e_1 \not\in E(T^{\,\overline{w_1}})$, we have $e_1 \not\in E(P_{zs}^{\,\overline{w_2}})$. Therefore $\Start(\ora{zP_{zs}^{\,\overline{w_2}}}) \neq e_1=\Start(\ora{zP_1w_1})$, moreover $V(P_{zs}^{\,\overline{w_2}}) \cap (V(P_1) \setminus \{z\}) \supseteq \{s\} \neq \varnothing$ so we can apply Union Lemma \ref{Lemma2}: since $P_1 \cup P_{zs}^{\,\overline{w_2}}$ cannot contain a $z$-cycle by Claim \ref{claim_TwoTypes}, it contains a $w_1$-tadpole $T$. We have $V(T) \subseteq V(P_1) \cup V(P_{zs}^{\,\overline{w_2}})$ hence $V(T) \cap (V(P_{xw}) \cup e) = \{w_1\}$, so the walk $\ora{xP_{xw}w} \oplus (w,e,w_1) \oplus \ora{w_1T}$ represents an $x$-tadpole in $D$, contradicting Proposition \ref{prop_structure}\ref{item6}.
			\end{itemize}
	\end{itemize}
	\noindent In conclusion, we have shown the existence of $X \in \F_D$ such that $V(X) \cap V(P_{xw}) \neq \varnothing$ and $e \setminus \{w\} \not\subseteq V(X)$. To prove that $D$ is of type (1), it only remains to show that $x \not\in V(C)$. Suppose for a contradiction that $x \in V(C)$, i.e., $x=w$, i.e., $V(P_{xw})=\{x\}$. Since $V(X) \cap V(P_{xw}) \neq \varnothing$ by definition of $X$, we get $x \in V(X)$, so there exists an edge $e'$ of $X$ that is incident to $x$. Moreover, $e$ is also incident to $w=x$. Since $e \not\in E(X)$ by definition of $X$, we have $e' \neq e$. Therefore, $e$ and $e'$ are two distinct edges of $D$ that are incident to $x$, contradicting Proposition \ref{prop_degree1}. This ends the proof.
\end{proof}

\section{Proofs of the main results}\label{Section5}

\subsection{Proof of Theorem \ref{theo_main_structure3}}

\noindent Proposition \ref{prop_forcing} shows that nunchakus are a Maker win thanks to a forcing strategy. Therefore, by Proposition \ref{prop_subwin} (Subhypergraph Monotonicity), any marked hypergraph that contains a nunchaku is a Maker win. We now show that the converse holds for the specific class of 3-uniform marked hyperforests. We can even give the exact value of $\tau_M(H)$ for any Maker win $H$ in that class (recall that $\tau_M(H)$ is defined as the minimum number of rounds in which Maker can ensure to get a fully marked edge when playing on $H$): the forcing strategy turns out to be suboptimal, as Maker can actually win on a nunchaku in a number of rounds that is logarithmic in the number of vertices rather than linear.

\begin{notation}
	Let $H$ be a marked hypergraph. We denote by $L(H)$ the length of a shortest nunchaku in $H$. If $H$ contains no nunchakus, then $L(H)=\infty$ by convention.
\end{notation}

\begin{lemme}\label{lemma_forest}
	Let $H$ be a 3-uniform marked hyperforest with no fully marked edges. Then $H$ is a Maker win if and only if $H$ contains a nunchaku. Moreover, if $H$ is a Maker win, then $\tau_M(H)=1+\lceil \log_2(L(H)) \rceil$.
\end{lemme}

\begin{proof}

	The case where $H$ contains a nunchaku of length 1 is obvious:
	
	\begin{claim}\label{claim_forest0}
		Let $H$ be a 3-uniform marked hyperforest with no fully marked edges, and suppose $L(H)=1$. Then $H$ is a trivial Maker win and $\tau_M(H)=1=1+\lceil \log_2(L(H)) \rceil$.
	\end{claim}
	\begin{proofclaim}[Proof of Claim \ref{claim_forest0}]
		This is obvious since a nunchaku of length 1 consists of a single edge, which contains exactly one non-marked vertex.
	\end{proofclaim}
	
	\noindent When there exists a nunchaku of length at least 2, Maker can use a "dichotomy strategy" to halve the length of a shortest nunchaku each round, until she gets one of length 1 (which is a trivial Maker win):
	
	\begin{claim}\label{claim_forest1}
		Let $H$ be a 3-uniform marked hypergraph such that $2 \leq L(H) < \infty$. Then there exists $x \in V(H)\setminus M(H)$ such that, for all $y \in V(H^{+x})\setminus M(H^{+x})$, we have $L(H^{+x-y})\leq \left\lceil\frac{L(H)}{2}\right\rceil$.
	\end{claim}
	\begin{proofclaim}[Proof of Claim \ref{claim_forest1}]
		Let $N$ be a shortest nunchaku in $H$. Let $x \in \inn(N)$ be in the exact middle of $N$ if $N$ is of even length, or as close to the middle as possible if $N$ is of odd length. By picking $x$, Maker creates two nunchakus of length at most $\left\lceil\frac{L(H)}{2}\right\rceil$ whose sole common vertex is $x$, so Breaker's answer $y$ cannot be contained in both of them at once. Therefore, at least one of these two nunchakus will be present in $H^{+x-y}$.
	\end{proofclaim}
	
	\noindent Meanwhile, Breaker has a strategy ensuring that, if there exists a nunchaku at the beginning of a round, then the length of a shortest nunchaku has not been more than halved after the round, and if there are no nunchakus before a round, then there are still none after the round:

	\begin{claim}\label{claim_forest2}
		Let $H$ be a 3-uniform marked hyperforest. Then, for all $x \in V(H)\setminus M(H)$, there exists $y \in V(H^{+x})\setminus M(H^{+x})$ such that $L(H^{+x-y})\geq \left\lceil\frac{L(H)}{2}\right\rceil$ ($=\infty$ if $L(H)=\infty$).
	\end{claim}
	\begin{proofclaim}[Proof of Claim \ref{claim_forest2}]
		Let $x \in V(H) \setminus M(H)$. Note that, for any $y$, the nunchakus in $H^{+x-y}$ are exactly the nunchakus in $H^{+x}$ that do not contain $y$. Therefore, let $\N$ be the collection of all nunchakus in $H^{+x}$ whose length is less than $\left\lceil\frac{L(H)}{2}\right\rceil$: proving the claim comes down to showing the existence of some $y \in V(H^{+x})\setminus M(H^{+x})$ such that all elements of $\N$ contain $y$. We can assume $\N \neq \varnothing$, otherwise there is nothing to show.
		\\ First of all, notice that all elements of $\N$ are $x$-nunchakus. Indeed, if some element of $\N$ was not an $x$-nunchaku, i.e., did not contain $x$, then it would be a nunchaku in $H$, which is impossible since it is of length less than $\left\lceil\frac{L(H)}{2}\right\rceil \leq L(H)$. Therefore, let $N_x \in \N$: we know $N_x$ is an $xm$-nunchaku for some $m \in M(H)$. Let $y =o(x,\ora{xN_xm})$, which is non-marked since $M(N_x)=\{x,m\}$. We now show that all elements of $\N$ contain $y$. Suppose for a contradiction that there exists $N'_x \in \N$ such that $y \not\in V(N'_x)$: we know $N'_x$ is an $xm'$-nunchaku for some $m' \in M(H)$.
		\begin{itemize}[wide,noitemsep,nolistsep]
			\item Suppose $V(N_x) \cap V(N'_x) \neq \{x\}$. Since $y \not\in V(N'_x)$, we have $\Start(\ora{xN_xm}) \neq \Start(\ora{xN'_xm'})$, therefore Union Lemma \ref{Lemma2} ensures that $N_x \cup N'_x$ contains an $x$-cycle or an $m$-tadpole. Both possibilities contradict the fact that $H$ is a hyperforest.
			\item Suppose $V(N_x) \cap V(N'_x) = \{x\}$. Then $N_x \cup N'_x$ is an $mm'$-chain in $H^{+x}$ and $M(N_x \cup N'_x)=\{m,m',x\}$. Let $N$ be the same as $N_x \cup N'_x$ except that $x$ is non-marked: since $N_x \cup N'_x$ is a subhypergraph of $H^{+x}$, $N$ is a subhypergraph of $H$. Therefore $N$ is an $mm'$-nunchaku in $H$, of length equal to the sum of the lengths of $N_x$ and $N'_x$. By definition of $\N$, $N_x$ and $N'_x$ are both of length less than $\left\lceil\frac{L(H)}{2}\right\rceil$, therefore $N$ is of length less than $L(H)$, contradicting the definition of $L(H)$. \qedhere
		\end{itemize}
	\end{proofclaim}
	\noindent We now have all the elements to prove the theorem by induction on $|V(H) \setminus M(H)|$.
    \begin{itemize}[leftmargin=*]
        \item For the base case, let $H$ be a 3-uniform marked hyperforest with no fully marked edges such that $|V(H) \setminus M(H)| \leq 1$. If $H$ contains a nunchaku, then it is necessarily of length 1, so Claim \ref{claim_forest0} concludes. If $H$ contains no nunchakus, then in particular $H$ contains no nunchakus of length 1: since $H$ has no fully marked edges, this means $H$ is not a trivial Maker win, so $H$ is a Breaker win.
        \item For the induction step, let $H$ be a 3-uniform marked hyperforest with no fully marked edges such that $|V(H) \setminus M(H)| \geq 2$, and assume that the theorem is true for all 3-uniform marked hyperforests with less than $|V(H) \setminus M(H)|$ non-marked vertices. Since Claim \ref{claim_forest0} concludes if $L(H)=1$, also assume $L(H)>1$: this ensures that $H$ is not a trivial Maker win and that there cannot be a fully marked edge after one round. Therefore, by the induction hypothesis:
	$$ \tau_M(H) = 1+\underset{x \in V(H) \setminus M(H)}{\min}\,\,\underset{y \in V(H^{+x}) \setminus M(H^{+x})}{\max}\,\,(1+\lceil \log_2(L(H^{+x-y})) \rceil),$$
	with the convention that $1+\lceil \log_2(L(H^{+x-y})) \rceil=\infty$ if $L(H^{+x-y})=\infty$.
	\begin{itemize}
		\item[--] Firstly, suppose $H$ contains a nunchaku, i.e., $2 \leq L(H)<\infty$. By Claims \ref{claim_forest1} and \ref{claim_forest2} respectively, the above equality yields $\tau_M(H) \leq 1+ \left( 1+\left\lceil \log_2\left(\left\lceil\frac{L(H)}{2}\right\rceil\right) \right\rceil \right)$ and $\tau_M(H) \geq 1+ \left( 1+\left\lceil \log_2\left(\left\lceil\frac{L(H)}{2}\right\rceil\right) \right\rceil \right)$. Therefore $\tau_M(H) = 1+ \left( 1+\left\lceil \log_2\left(\left\lceil\frac{L(H)}{2}\right\rceil\right) \right\rceil \right) = 1+\lceil \log_2(L(H)) \rceil < \infty$, so $H$ is a Maker win.
		\item[--] Finally, suppose $H$ contains no nunchakus, i.e., $L(H)=\infty$. By Claim \ref{claim_forest2}: for all $x \in V(H)\setminus M(H)$, there exists $y \in V(H^{+x})\setminus M(H^{+x})$ such that $L(H^{+x-y})=\infty$. Therefore $\tau_M(H)=\infty$, so $H$ is a Breaker win. This ends the proof. \qedhere
	\end{itemize}
    \end{itemize}
\end{proof}

\noindent We now recall the statement of Theorem \ref{theo_main_structure3} and prove it.

\begin{theoreme*}[\ref{theo_main_structure3}]
	Let $H$ be a 3-uniform marked hyperforest that is not a trivial Maker win, with $|V(H) \setminus M(H)| \geq 2$. Then $H$ is a Breaker win if and only if $J(\D_0,H)$ holds.
\end{theoreme*}

\begin{proof}
	Recall that the "only if" direction is automatic by Proposition \ref{prop_cn}. Suppose that $H$ is a Maker win. Since $H$ is not a trivial Maker win, we know $H$ has no fully marked edges, so Lemma \ref{lemma_forest} ensures that $H$ contains a nunchaku $N$. Again, since $H$ is not a trivial Maker win, $N$ is of length at least 2. Let $x \in \inn(N)$, so that $N$ is the union of two $x$-snakes $S$ and $S'$ such that: $|M(S)|=1$, $|M(S')|=1$, and $V(S) \cap V(S')=\{x\}$. Therefore, $\{S,S'\}$ is a $\D_0$-fork at $x$ in $N$, so $J(\D_0,N)$ does not hold and neither does $J(\D_0,H)$.
\end{proof}

\subsection{Proofs of Theorems \ref{theo_main_structure2}--\ref{theo_main_structure1}--\ref{theo_main_algo}--\ref{theo_main_duration} assuming Lemma \ref{lemma_main_structure}}

\noindent For now, we assume the following structural lemma, which will be proved later in this section. The remaining main results ensue quite easily.

\begin{lemme}\label{lemma_main_structure}
	Let $H$ be a 3-uniform marked hypergraph that is not a trivial Maker win, with $|V(H) \setminus M(H)| \geq 2$. Suppose that $J(\D_1,H)$ holds. Then, for any $x \in V(H) \setminus M(H)$ such that there exists an $x$-snake in $H$, we have $\I{H^{+x}}{x\D_2(H)} \neq \varnothing$.
\end{lemme}

\subsubsection{Structural results}

\noindent  We recall the statements of Theorems \ref{theo_main_structure2} and \ref{theo_main_structure1}, and prove both of them. We start with the latter, as it is used to prove the former.

\begin{theoreme*}[\ref{theo_main_structure1}]
	Let $H$ be a 3-uniform marked hypergraph that is not a trivial Maker win, with $|V(H) \setminus M(H)| \geq 2$. Then $H$ is a Breaker win if and only if $J(\D_2,H)$ holds. More precisely:
	\begin{enumerate}[noitemsep,nolistsep,label={\textup{(\roman*)}}]
		\item If $J(\D_2,H)$ does not hold, then $H$ is a Maker win and: any $x_1 \in V(H) \setminus M(H)$ such that $\I{H^{+x_1}}{x_1\D_2(H)}=\varnothing$ is a winning first pick for Maker. \label{item_Maker}
		\item If $J(\D_2,H)$ holds, then $H$ is a Breaker win and: for any first pick $x_1 \in V(H) \setminus M(H)$ of Maker, any $y_1 \in \I{H^{+x_1}}{x_1\D_2(H)}$ is a winning answer for Breaker. \label{item_Breaker}
	\end{enumerate}
	Moreover, $H$ is a Maker win if and only if Maker has a strategy ensuring that there is a nunchaku or a necklace in the updated marked hypergraph obtained after at most three full rounds of play on $H$.
\end{theoreme*}

\begin{proof}[Proof of Theorem \ref{theo_main_structure1} assuming Lemma \ref{lemma_main_structure}]
	
	Item \ref{item_Maker} is a direct consequence of Proposition \ref{prop_cn}. We now show item \ref{item_Breaker} by induction on $|V(H) \setminus M(H)|$.
	
	\begin{itemize}[leftmargin=*]
	
	\item Let us start with the base case $|V(H) \setminus M(H)| \in \{2,3\}$. Let $x_1 \in V(H) \setminus M(H)$ and $y_1 \in \I{H^{+x_1}}{x_1\D_2(H)}$, which exists since $J(\D_2,H)$ holds. The trivial danger of size 3 is in $\D_0 \subseteq \D_2$ (snake of length 1), so any edge of $H$ that contains both $x_1$ and some $m \in M(H)$ must have $y_1$ as its third vertex. Therefore, as $H$ is not a trivial Maker win, $H^{+x_1-y_1}$ is not a trivial Maker win either. Since $|V(H^{+x_1-y_1}) \setminus M(H^{+x_1-y_1})|\leq 1$, this means $H^{+x_1-y_1}$ is a Breaker win, so $H$ is a Breaker win.
	\item For the induction step, assume $|V(H) \setminus M(H)| \geq 4$ and item \ref{item_Breaker} to be true for marked hypergraphs with less non-marked vertices than $H$. Let $x_1 \in V(H) \setminus M(H)$ and $y_1 \in \I{H^{+x_1}}{x_1\D_2(H)}$, which exists since $J(\D_2,H)$ holds: we must show that $H^{+x_1-y_1}$ is a Breaker win. Let us first list a few important properties of $H^{+x_1-y_1}$:
	\begin{enumerate}[noitemsep,nolistsep,label={\textup{(\alph*)}}]
		\item $|V(H^{+x_1-y_1}) \setminus M(H^{+x_1-y_1})| = |V(H) \setminus M(H)| - 2 \geq 2$. \label{property_a}
		\item $H^{+x_1-y_1}$ is not a trivial Maker win. Just like in the base case, this comes from the fact that $H$ is not a trivial Maker win and the trivial danger of size 3 is in $\D_2$. \label{property_b}
		\item $J(\D_1,H^{+x_1-y_1})$ holds. Indeed, $J(\D_1,H)$ holds (because $J(\D_2,H)$ holds) and $y_1 \in \I{H^{+x_1}}{x_1\D_2(H)}$, so Proposition \ref{prop_equiv_dangers2} applies and ensures that $J(\D_1,H^{+x_1-y_1})$ holds. \label{property_c}
	\end{enumerate}
	Thanks to \ref{property_a} and \ref{property_b}, checking that property $J(\D_2,H^{+x_1-y_1})$ holds is sufficient to prove that $H^{+x_1-y_1}$ is a Breaker win, according to the induction hypothesis. Let $x \in V(H^{+x_1-y_1}) \setminus M(H^{+x_1-y_1})$: we want to show that $\I{(H^{+x_1-y_1})^{+x}}{x\D_2(H^{+x_1-y_1})} \neq \varnothing$. Assume that there exists some $D_0 \in x\D_2(H^{+x_1-y_1})$, otherwise $\I{(H^{+x_1-y_1})^{+x}}{x\D_2(H^{+x_1-y_1})}=\I{(H^{+x_1-y_1})^{+x}}{\varnothing}=V((H^{+x_1-y_1})^{+x}) \setminus M((H^{+x_1-y_1})^{+x}) \neq \varnothing$ trivially as $|V(H)\setminus M(H)|\geq 4$.
	
	\begin{enumerate}[label={\arabic*)}]
	
		\item First case: there are no $xx_1$-snakes in $H^{+x_1-y_1}$.
			\\ What happens here is that any vertex that hits all the $\D_2$-dangers at $x$ in $H$ still works in $H^{+x_1-y_1}$, because the marking of $x_1$ has not created any new $\D_2$-danger at $x$. Indeed, for all $D \in x\D_2(H^{+x_1-y_1})$:
			\begin{itemize}
				\item[--] If $(D,x) \in \D_1$ and $D$ is an $x$-snake, then $x_1 \not\in V(D)$ since we are assuming that there are no $xx_1$-snakes in $H^{+x_1-y_1}$.
				\item[--] If $(D,x) \in \D_1$ and $D$ is an $x$-tadpole, then $x_1 \not\in V(D)$ since $M(D)=\varnothing$.
				\item[--] If $(D,x) \in \widehat{\D_1}$, then $x_1 \not\in V(D)$ since $M(D)=\varnothing$ by Proposition \ref{prop_nomarkedvertex} (which \ref{property_c} allows us to use).
			\end{itemize}
			Therefore, $x\D_2(H^{+x_1-y_1}) \subseteq x\D_2(H^{-x_1-y_1}) \subseteq x\D_2(H)$. Let $y \in \I{H^{+x}}{x\D_2(H)}$, which exists since $J(\D_2,H)$ holds. To show that $y \in \I{(H^{+x_1-y_1})^{+x}}{x\D_2(H^{+x_1-y_1})}$, since $x\D_2(H^{+x_1-y_1}) \subseteq x\D_2(H)$, it suffices to check that $y \not\in\{x_1,y_1\}$. For this, we use $D_0$. On the one hand, we have $D_0 \in x\D_2(H^{+x_1-y_1}) \subseteq x\D_2(H)$ hence $y \in V(D_0)$. On the other hand, we have $D_0 \in x\D_2(H^{+x_1-y_1}) \subseteq x\D_2(H^{-x_1-y_1})$ hence $x_1,y_1 \not\in V(D_0)$. In conclusion, we do have $y \not\in\{x_1,y_1\}$, so $y \in \I{(H^{+x_1-y_1})^{+x}}{x\D_2(H^{+x_1-y_1})}$ hence $\I{(H^{+x_1-y_1})^{+x}}{x\D_2(H^{+x_1-y_1})} \neq \varnothing$.
			
		\item Second case: there is an $xx_1$-snake in $H^{+x_1-y_1}$.
			\\ Here, we have the $x$-snake that is necessary to apply Lemma \ref{lemma_main_structure} to $H^{+x_1-y_1}$. The other assumptions of this lemma are also verified thanks to \ref{property_a}, \ref{property_b} and \ref{property_c}. In conclusion, Lemma \ref{lemma_main_structure} applies and yields $\I{(H^{+x_1-y_1})^{+x}}{x\D_2(H^{+x_1-y_1})} \neq \varnothing$ as desired.
	\end{enumerate}
	
	\end{itemize}
	As for the final assertion of the theorem about the appearance of a nunchaku or a necklace, the "if" direction is clear since nunchakus and necklaces are Maker wins (Proposition \ref{prop_forcing}), while the "only if" direction follows from item \ref{item_Breaker} and Proposition \ref{prop_interpretation}.
\end{proof}

\begin{theoreme*}[\ref{theo_main_structure2}]
	Let $H$ be a 3-uniform marked hypergraph that is not a trivial Maker win, with $|V(H) \setminus M(H)| \geq 2$. Suppose that, for any $x \in V(H) \setminus M(H)$, there exists an $x$-snake in $H$. Then $H$ is a Breaker win if and only if $J(\D_1,H)$ holds. Moreover, $H$ is a Maker win if and only if Maker has a strategy ensuring that there is a nunchaku or a necklace in the updated marked hypergraph obtained after at most two full rounds of play on $H$.
\end{theoreme*}
	
\begin{proof}[Proof of Theorem \ref{theo_main_structure2} assuming Lemma \ref{lemma_main_structure}]
	If $H$ is a Breaker win, then $J(\D_1,H)$ holds by Proposition \ref{prop_cn}. Conversely, assume that $J(\D_1,H)$ holds. We claim that, actually, $J(\D_2,H)$ holds: indeed, for all $x \in V(H) \setminus M(H)$, there exists an $x$-snake in $H$ by assumption hence $\I{H^{+x}}{x\D_2(H)} \neq \varnothing$ by Lemma \ref{lemma_main_structure}. Therefore, $H$ is a Breaker win according to Theorem \ref{theo_main_structure1}. As for the final assertion of the theorem about the appearance of a nunchaku or a necklace, the "if" direction is clear since nunchakus and necklaces are Maker wins (Proposition \ref{prop_forcing}), while the "only if" direction follows from Proposition \ref{prop_interpretation}.
\end{proof}

\subsubsection{Algorithmic result}

\noindent The algorithm derives from the reduction, given by Theorem \ref{theo_main_structure1}, of the \makerbreaker problem on 3-uniform marked hypergraphs to the problem of the existence of a chain between two given vertices in a 3-uniform hypergraph. In a separate paper, we have shown that the latter is a tractable problem.

\begin{definition}\label{def:LCC}
	Let $H$ be a 3-uniform hypergraph and let $a \in V(H)$. The \textit{linear connected component of $a$ in $H$} is defined as the set $LCC_H(a)$ of all vertices $b$ of $H$ such that there exists an $ab$-chain in $H$.
\end{definition}

\begin{theoreme}[\textup{\cite[Theorem 4.2 with $k=3$]{GGS22}}]\label{theo_archipelagos}
	There exists an algorithm that, given a 3-uniform hypergraph $H$ and a vertex $x \in V(H)$, computes $LCC_H(x)$ in $O(m^2)$ time where $m=|E(H)|$.
\end{theoreme}

\noindent We now recall the statement of Theorem \ref{theo_main_algo} and prove it.

\begin{theoreme*}[\ref{theo_main_algo}]
	\makerbreaker is solved in polynomial time on 3-uniform marked hypergraphs, i.e., hypergraphs of rank 3. More precisely, there exists an algorithm which, given a hypergraph $H$ of rank 3 with $n$ vertices, $m$ edges and maximum degree $\Delta$, decides whether $H$ is a Maker win in time $O(\max(n^5m^2,n^6\Delta))$.
\end{theoreme*}

\begin{proof}[Proof of Theorem \ref{theo_main_algo} assuming Lemma \ref{lemma_main_structure}]
	First of all, let us transform the input (non-marked) hypergraph of rank 3 into a (still non-marked) 3-uniform hypergraph. We can assume that all edges are of size 2 or 3, otherwise we have a trivial Maker win. Using Proposition \ref{prop_duplication}, we then transform each edge of size 2 into two edges of size 3. At most, we have added two vertices and doubled the number of edges, so all relevant orders of magnitude are preserved. All in all, up to a preprocessing step in time $O(m)$, we can assume that $H$ is a non-marked 3-uniform hypergraph.
	\\ Since $\textsc{MakerBreaker}$ is trivially solved in time $O(1)$ on hypergraphs with less than 6 vertices, further assume $|V(H)| \geq 6$. By Theorem \ref{theo_main_structure1}, $H$ is a Maker win if and only if:
	\begin{center}
	\begin{tabular}{lll}
		$\exists x_1 \in V(H)$, & $\forall y_1 \in V(H) \setminus \{x_1\}$, & \hphantom{mmmmmm} \\
		$\exists x_2 \in V(H) \setminus \{x_1,y_1\}$, & $\forall y_2 \in V(H) \setminus \{x_1,y_1,x_2\}$, & \hphantom{mmmmmm} \\
		$\exists x_3 \in V(H) \setminus \{x_1,y_1,x_2,y_2\}$, & $\forall y_3 \in V(H) \setminus \{x_1,y_1,x_2,y_2,x_3\}$, & \hphantom{mmmmmm}
	\end{tabular} \\ $H^{+x_1-y_1+x_2-y_2+x_3-y_3}$ contains a fully marked edge, a nunchaku or a necklace.
	\end{center}
	Our algorithm checks whether the above assertion is true, as follows. Suppose that we are given $x_1,y_1,x_2,y_2,x_3,y_3$, and consider the marked hypergraph $H^{+x_1-y_1+x_2-y_2+x_3-y_3}$, in which $x_1,x_2,x_3$ are the only marked vertices:
	\begin{itemize}[noitemsep,nolistsep]
		\item Clearly, $H^{+x_1-y_1+x_2-y_2+x_3-y_3}$ contains a fully marked edge or a nunchaku if and only if it contains a chain between two marked vertices. By Theorem \ref{theo_archipelagos}, this can be tested in time $O(m^2)$: for all $i \in \{1,2\}$, compute $LCC_{H^{-y_1-y_2-y_3}}(x_i)$ and check whether it contains $x_j$ for some $j \in \{2,3\} \setminus \{i\}$. We note that membership to a given linear connected component $L$ is verifiable in constant time: this can be done by implementing $L$ as an array of length $n$, with indices corresponding to the vertices of $H$, where the element at index $i$ is a 1 if the corresponding vertex is in $L$ or a 0 otherwise.
		\item If $H^{+x_1-y_1+x_2-y_2+x_3-y_3}$ contains no fully marked edges and no nunchakus, then it contains a necklace if and only if it contains some edge $e=\{x,a,b\}$ with $x$ marked such that there exists an $xb$-chain $P$ that does not contain $a$ (the necklace is then $P \cup e$). This can be tested in time $O(\Delta m^2)$: for all $i \in \{1,2,3\}$ and for every edge $\{x_i,a,b\}$, compute $LCC_{H^{-y_1-y_2-y_3-a}}(x_i)$ and check whether it contains $b$, then repeat when exchanging the roles of $a$ and $b$.
	\end{itemize}
	This yields an algorithm in time $O(n^6(m^2+\Delta m^2))=O(n^6\Delta m^2)$. However, we can easily apply a slight improvement to it. Indeed, since the linear connected components that we need only depend on four out of the six vertices $x_1,y_1,x_2,y_2,x_3,y_3$, computing them inside of the six nested loops leads to some redundancies. Instead, it is better to compute all of them once and for all at the beginning of the algorithm, store them all (using nested arrays) and then access them when needed. Computing $LCC_{H^{-y_1-y_2-y_3}}(x)$ and $LCC_{H^{-y_1-y_2-y_3-a}}(x)$ for all possibilities of $y_1,y_2,y_3,a,x$ is done in time $O(n^5m^2)$, so the final running time is $O(n^5m^2+n^6\Delta)$.
\end{proof}

\subsubsection{"Fast-winning Maker strategy" result}

\begin{lemme}\label{lemma_duration}
    If $N$ is a nunchaku or a necklace of length $L$, then $\tau_M(N) = 1 + \lceil \log_2(L) \rceil$.
\end{lemme}

\begin{proof}
    The statement for nunchakus is a direct corollary of Lemma \ref{lemma_forest}. The statement for necklaces ensues since applying the reduction from Proposition \ref{prop_reduction}, which clearly preserves $\tau_M$, gives the same hypergraph for both (recall Figure \ref{Nunchaku0}).
\end{proof}

\noindent We recall the statement of Theorem \ref{theo_main_duration} and prove it.

\begin{theoreme*}[\ref{theo_main_duration}]
	Let $H$ be a 3-uniform marked hypergraph with $|V(H)\setminus M(H)|\geq 6$. If $H$ is a Maker win, then $\tau_M(H) \leq 3+\lceil \log_2(|V(H)\setminus M(H)|-5)\rceil$.
\end{theoreme*}

\begin{proof}[Proof of Theorem \ref{theo_main_duration} assuming Lemma \ref{lemma_main_structure}]
Suppose that $H$ is a Maker win. By Theorem \ref{theo_main_structure1}, Maker has a strategy ensuring that, after three rounds of play with successive picks $x_1,y_1,x_2,y_2,x_3,y_3$, there is a fully marked edge, a nunchaku or a necklace in $H^{+x_1-y_1+x_2-y_2+x_3-y_3}$. By Proposition \ref{prop_subwin} (Subhypergraph Monotonicity), to conclude the proof, it suffices to show that any nunchaku or necklace $N$ in $H^{+x_1-y_1+x_2-y_2+x_3-y_3}$ satisfies $\tau_M(N) \leq \lceil \log_2(|V(H) \setminus M(H)|-5)\rceil$.
	\begin{itemize}[wide,noitemsep,nolistsep]
		\item A nunchaku $N$ is of length $\frac{|V(N)|-1}{2}$, so it satisfies $\tau_M(N)=1+\left\lceil \log_2 \left(\frac{|V(N)|-1}{2}\right) \right\rceil = \lceil \log_2(|V(N)|-1) \rceil$ by Lemma \ref{lemma_duration}. Moreover, a nunchaku $N$ in $H^{+x_1-y_1+x_2-y_2+x_3-y_3}$ has two marked vertices and all its other vertices are in $V(H) \setminus (M(H) \cup \{x_1,y_1,x_2,y_2,x_3,y_3\})$, so it satisfies $|V(N)|\leq 2+(|V(H) \setminus M(H)|-6)$ hence $\tau_M(N) \leq \lceil \log_2(|V(H) \setminus M(H)|-5)\rceil$.
		\item A necklace $N$ is of length $\frac{|V(N)|}{2}$, so it satisfies $\tau_M(N)= \lceil \log_2(|V(N)|) \rceil$ by Lemma \ref{lemma_duration}. Moreover, a necklace $N$ in $H^{+x_1-y_1+x_2-y_2+x_3-y_3}$ has one marked vertex and all its other vertices are in $V(H) \setminus (M(H) \cup \{x_1,y_1,x_2,y_2,x_3,y_3\})$, so it satisfies $|V(N)|\leq 1+(|V(H) \setminus M(H)|-6)$ hence $\tau_M(N) \leq \lceil \log_2(|V(H) \setminus M(H)|-5)\rceil$. \qedhere
	\end{itemize}
\end{proof}

\subsection{Proof of Lemma \ref{lemma_main_structure}}

\noindent As we have just seen, all our main results will be proved once Lemma \ref{lemma_main_structure} is. This lemma basically tells us that, under the premise that $J(\D_1,H)$ holds, any $x$ at which there is a snake not only has no $\D_1$-forks (which is a given) but actually has no $\D_2$-forks either. Let us recall its full statement.

\begin{lemme*}[\ref{lemma_main_structure}]
	Let $H$ be a 3-uniform marked hypergraph that is not a trivial Maker win, with $|V(H) \setminus M(H)| \geq 2$. Suppose that $J(\D_1,H)$ holds. Then, for any $x \in V(H) \setminus M(H)$ such that there exists an $x$-snake in $H$, we have $\I{H^{+x}}{x\D_2(H)} \neq \varnothing$.
\end{lemme*}

\noindent We now prove this lemma. Let $H$ be a 3-uniform marked hypergraph that is not a trivial Maker win, with $|V(H) \setminus M(H)| \geq 2$, and suppose that $J(\D_1,H)$ holds. Let $x \in V(H) \setminus M(H)$ and $m \in M(H)$ such that there exists an $xm$-snake in $H$: we want to find some $y \in \I{H^{+x}}{x\D_2(H)}$. Since $J(\D_1,H)$ holds, we already know that $\I{H^{+x}}{x\D_1(H)} \neq \varnothing$, however picking an arbitrary $y \in \I{H^{+x}}{x\D_1(H)}$ does not work in general, as shown in Figure \ref{BadChoices}. In this example, we can see that $H$ satisfies the conditions of Lemma \ref{lemma_main_structure}, and that the only $\D_1$-dangers at $x$ in $H$ are two $xm$-snakes whose intersection $\I{H^{+x}}{x\D_1(H)}$ is represented by the vertices in square boxes. We can see that several of them are not in $\I{H^{+x}}{x\D_2(H)}$, because they miss $D$ which is a $\widehat{\D_1}$-danger at $x$: this is the case for the vertex $y'$ for instance.

\begin{figure}[h]
	\centering
	\includegraphics[scale=.58]{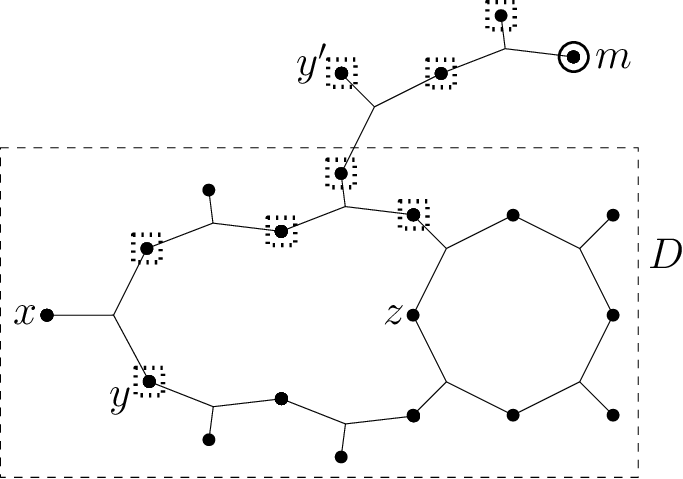}
	\caption{In this example, we have $y \in \I{H^{+x}}{x\D_2(H)}$, but $y' \in \I{H^{+x}}{x\D_1(H)} \setminus \I{H^{+x}}{x\D_2(H)}$ since $y' \not\in V(D)$. Square dotted boxes highlight the vertices in $\I{H^{+x}}{x\D_1(H)}$.}\label{BadChoices}
\end{figure}

\noindent This inspires us to choose $y \in \I{H^{+x}}{x\D_1(H)}$ furthest away from $m$, in terms of the following semimetric (like in Figure \ref{BadChoices} for example):

\begin{notation}\label{not:dist}
	For all $a,b \in V(H)$, we denote by $\dist_H(a,b)$ the length of a shortest $ab$-chain in $H$, where $\dist_H(a,b)=\infty$ by convention if there exist none.
\end{notation}

\noindent We now fix $y \in \I{H^{+x}}{x\D_1(H)}$ maximizing $\dist_H(y,m)$, and we suppose for a contradiction that $y \not\in \I{H^{+x}}{x\D_2(H)}$: there exists $D \in x\widehat{\D_1}(H)$, with a maximal decomposition $(z , \F_D = \S_z \cup \T_z \cup \P_{zx})$, such that $y \not\in V(D)$. The idea of the proof is to eventually exhibit a vertex $w \in V(D)$ such that $\I{H^{+w}}{w\D_1(H)}=\varnothing$, contradicting the fact that $J(\D_1,H)$ holds.

\subsubsection{Preliminary statements}

\noindent Since $H$ is not a trivial Maker win and $J(\D_1,H)$ holds, all results from Section \ref{Section4} apply to $D$. In particular:

\begin{proposition}
	$D$ has the following properties:
	 \begin{itemize}[noitemsep,nolistsep]
		\item $M(D)=\varnothing$. In particular, $m \not\in V(D)$.
		\item $\F_D=\T_z \cup \P_{zx}$.
		\item There is exactly one edge of $D$ that is incident to $x$: we call it $e_x$.
	\end{itemize}
\end{proposition}

\begin{proof}
	This is given by Propositions \ref{prop_nomarkedvertex} and \ref{prop_degree1}.
\end{proof}

\noindent The next two properties can be summed up as follows:
\begin{itemize}[noitemsep,nolistsep]
	\item[--] When following a chain starting from $m$, we cannot enter $D$ strictly before encountering $y$.
	\item[--] When following a chain starting from $x$ by an edge other than $e_x$, we cannot re-enter $D$ strictly before encountering $y$.
\end{itemize}

\begin{proposition}\label{prop_chemin1}
	Any $m$-chain $P_m$ in $H$ such that $V(P_m) \cap V(D) \neq \varnothing$ contains $y$.
\end{proposition}

\begin{proof}
	Since $J(\D_1,H)$ holds and $m \not\in V(D)$, Union Lemma \ref{Lemma5} with $c=m$ ensures that $D \cup P_m$ contains an $m$-tadpole or an $mx$-chain (i.e., an $xm$-snake). There cannot be an $m$-tadpole in $H$ according to Proposition \ref{prop_markedvertex}, therefore $D \cup P_m$ contains an $xm$-snake. Since $y \in \I{H^{+x}}{x\D_1(H)}$, that $xm$-snake must contain $y$, moreover $y \not\in V(D)$ by assumption so $y \in V(P_m)$.
\end{proof}

\begin{proposition}\label{prop_chemin2}
	Any $x$-chain $P_x$ in $H$ such that $\Start(\ora{xP_x})\neq e_x$ and $V(P_x) \cap (V(D)\setminus \{x\}) \neq \varnothing$ contains $y$.
\end{proposition}

\begin{proof}
	Since $J(\D_1,H)$ holds, we have $\I{H^{+z}}{z\D_1(H)} \neq \varnothing$, so Union Lemma \ref{Lemma6} ensures that $D \cup P_x$ contains an $x$-tadpole or an $x$-snake. In both cases, it contains $y$, moreover $y \not\in V(D)$ by assumption, so $y \in V(P_x)$.
\end{proof}

\noindent We now state a useful preliminary lemma:

\begin{lemme}\label{lemma_dist}
	Any $v \in V(D) \setminus \{x\}$ satisfies $\dist_H(v,m)\geq\dist_H(y,m)$, moreover:
	\begin{itemize}[noitemsep,nolistsep]
		\item If $\dist_H(v,m)>\dist_H(y,m)$, then there exists an $xm$-snake $S_{xm}^{\overline{v}}$ in $H$ that does not contain $v$.
		\item If $\dist_H(v,m)=\dist_H(y,m)$, then any shortest $vm$-snake $S_{vm}$ in $H$ satisfies $V(S_{vm}) \cap V(D)=\{v\}$ and $o(v,\ora{vS_{vm}m})=y$, moreover there are no $v$-tadpoles in $D$.
	\end{itemize}
\end{lemme}

\begin{proof}
	The fact that $\dist_H(v,m)\geq\dist_H(y,m)$ is a direct consequence of Proposition \ref{prop_chemin1}: since $v \in V(D)$, any $vm$-snake in $H$ contains $y$.
	\begin{itemize}[leftmargin=*]
		\item Suppose $\dist_H(v,m)>\dist_H(y,m)$.
			\\ Let $S_{ym}$ be a shortest $ym$-snake in $H$: note that there does exist one, since there exists an $xm$-snake by assumption, which must contain $y$ and therefore contains a $ym$-snake by Substructure Lemma \ref{lemma_subchain1}. Since $S_{ym}$ is shortest and $\dist_H(v,m)>\dist_H(y,m)$, we have $v \not\in V(S_{ym})$.
			\\ We necessarily have $v \not\in \I{H^{+x}}{x\D_1(H)}$, otherwise the fact that $\dist_H(v,m)>\dist_H(y,m)$ would contradict our choice of $y$. Since $v$ is non-marked (recall that $M(D)=\varnothing$) and distinct from $x$, this means there exists some $X \in x\D_1(H)$ such that $v \not\in V(X)$. On the other hand, since $y \in \I{H^{+x}}{x\D_1(H)}$, we have $y \in V(X)$ hence $y \in V(X) \cap V(S_{ym})$. This allows us, thanks to an adequate union lemma in $X \cup S_{ym}$, to find the desired $xm$-snake $S_{xm}^{\overline{v}}$ (ensuring that $v \not\in V(S_{xm}^{\overline{v}})$ since $v \not\in V(X) \cup V(S_{ym})$):
			\begin{itemize}[noitemsep,nolistsep]
				\item Suppose $X$ is an $x$-snake, and write $X=S$. If the marked vertex of $S$ is $m$, then we define $S_{xm}^{\overline{v}} = S$ as the desired $xm$-snake. Otherwise $m \not\in V(S)$, so apply Union Lemma \ref{Lemma4} with $a=x$, $S_{ab}=S$, $c=m$ and $P_c=S_{ym}$. Since $S \cup S_{ym} \subseteq H$ cannot contain an $m$-snake by Proposition \ref{prop_markedvertex}, it contains an $mx$-chain, i.e., an $xm$-snake $S_{xm}^{\overline{v}}$.
				\item Suppose $X$ is an $x$-tadpole, and write $X=T$. We have $M(T)=\varnothing$ hence $m \not\in V(T)$, so apply Union Lemma \ref{Lemma3} with $a=x$, $c=m$ and $P_c=S_{ym}$. Since $T \cup S_{ym} \subseteq H$ cannot contain an $m$-tadpole by Proposition \ref{prop_markedvertex}, it contains an $mx$-chain, i.e., an $xm$-snake $S_{xm}^{\overline{v}}$.
			\end{itemize}
		\item Suppose $\dist_H(v,m)=\dist_H(y,m)$.
			\\ Let $S_{vm}$ be a shortest $vm$-snake in $H$. Since $v \in V(D)$, we have $y \in V(S_{vm})$ by Proposition \ref{prop_chemin1}. If $y \neq o(v,\ora{vS_{vm}m})$ (see Figure \ref{LemmaDist}, top), then $S_{vm}$ contains a $ym$-snake that is shorter than $S_{vm}$: this is impossible since $\dist_H(v,m)=\dist_H(y,m)$ and $S_{vm}$ has been chosen shortest. Therefore $y = o(v,\ora{vS_{vm}m})$ (see Figure \ref{LemmaDist}, bottom). The $m$-chain $S_{vm}^{-y-v}$ does not contain $y$ and thus contains no vertices in $D$ by Proposition \ref{prop_chemin1}, hence why $V(S_{vm}) \cap V(D)=\{v\}$. Finally, there cannot be a $v$-tadpole $T$ in $D$, because $S_{vm} \cup T$ would then be an $m$-tadpole in $H$ since $V(S_{vm}) \cap V(T)=\{v\}$, contradicting Proposition \ref{prop_markedvertex}.\qedhere
			
			\begin{figure}[h]
				\centering
				\includegraphics[scale=.58]{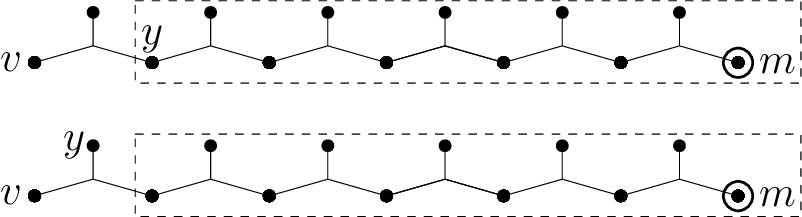}
				\caption{The snake $S_{vm}$ if $y \neq o(v,\ora{vS_{vm}m})$ (top, the contradictory $ym$-snake is highlighted) or if $y = o(v,\ora{vS_{vm}m})$ (bottom, the snake $S_{vm}^{-y-v}$ is highlighted).}\label{LemmaDist}
			\end{figure}

	\end{itemize}

\end{proof}

\noindent As we have often done in Section \ref{Section4}, we fix a vertex $s \in \I{H^{+z}}{z\D_1(H)}$, which exists since $J(\D_1,H)$ holds. Recall that $s \in V(D)\setminus \{x,z\}$ by Proposition \ref{prop_s}. Before we engage in the core of the proof, Table \ref{table_summary} summarizes the objects involved and some of their basic properties that will be used thereafter.

\begin{table}[h] \centering
{\renewcommand{\arraystretch}{1.25}
\begin{tabular}{|p{0.02\textwidth}|p{0.4\textwidth}|p{0.46\textwidth}|}
\hline
$H$   & $\cdot$ not a trivial Maker win & \multirow{13}{*}{\includegraphics[scale=.58]{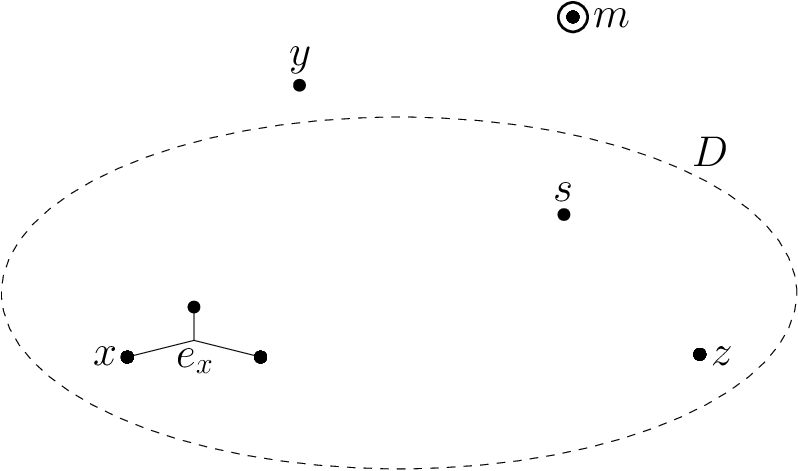}} \\
      & $\cdot$ $J(\D_1,H)$ holds & \\ \cline{1-2}
$x$   & $\cdot$ $x \in V(H) \setminus M(H)$ & \\ \cline{1-2}
$D$   & $\cdot$ $\widehat{\D_1}$-danger at $x$ in $H$ & \\
	  & $\cdot$ $M(D)=\varnothing$ & \\
      & $\cdot$ $\F_D=\T_z \cup \P_{zx}$ & \\ \cline{1-2}
$z$   & $\cdot$ from a decomposition $(z,\F_D)$ of $D$ & \\ \cline{1-2}
$m$   & $\cdot$ $m \in M(H)$ & \\
      & $\cdot$ $m \not\in V(D)$ & \\ \cline{1-2}
$y$   & $\cdot$ $y \in \I{H^{+x}}{x\D_1(H)}$ & \\
      & $\cdot$ $y \not\in V(D)$ & \\
      & $\cdot$ $y$ maximizes $\dist_H(\,\cdot\,,m)$ & \\ \cline{1-2}
$s$   & $\cdot$ $s \in \I{H^{+z}}{z\D_1(H)}$ & \\
      & $\cdot$ $s \in V(D)\setminus \{x,z\}$ & \\ \cline{1-2}
$e_x$ & $\cdot$ unique edge incident to $x$ in $D$ & \\ \hline
\end{tabular} \caption{The objects involved in the proof of Lemma \ref{lemma_main_structure}, and some of their properties.}\label{table_summary}
}
\end{table}

\subsubsection{Roadmap of the proof}

\noindent As previously stated, the idea of the proof is to eventually exhibit a vertex $w \in V(D)$ such that $\I{H^{+w}}{w\D_1(H)}=\varnothing$, contradicting the fact that $J(\D_1,H)$ holds. The roadmap to achieve this is given by the following result, which we will prove in this segment:

\begin{proposition}\label{prop_roadmap}
	Let $w \in V(D) \setminus \{x\}$. Suppose that $D$ contains the following three subhypergraphs:
	\begin{enumerate}[noitemsep,nolistsep,label={\textup{(\roman*)}}]
		\item a $z$-cycle $C$ containing $w$;
		\item an $xw$-chain $P_{xw}$ such that $V(P_{xw}) \cap \inn(C) \subseteq \{w\}$;
		\item a $w$-tadpole that does not contain $s$.
	\end{enumerate}
	Then $\I{H^{+w}}{w\D_1(H)}=\varnothing$.
\end{proposition}

\noindent Showing that $\I{H^{+w}}{w\D_1(H)}=\varnothing$ in the proof of Proposition \ref{prop_roadmap} will require the ability, for every non-marked vertex $v \neq w$, to exhibit a $\D_1$-danger at $w$ that does not contain $v$. The following lemma applied to $d=w$ gives us that object under certain conditions.

\begin{lemme}\label{lemma_exclusion}
	Let $d,v \in V(D) \setminus \{x\}$. Suppose that $\dist_H(v,m)>\dist_H(y,m)$ and that there exists a $dx$-chain $P_{dx}^{\overline{v}}$ in $D$ that does not contain $v$. Then there exists a $dm$-snake in $H$ that does not contain $v$.
\end{lemme}

\begin{proof}

	Suppose for a contradiction that:
	\begin{equation}\label{cont3}
		\text{All $dm$-snakes in $H$ contain $v$.}\tag{$\ast$}
	\end{equation}
	We are going to exhibit an $m$-tadpole in $H$, contradicting Proposition \ref{prop_markedvertex}. This $m$-tadpole will be obtained inside the union of an $xy$-chain and an $xm$-snake having specific properties, whose existence is given by the following two claims which we prove independently from each other. Define $t = o(x,\ora{xP_{dx}^{\overline{v}}d})$, and note that $\Start(\ora{xP_{dx}^{\overline{v}}d})=e_x$ since $e_x$ is the only edge incident to $x$ in $D$.
	
	\begin{claim}\label{claim_exclusion1}
		There exists an $xy$-chain $P_{xy}^{\overline{v}}$ in $H$ such that:
		\begin{itemize}[noitemsep,nolistsep]
			\item $V(P_{xy}^{\overline{v}}) \subseteq V(P_{dx}^{\overline{v}}) \cup \{y\}$.
			\item $o(x,\ora{xP_{xy}^{\overline{v}}y})=t$.
		\end{itemize}
	\end{claim}
	
	\begin{proofclaim}[Proof of Claim \ref{claim_exclusion1}]
		We have $\dist_H(v,m)>\dist_H(y,m)$ by assumption, so by Lemma \ref{lemma_dist} there exists an $xm$-snake $S_{xm}^{\overline{v}}$ in $H$ that does not contain $v$. Since $P_{dx}^{\overline{v}} \subseteq D$ and $m \not\in V(D)$, the edge $e^* = \End(\ora{m\proj{V(P_{dx}^{\overline{v}})}{m}{S_{xm}^{\overline{v}}}})$ is well defined. According to (\ref{cont3}), there are no $dm$-snakes in $P_{dx}^{\overline{v}} \cup S_{xm}^{\overline{v}}$. Therefore, by Union Lemma \ref{Lemma1} applied to $a=d$, $b=x$, $c=m$, $P_{ab}=P_{dx}^{\overline{v}}$ and $P_c=S_{xm}^{\overline{v}}$: $|e^* \cap V(P_{dx}^{\overline{v}})|=2$, $e^* \perp \ola{xP_{dx}^{\overline{v}}d}$, moreover there is an $x$-tadpole $T$ in $P_{dx}^{\overline{v}} \cup e^*$. Since $|e^* \cap V(P_{dx}^{\overline{v}})|=2$, there is exactly one vertex of $T$ that is not in $P_{dx}^{\overline{v}}$. That vertex is necessarily $y$, as illustrated in Figure \ref{ExclusionLemma8}: indeed, we know $y \in V(T)$ because $y \in \I{H^{+x}}{x\D_1(H)}$, and $y \not\in V(D) \supseteq V(P_{dx}^{\overline{v}})$. By Substructure Lemma \ref{lemma_subchain4bis}, $T$ contains an $xy$-chain $P_{xy}^{\overline{v}}$, and we have $V(P_{xy}^{\overline{v}}) \subseteq V(T) \subseteq V(P_{dx}^{\overline{v}}) \cup \{y\}$. See Figure \ref{ExclusionLemma8}.
		
		\begin{figure}[h]
			\centering
			\includegraphics[scale=.58]{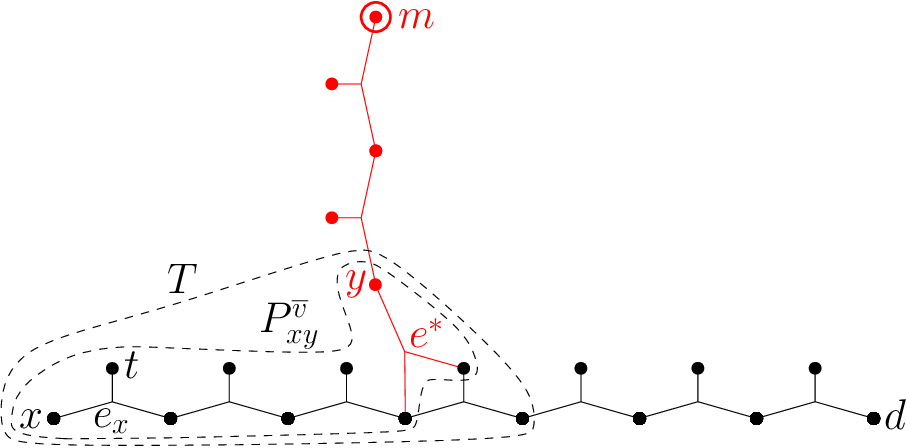}
			\caption{Definition of $P_{xy}^{\overline{v}}$. The represented chains are $P_{dx}^{\overline{v}}$ (horizontal) and $\proj{V(P_{dx}^{\overline{v}})}{m}{S_{xm}^{\overline{v}}}$ (vertical).}\label{ExclusionLemma8}
		\end{figure}	
		
		\noindent Finally, let us check that $o(x,\ora{xP_{xy}^{\overline{v}}y})=t$. Since $|e^* \cap V(P_{dx}^{\overline{v}})|=2$ and $e^* \perp \ola{xP_{dx}^{\overline{v}}d}$, there are two possibilities:
		\begin{itemize}[noitemsep,nolistsep]
			\item[--] First case: $\{x,t\} \subseteq e^*$. Then $P_{xy}^{\overline{v}}$ consists of the single edge $e^*=\{x,t,y\}$, so obviously $o(x,\ora{xP_{xy}^{\overline{v}}y})=t$.
			\item[--] Second case: $\{x,t\} \cap e^* = \varnothing$. Then $\Start(\ora{xP_{xy}^{\overline{v}}y}) \neq e^*$, so $\Start(\ora{xP_{xy}^{\overline{v}}y})=\Start(\ora{xP_{dx}^{\overline{v}}d})=e_x \ni t$. Moreover $t$ is of degree 1 in $P_{dx}^{\overline{v}} \cup e^* \supseteq P_{xy}^{\overline{v}}$, so necessarily $t \not\in \inn(P_{xy}^{\overline{v}})$ hence $t=o(x,\ora{xP_{xy}^{\overline{v}}y})$. \qedhere
		\end{itemize}
	\end{proofclaim}
	
	\begin{claim}\label{claim_exclusion2}
		There exists an $xm$-snake $S_{xm}^{\overline{t}}$ in $H$ such that:
		\begin{itemize}[noitemsep,nolistsep]
			\item $t \not\in V(S_{xm}^{\overline{t}})$.
			\item $y \in \inn(S_{xm}^{\overline{t}})$. (We define $P_{xy}^{\overline{t}}$, resp. $S_{ym}^{\overline{t}}$, as the unique $xy$-chain, resp. $ym$-snake, in $S_{xm}^{\overline{t}}$.)
			\item $(V(S_{xm}^{\overline{t}}) \cap V(P_{dx}^{\overline{v}})) \setminus \{x\} \subseteq \{u\}$ where we have defined $u =o(y,\ola{xP_{xy}^{\overline{t}}y})$.
		\end{itemize}
	\end{claim}
	
	\begin{proofclaim}[Proof of Claim \ref{claim_exclusion2}]
		It is impossible that $\dist_H(t,m)=\dist_H(y,m)$: indeed, a shortest $tm$-snake $S_{tm}$ in $H$ would then satisfy $V(S_{tm}) \cap V(D)=\{t\}$ by Lemma \ref{lemma_dist}, hence $v \not\in V(S_{tm})$, so the walk $\ora{dP_{dx}^{\overline{v}}t} \oplus \ora{tS_{tm}m}$ would represent a $dm$-snake in $H$ that does not contain $v$, contradicting (\ref{cont3}). Therefore, Lemma \ref{lemma_dist} ensures that $\dist_H(t,m)>\dist_H(y,m)$, and that there exists an $xm$-snake $S_{xm}^{\overline{t}}$ in $H$ such that $t \not\in V(S_{xm}^{\overline{t}})$. Since $y \in \I{H^{+x}}{x\D_1(H)}$, we obviously have $y \in V(S_{xm}^{\overline{t}})$.
		\\ Write $\ora{xS_{xm}^{\overline{t}}m}=(x,e_1,\ldots,e_L,m)$. Recalling Notation \ref{notation_walk}, write $\ora{xS_{xm}^{\overline{t}}m}\vert_{\{y\}}=(x,e_1,\ldots,e_i)$ and $\ola{xS_{xm}^{\overline{t}}m}\vert_{\{y\}}=(m,e_L,e_{L-1},\ldots,e_j)$. Note that $j=i+1$ if $y \in\inn(S_{xm}^{\overline{t}})$ and $j=i$ otherwise.
		\begin{itemize}[noitemsep,nolistsep]
			\item[--] By definition of the walk $\ola{xS_{xm}^{\overline{t}}m}\vert_{\{y\}}$, the $m$-chain $\HH{(m,e_L,e_{L-1},\ldots,e_{j+1})}$ does not contain $y$. Therefore, $(e_{j+1} \cup \ldots \cup e_L) \cap V(D)=\varnothing$ by Proposition \ref{prop_chemin1}.
			\item[--] By definition of the walk $\ora{xS_{xm}^{\overline{t}}m}\vert_{\{y\}}$, the $x$-chain $\HH{(x,e_1,\ldots,e_{i-1})}$ does not contain $y$. Moreover, $e_1 \neq e_x$ because $t \not\in V(S_{xm}^{\overline{t}}) \supseteq e_1$. Therefore, $(e_1 \cup \ldots \cup e_{i-1}) \cap (V(D) \setminus \{x\})=\varnothing$ by Proposition \ref{prop_chemin2}.
		\end{itemize}
		Suppose that $y \not\in\inn(S_{xm}^{\overline{t}})$: then, $i=j$ hence $V(S_{xm}^{\overline{t}})=e_1 \cup \ldots \cup e_{i-1} \cup \{y\} \cup e_{j+1} \cup \ldots \cup e_L$. By the above, this yields $V(S_{xm}^{\overline{t}}) \cap V(D)=\{x\}$ and in particular $v \not\in V(S_{xm}^{\overline{t}})$, therefore the walk $\ora{dP_{dx}^{\overline{v}}x} \oplus \ora{xS_{xm}^{\overline{t}}m}$ represents a $dm$-snake that does not contain $v$. This contradicts (\ref{cont3}).
		\\ Therefore, we have $y \in \inn(S_{xm}^{\overline{t}})$. Let $P_{xy}^{\overline{t}}$ (resp. $S_{ym}^{\overline{t}}$) be the unique $xy$-chain (resp. $ym$-snake) in $S_{xm}^{\overline{t}}$, and define $u = o(y,\ola{xP_{xy}^{\overline{t}}y})$ and $u'= o(y,\ora{yS_{ym}^{\overline{t}}m})$, as in Figure \ref{ExclusionLemma9}. Since $y \in \inn(S_{xm}^{\overline{t}})$, we have $j=i+1$ hence $V(S_{xm}^{\overline{t}})=e_1 \cup \ldots \cup e_{i-1} \cup \{u,y,u'\} \cup e_{j+1} \cup \ldots \cup e_L$. By the above, this yields $(V(S_{xm}^{\overline{t}}) \cap V(D)) \setminus \{x\} \subseteq \{u,u'\}$, hence $(V(S_{xm}^{\overline{t}}) \cap V(P_{dx}^{\overline{v}})) \setminus \{x\} \subseteq \{u,u'\}$.
		
		\begin{figure}[h]
			\centering
			\includegraphics[scale=.58]{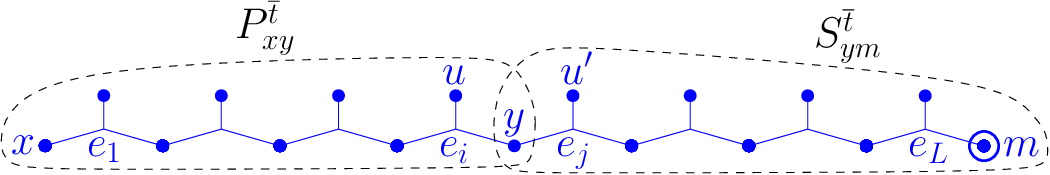}
			\caption{The $xm$-snake $S_{xm}^{\overline{t}}$.}\label{ExclusionLemma9}
		\end{figure}
		
		\noindent Finally, it is impossible that $u' \in V(P_{dx}^{\overline{v}})$: indeed, this would imply that $V(S_{ym}^{\overline{t}}) \cap V(P_{dx}^{\overline{v}}) =\{u'\}$ and that $u' \neq v$ hence $v \not\in V(S_{ym}^{\overline{t}})$, so the walk $\ora{dP_{dx}^{\overline{v}}x}\vert_{\{u'\}} \oplus \ora{u'S_{ym}^{\overline{t}}m}$ would represent a $dm$-snake not containing $v$, contradicting (\ref{cont3}). Therefore $(V(S_{xm}^{\overline{t}}) \cap V(P_{dx}^{\overline{v}})) \setminus \{x\} \subseteq \{u\}$, which concludes the proof of the claim.
	\end{proofclaim}
	
	\noindent We can now conclude by exhibiting an $m$-tadpole in $H$, which contradicts Proposition \ref{prop_markedvertex} since $H$ is not a trivial Maker win and $J(\D_1,H)$ holds.
	\\ Let $P_{xy}^{\overline{v}}$ be as in Claim \ref{claim_exclusion1}, and let $S_{xm}^{\overline{t}}$, $P_{xy}^{\overline{t}}$, $S_{ym}^{\overline{t}}$, $u$ be as in Claim \ref{claim_exclusion2}. We have $V(P_{xy}^{\overline{v}}) \subseteq V(P_{dx}^{\overline{v}}) \cup \{y\}$ by Claim \ref{claim_exclusion1}, and $(V(S_{xm}^{\overline{t}}) \cap V(P_{dx}^{\overline{v}})) \setminus \{x\} \subseteq \{u\}$ by Claim \ref{claim_exclusion2}: therefore, $V(S_{xm}^{\overline{t}}) \cap V(P_{xy}^{\overline{v}}) = \{x,y\}$ or $V(S_{xm}^{\overline{t}}) \cap V(P_{xy}^{\overline{v}}) = \{x,y,u\}$.
	
	\begin{itemize}[noitemsep,nolistsep]
	
		\item Case 1: $V(S_{xm}^{\overline{t}}) \cap V(P_{xy}^{\overline{v}}) = \{x,y\}$. The walk $\ola{yS_{ym}^{\overline{t}}m} \oplus \ola{xP_{xy}^{\overline{t}}y} \oplus \ora{xP_{xy}^{\overline{v}}y}$ clearly represents an $m$-tadpole (see Figure \ref{ExclusionLemma10}, top).
		
		\begin{figure}[h]
			\centering
			\includegraphics[scale=.58]{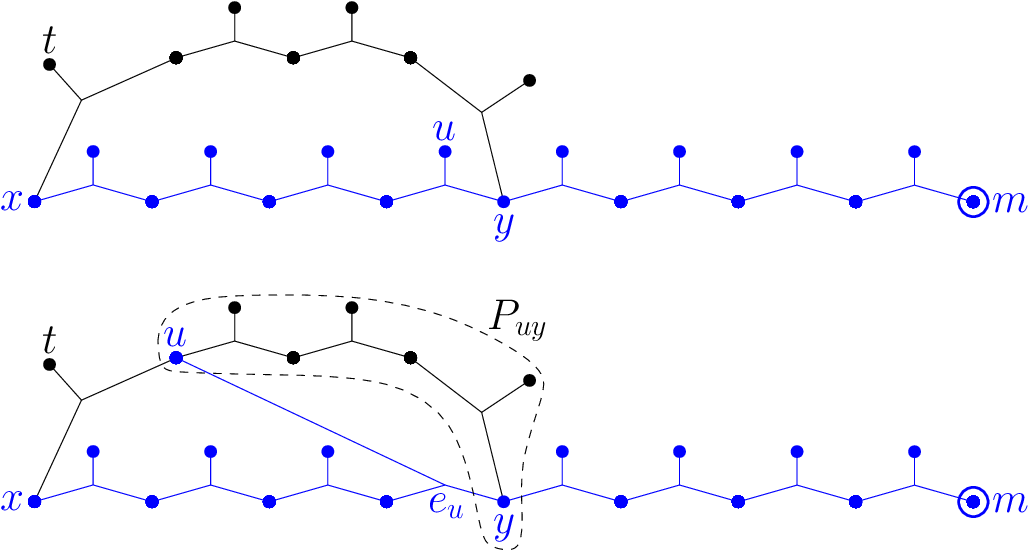}
			\caption{Conclusion of Lemma \ref{lemma_exclusion}. The represented chains are $S_{xm}^{\overline{t}}$ (horizontal) and $P_{xy}^{\overline{v}}$. Top: Case 1. Bottom: Case 2.}\label{ExclusionLemma10}
		\end{figure}
	
		\item Case 2: $V(S_{xm}^{\overline{t}}) \cap V(P_{xy}^{\overline{v}}) = \{x,y,u\}$. Let $e_u$ be the edge of $S_{xm}^{\overline{t}}$ containing $u$, and let $P_{uy}$ be the unique $uy$-chain in $P_{xy}^{\overline{v}}$. Since $u \neq t$, we have $x \not\in V(P_{uy})$, therefore the walk $\ola{yS_{ym}^{\overline{t}}m} \oplus (y,e_u,u) \oplus \ora{uP_{uy}y}$ represents an $m$-tadpole (see Figure \ref{ExclusionLemma10}, bottom). \qedhere
		
	\end{itemize}
	
\end{proof}

\begin{corollaire}\label{coro_s}
	We have $\dist_H(s,m)=\dist_H(y,m)$.
\end{corollaire}
 
\begin{proof}
	Since all $z$-snakes in $H$ contain $s$, we can apply the contrapositive of Lemma \ref{lemma_exclusion} with $d=z$ and $v=s$ (recall that $s \not\in\{x,z\}$, so that Lemma \ref{lemma_exclusion} indeed applies). We get that $\dist_H(s,m) \leq \dist_H(y,m)$ or all $zx$-chains in $D$ contain $s$. We know the latter is false: a counterexample chain $\Ps$ is given by Proposition \ref{prop_s}. Therefore, we conclude that $\dist_H(s,m) \leq \dist_H(y,m)$, hence $\dist_H(s,m) = \dist_H(y,m)$ by Lemma \ref{lemma_dist}.
\end{proof}

\noindent The previous corollary has a simple consequence which we will use extensively:

\begin{proposition}\label{prop_s_bis}
	There are no $s$-tadpoles in $D$. In particular, any $z$-tadpole $T$ in $D$ satisfies $s \in \out(C_T)$.
\end{proposition}

\begin{proof}
	We have $\dist_H(s,m) = \dist_H(y,m)$ by Corollary \ref{coro_s}, so there are no $s$-tadpoles in $D$ according to Lemma \ref{lemma_dist}. Let $T$ be a $z$-tadpole in $D$: we know $s \in V(T)$ by definition of $s$. If we had $s \not\in \out(C_T)$, then there would be an $s$-tadpole in $T \subseteq D$ by Substructure Lemma \ref{lemma_subtadpole}, therefore $s \in \out(C_T)$.
\end{proof}

\noindent For example, one application of the previous proposition is the following:

\begin{proposition}\label{prop_nodiamond}
	There are no $z$-cycles of length 2 in $D$.
\end{proposition}

\begin{proof}
	Suppose for a contradiction that there exists a $z$-cycle $C$ of length 2 in $D$. We have $s \in \out(C)$ by Proposition \ref{prop_s_bis}: write $V(C)=\{z,s,a,b\}$ and $E(C)=\{e_1,e_2\}$ where $e_1=\{z,a,s\}$ and $e_2=\{z,a,b\}$. By Proposition \ref{prop_structure}\ref{item2}, we know there exists some $X^{\overline{a}} \in \F_D$ such that $a \not\in V(X^{\overline{a}})$.
	\begin{itemize}[wide,noitemsep,nolistsep]
		\item First suppose $s \in V(X^{\overline{a}})$. By Substructure Lemma \ref{lemma_subchain1} (if $X^{\overline{a}} \in \P_{zx}$) or Substructure Lemma \ref{lemma_subchain4bis} (if $X^{\overline{a}} \in \T_z$), there exists a $zs$-chain $P_{zs}$ in $X^{\overline{a}}$. Since $a \not\in V(P_{zs})$, we get an $s$-cycle $P_{zs} \cup e_1$ in $D$, contradicting Proposition \ref{prop_s_bis}.
		\item Now suppose $s \not\in V(X^{\overline{a}})$. Since $s \in \I{H^{+z}}{z\D_1(H)}$, this implies that $X^{\overline{a}} \in \P_{zx}$. Write $X^{\overline{a}} = P$. If $b \not\in V(P)$, then $V(C) \cap V(P)=\{z\}$, therefore $P \cup C$ is an $x$-tadpole in $D$, contradicting Proposition \ref{prop_structure}\ref{item6}. If $b \in V(P)$, then there exists a $zb$-chain $P_{zb}$ in $P$ by Substructure Lemma \ref{lemma_subchain1}, and we get a $z$-cycle $P_{zb} \cup e_2$ in $D$ that does not contain $s$, also a contradiction. \qedhere
	\end{itemize}
\end{proof}

\noindent We can now prove Proposition \ref{prop_roadmap}.

\begin{proof}[Proof of Proposition \ref{prop_roadmap}]
	Let $w \in V(D)\setminus \{x\}$ be such that $D$ contains the following three subhypergraphs:
	\begin{enumerate}[noitemsep,nolistsep,label={\textup{(\roman*)}}]
		\item a $z$-cycle $C$ containing $w$;
		\item an $xw$-chain $P_{xw}$ such that $V(P_{xw}) \cap \inn(C) \subseteq \{w\}$;
		\item a $w$-tadpole $T$ that does not contain $s$.
	\end{enumerate}
	We are going to consider $C$, $P_{xw}$ and $T$ successively. Each of these three objects will imply the existence of some $\D_1$-dangers at $w$, which will improve our upper bound on $\I{H^{+w}}{w\D_1(H)}$ (with respect to set inclusion) until we get the desired conclusion that $\I{H^{+w}}{w\D_1(H)}=\varnothing$. 
	
	\begin{enumerate}[wide,label=\textbf{\arabic*)}]
	
		\item Step 1: we show that $\I{H^{+w}}{w\D_1(H)} \subseteq \inn(C) \cup \{s\} \cup (V(H) \setminus V(D))$.
			\\ In this step, we use $C$. Recall that $s \in \out(C)$ by Proposition \ref{prop_s_bis} and that $C$ is of length at least 3 by Proposition \ref{prop_nodiamond}.
			\begin{claim}\label{claim_roadmap}
				For all $ws$-chains $P_{ws}$ in $C$: $\I{H^{+w}}{w\D_1(H)} \subseteq (V(P_{ws}) \setminus \{w\}) \cup (V(H) \setminus V(D))$.
			\end{claim}
			\begin{proofclaim}[Proof of Claim \ref{claim_roadmap}]
				We know $\dist_H(s,m)=\dist_H(y,m)$ by Corollary \ref{coro_s}. Let $S_{sm}$ be a shortest $sm$-snake in $H$: Lemma \ref{lemma_dist} thus ensures that $V(S_{sm}) \cap V(D)=\{s\}$ hence $V(S_{sm}) \cap V(C)=\{s\}$. Therefore, any $ws$-chain $P_{ws}$ in $C$ yields a $wm$-snake $S_{wm} = P_{ws} \cup S_{sm}$ in $H$ and:
				\begin{align*}
					\I{H^{+w}}{w\D_1(H)} & \subseteq V(S_{wm}) \setminus \{w\} \\
							   			 & \subseteq (V(P_{ws}) \setminus \{w\}) \cup (V(S_{sm}) \setminus \{s\}) \\
							   			 & \subseteq (V(P_{ws}) \setminus \{w\}) \cup (V(H) \setminus V(D)). \qedhere
				\end{align*}
			\end{proofclaim}
			We know there exists a $ws$-chain in $C$ by Substructure Lemma \ref{lemma_subchain3bis}, so Claim \ref{claim_roadmap} ensures that $\I{H^{+w}}{w\D_1(H)} \subseteq (V(C) \setminus \{w\}) \cup (V(H) \setminus V(D))$. Moreover, for all $v \in \out(C) \setminus \{w,s\}$, there exists a $ws$-chain in $C$ that does not contain $v$ by Substructure Lemma \ref{lemma_subchain3}, so $v \not\in \I{H^{+w}}{w\D_1(H)}$ by Claim \ref{claim_roadmap}. All in all, we get $\I{H^{+w}}{w\D_1(H)} \subseteq \inn(C) \cup \{s\} \cup (V(H) \setminus V(D))$, which concludes Step 1.
			
		\item Step 2: we show that $\I{H^{+w}}{w\D_1(H)} \subseteq \{s\} \cup (V(H) \setminus V(D))$.
			\\ In this step, we use $P_{xw}$. Comparing with Step 1, we need to show that $\I{H^{+w}}{w\D_1(H)}$ is disjoint from $\inn(C)$. Let $v \in \inn(C)$. If $v=w$, then obviously $v \not\in \I{H^{+w}}{w\D_1(H)}$, so assume $v \neq w$. By definition of $P_{xw}$, we have $V(P_{xw}) \cap \inn(C) \subseteq \{w\}$, so $v \not\in V(P_{xw})$. Lemma \ref{lemma_exclusion} thus applies with: $d=w$, our vertex $v$, and $P_{dx}=P_{xw}$. We get a $wm$-snake in $H$ that does not contain $v$, hence $v \not\in \I{H^{+w}}{w\D_1(H)}$, which concludes Step 2.
			
		\item Step 3: we show that $\I{H^{+w}}{w\D_1(H)} = \varnothing$.
			\\ In this step, we use $T$. We already know that $\I{H^{+w}}{w\D_1(H)} \subseteq \{s\} \cup (V(H) \setminus V(D))$. Moreover, $\I{H^{+w}}{w\D_1(H)} \subseteq V(T)$ because $T$ is a $w$-tadpole, where $V(T)$ is disjoint from $\{s\} \cup (V(H) \setminus V(D))$ by definition. In conclusion, $\I{H^{+w}}{w\D_1(H)} = \varnothing$. \qedhere
		
	\end{enumerate}

\end{proof}

\noindent Our goal is now to show that, for a suitable vertex $w$, $D$ contains all three subhypergraphs listed in Proposition \ref{prop_roadmap}. A lot of the work has already been done through Proposition \ref{prop_insidestructure}: we now separate the case where $D$ is of type (1) from the case where $D$ is of type (2).

\subsubsection{Finishing the proof when $D$ is of type (2)}

\noindent We first suppose that $D$ is of type (2). By definition (recall Proposition \ref{prop_insidestructure}), this means $D$ contains the following two subhypergraphs:
\begin{itemize}[noitemsep,nolistsep]
	\item a $z$-cycle $C$ such that $x \not\in V(C)$;
	\item an $xw$-chain $P_{xw}$, for some $w \in V(C)$, such that $V(P_{xw}) \cap V(C)=\{w,w'\}$ where we have defined $w' =o(w,\ola{xP_{xw}w})$.
\end{itemize}
Note that $C$ is of length at least 3 by Proposition \ref{prop_nodiamond}, and that $s \in \out(C)$ by Proposition \ref{prop_s_bis}. Define $e^* = \End(\ora{xP_{xw}w})$.
\\ Also note that $s \not\in \{w,w'\}$: indeed, let $P$ be a $ww'$-chain in $C$ (which exists by Substructure Lemma \ref{lemma_subchain3bis}), if we had $s \in \{w,w'\}$, then $P \cup e^*$ would be an $s$-cycle in $D$, contradicting Proposition \ref{prop_s_bis}. Therefore, since $s \in \out(C)$, Substructure Lemma \ref{lemma_subchain3} ensures that there exists a unique $ww'$-chain $P_{ww'}$ in $C$ that does not contain $s$.
\\ Define $C'= P_{ww'} \cup e^*$: $C'$ is both a $w$-cycle and a $w'$-cycle in $D$, and it does not contain $s$. Moreover, we have $z\not\in V(P_{ww'})$: indeed, we would otherwise have $z \in \inn(C) \cap V(P_{ww'})=\{w,w'\} \cup \inn(P_{ww'}) = \inn(C')$, so $C'$ would be a $z$-cycle not containing $s$, contradicting Proposition \ref{prop_s_bis}. See Figure \ref{Type2-2}.

\begin{figure}[h]
	\centering
	\includegraphics[scale=.58]{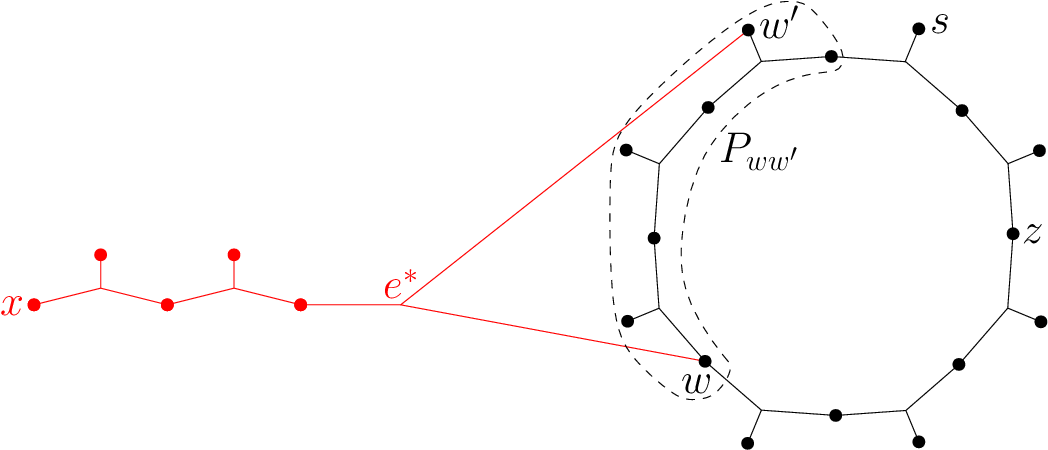}
	\caption{The cycle $C$ (on the far right) and the chain $P_{xw}$. In this drawing, we have $w \in \inn(C)$ and $w' \in \out(C)$.}\label{Type2-2}
\end{figure}

\begin{claim}\label{claim_Type2}
	$w \in \out(C) \text{ or } w' \in \out(C).$
\end{claim}
\begin{proofclaim}[Proof of Claim \ref{claim_Type2}]
	Suppose for a contradiction that $w,w' \in \inn(C)$. Write $\ora{zC}=(z,e_1,\ldots,e_L,z)$. Since $L \geq 3$ and $w,w' \in \inn(C) \setminus \{z\}$, there exist $1 \leq i \neq i' \leq L-1$ such that $e_i \cap e_{i+1}=\{w\}$ and $e_{i'} \cap e_{i'+1}=\{w'\}$. Since $w$ and $w'$ have symmetrical roles, assume $i<i'$. Let $1 \leq j \leq L$ be the unique index such that $s \in e_j$.
\\ Since $e_1$ and $e_L$ are the only edges of $C$ containing $z$, the $ww'$-chain represented by the walk $(w, e_{i+1}, ..., e_{i'}, w')$ does not contain $z$, so it is necessarily $P_{ww'}$ according to the uniqueness statement of Substructure Lemma \ref{lemma_subchain3}. Since $s \not\in V(P_{ww'})$, this yields $1 \leq j \leq i$ or $i'+1 \leq j \leq L$: by symmetry, assume $i'+1 \leq j \leq L$. Then  $(z,e_1,\ldots,e_i,w) \oplus \ora{wP_{ww'}w'} \oplus (w',e^*,w)$ represents a $z$-tadpole not containing $e_j$, i.e., not containing $s$, a contradiction which concludes the proof of the claim.
\end{proofclaim}
\noindent Using Claim \ref{claim_Type2}, assume $w' \in \out(C)$ by symmetry. This ensures that $V(P_{xw}) \cap \inn(C) \subseteq \{w\}$. 
\\ In conclusion, we can apply Proposition \ref{prop_roadmap} to the vertex $w$, with: the $z$-cycle $C$ containing $w$, the $xw$-chain $P_{xw}$ which satisfies $V(P_{xw}) \cap \inn(C) \subseteq \{w\}$, and the $w$-cycle $C'$ which does not contain $s$. We get $\I{H^{+w}}{w\D_1(H)}=\varnothing$, contradicting property $J(\D_1,H)$. This ends the proof of Lemma \ref{lemma_main_structure} when $D$ is of type (2).

\subsubsection{Finishing the proof when $D$ is of type (1)}

\noindent We now suppose that $D$ is of type (1). By definition (see Proposition \ref{prop_insidestructure}), this means $D$ contains the following three subhypergraphs:
\begin{itemize}[noitemsep,nolistsep]
	\item a $z$-cycle $C$ such that $x \not\in V(C)$;
	\item an $xw$-chain $P_{xw}$ for some $w \in \out(C)$ such that $V(P_{xw}) \cap V(C)=\{w\}$;
	\item some $X \in \F_D$ such that $V(X) \cap V(P_{xw}) \neq \varnothing$ and $\{w_1,w_2\} \not\subseteq V(X)$ where $e=\{w,w_1,w_2\}$ denotes the unique edge of $C$ containing $w$.
\end{itemize}
Note that $C$ is of length at least 3 by Proposition \ref{prop_nodiamond} and that $s \in \out(C)$ by Proposition \ref{prop_s_bis}.
\\ Since $C$ contains $w$ and $V(P_{xw}) \cap \inn(C) = \varnothing \subseteq \{w\}$, the only subhypergraph in $D$ that we are missing to apply Proposition \ref{prop_roadmap} is a $w$-tadpole that does not contain $s$. The rest of the proof consists in finding a $w$-cycle in $D$ that does not contain $s$.

\begin{claim}\label{claim_Type1}
	There exists a $w$-chain $P_w$ in $D$ such that:
	\begin{enumerate}[noitemsep,nolistsep,label={\textup{(\alph*)}}]
		\item The only edge of $P_w$ that intersects $V(C) \setminus \{w\}$ is $e^* = \End(\ora{wP_w})$. In particular: $|V(P_w) \cap (V(C) \setminus \{w\})| = |e^* \cap (V(C) \setminus \{w\})| \in \{1,2\}$. \label{item_a}
		\item $\{w_1,w_2\} \not\subseteq V(P_w)$. \label{item_b}
		\item $s \not\in V(P_w)$. \label{item_c}
	\end{enumerate}
\end{claim}

\begin{proofclaim}[Proof of Claim \ref{claim_Type1}]

	Since $V(X) \cap V(P_{xw}) \neq \varnothing$, the projection $\proj{V(P_{xw})}{z}{X}$ is well defined. There are no $x$-tadpoles in $D \supseteq P_{xw} \cup \proj{V(P_{xw})}{z}{X}$ by Proposition \ref{prop_structure}\ref{item6}, so Union Lemma \ref{Lemma1} with $a=w$, $b=x$ and $c=z$ ensures that $P_{xw} \cup \proj{V(P_{xw})}{z}{X}$ contains a $wz$-chain $P_{wz}$. Now, since $V(P_{wz}) \cap (V(C) \setminus \{w\}) \supseteq \{z\} \neq \varnothing$, the projection $P_w = \proj{V(C) \setminus \{w\}}{w}{P_{wz}}$ is well defined.
	\\ Define $e^* =\End(\ora{wP_w})$: by definition of a projection, the edge $e^*$ is the only edge of $P_w$ that intersects $V(C) \setminus \{w\}$, and $|e^* \cap (V(C) \setminus \{w\})| \in \{1,2\}$, hence item \ref{item_a}.
	\\ Since $P_w \subseteq P_{wz} \subseteq X \cup P_{xw}$, where $\{w_1,w_2\} \not\subseteq V(X)$ by definition of $X$ and $\{w_1,w_2\}\cap V(P_{xw})=\varnothing$, we have $\{w_1,w_2\} \not\subseteq V(P_w)$, i.e., item \ref{item_b}.
	\\ Finally, suppose for a contradiction that $s \in V(P_w)$:
	\begin{itemize}[leftmargin=*]
		\item First suppose $s=w$. By Substructure Lemma \ref{lemma_subchain2}, $C^{-s}$ is a $w_1w_2$-chain. By Union Lemma \ref{Lemma1} with $a=w_1$, $b=w_2$ and $c=w=s$, $C^{-s} \cup P_w$ contains an $sw_1$-chain or an $sw_2$-chain. Take a shortest $sw_1$-chain or $sw_2$-chain in $C^{-s} \cup P_w$: by symmetry, assume it is an $sw_1$-chain $P_{sw_1}$. The minimality of the length ensures that either $w_2 \not\in V(P_{sw_1})$ or $w_2 = o(w_1,\ola{sP_{sw_1}w_1})$.
			\begin{itemize}[noitemsep,nolistsep]
				\item If $w_2 \not\in V(P_{sw_1})$, then $P_{sw_1} \cup e$ is an $s$-cycle in $D$, contradicting Proposition \ref{prop_s_bis}.
				\item If $w_2 = o(w_1,\ola{sP_{sw_1}w_1})$, then $\End(\ora{sP_{sw_1}w_1})$ is an edge of $C^{-s} \cup P_w$ containing both $w_1$ and $w_2$. However, there are no such edges in $C^{-s}$ because $C^{-s}$ is a $w_1w_2$-chain of length at least 2 (indeed, recall that $C$ is of length at least 3), and there are no such edges in $P_w$ because $\{w_1,w_2\} \not\subseteq V(P_w)$ by item \ref{item_b}. We have a contradiction.
			\end{itemize}
		\item Now suppose $s \neq w$. By item \ref{item_a}, $e^* = \End(\ora{wP_w})$ is the only edge of $P_w$ that intersects $V(C) \setminus \{w\}$, so $P_w$ is a $ws$-chain and either $V(P_w) \cap V(C)=\{w,s\}$ or $V(P_w) \cap V(C) = \{w,s,o(s,\ora{sP_ww})\}$.
			\begin{itemize}[noitemsep,nolistsep]
				\item If $V(P_w) \cap V(C)=\{w,s\}$, then let $P_{ws}$ be a $ws$-chain in $C$ (which exists by Substructure Lemma \ref{lemma_subchain3bis}): since $P_w$ and $P_{ws}$ are both $ws$-chains and $V(P_w) \cap V(P_{ws})=\{w,s\}$, we get an $s$-cycle $P_w \cup P_{ws}$, contradicting Proposition \ref{prop_s_bis}. 
				\item If $V(P_w) \cap V(C) = \{w,s,t\}$, where we have defined $t =o(s,\ora{sP_ww}) \in e^*$, then let $P_{st}$ be an $st$-chain in $C$ that does not contain $w$ (which exists by Substructure Lemma \ref{lemma_subchain3} because $w \in \out(C)$): since $V(P_{st}) \cap e^* = \{s,t\}$, we get an $s$-cycle $P_{st} \cup e^*$, contradicting Proposition \ref{prop_s_bis}.
			\end{itemize}
	\end{itemize}
	We have a contradiction in all cases, hence item \ref{item_c}.
\end{proofclaim}

\noindent From now on, the action takes place in $C \cup P_w$ exclusively: we are going to exhibit a $w$-cycle in $C \cup P_w$ that does not contain $s$. The idea is simply to get such a cycle by using $P_w$ to go from $w$ to $C$ and then rejoining $w$ by rotating along $C$ in the correct direction so as to avoid $s$ (for instance, see Figure \ref{Type1-1}, left and middle). This is always possible, unless this direction is blocked by a cycle of length 2, which cannot happen because there would then be an $s$-tadpole, contradicting Proposition \ref{prop_s_bis} (for instance, see Figure \ref{Type1-1}, right). We now carry out the rigorous proof of this, distinguishing between two cases.

\begin{figure}[h]
	\centering
	\includegraphics[scale=.58]{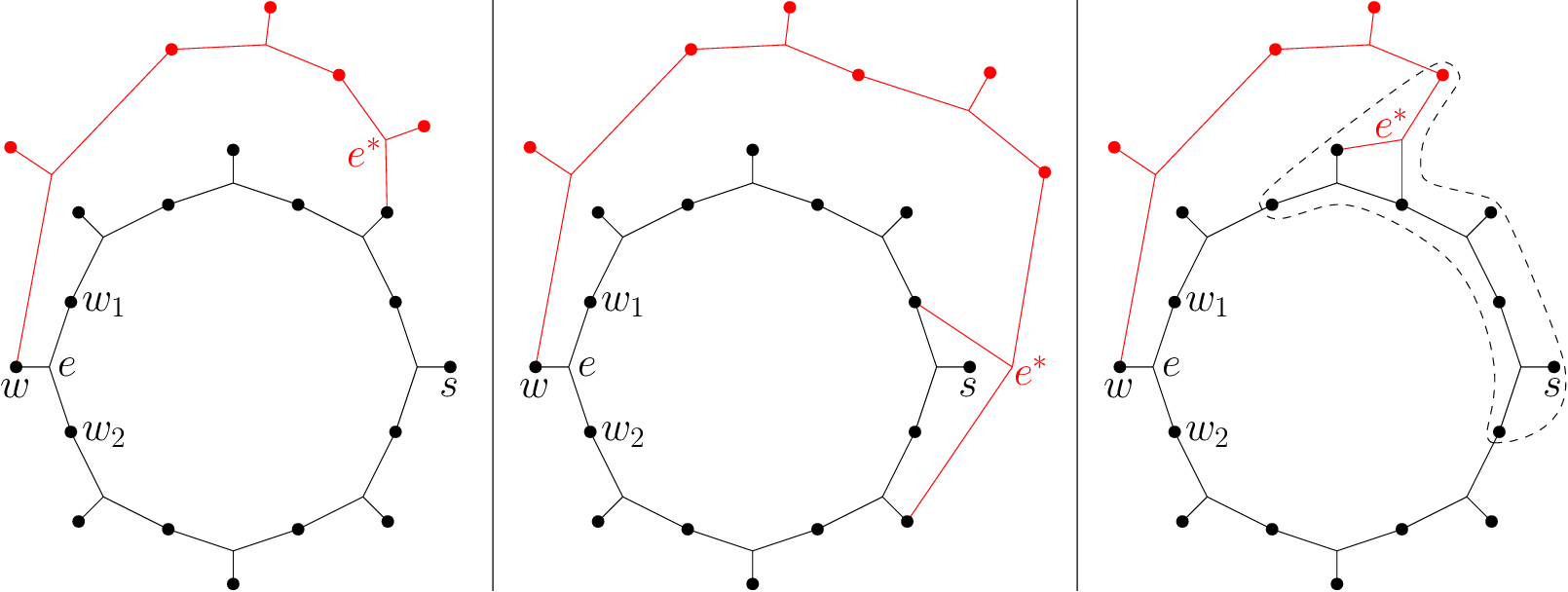}
	\caption{Represented here are $C$ and $P_w$. In the left and middle examples, there is a $w$-cycle not containing $s$. In the right example, there are none but there is an $s$-tadpole (highlighted).}\label{Type1-1}
\end{figure}

\begin{enumerate}[wide,label=\textbf{\arabic*)}]

	\item Case 1: $w_1 \in V(P_w)$ or $w_2 \in V(P_w)$.
		\\ By symmetry, assume $w_1 \in V(P_w)$. By Claim \ref{claim_Type1}\ref{item_b}, we have $\{w_1,w_2\} \not\subseteq V(P_w)$ hence $w_2 \not\in V(P_w)$. Therefore, $P_w$ is a $ww_1$-chain that does not contain $w_2$, so, defining $C' = P_w \cup e$, $C'$ is a $w$-cycle. Moreover, $s \not\in V(C')=V(P_w) \cup \{w_2\}$: indeed, we have $s \not\in V(P_w)$ by Claim \ref{claim_Type1}\ref{item_c}, and $s \in \out(C)$ whereas $w_2 \in \inn(C)$. Therefore, $C'$ is the desired cycle.
		
	\item Case 2: $w_1,w_2 \not\in V(P_w)$.
		\\ By Substructure Lemma \ref{lemma_subchain2}, $C^{-w}$ is a $w_1w_2$-chain. Write $\ora{w_1C^{-w}w_2}=(w_1,e_1,\ldots,e_L,w_2)$. We have $s \in \out(C)$, moreover $s \neq w$ by Claim \ref{claim_Type1}\ref{item_c}, so there exists a unique index $1 \leq i \leq L$ such that $s \in e_i$. If $i \neq 1$, define $s_1$ as the only vertex in $e_{i-1} \cap e_i$ and $P_1 = \HH{(w_1,e_1,\ldots,e_{i-1},s_1)}$, otherwise define $s_1=w_1$ and $P_1= \HH{(w_1)}$. Similarly, if $i \neq L$, define $s_2$ as the only vertex in $e_i \cap e_{i+1}$ and $P_2 = \HH{(w_2,e_L,e_{L-1},\ldots,e_{i+1},s_2)}$, otherwise define $s_2=w_2$ and $P_2 = \HH{(w_2)}$. For all $j \in \{1,2\}$, $P_j$ is a $w_js_j$-chain in $C$, and $V(P_1) \cap V(P_2)=\varnothing$. These notations are summed up in Figure \ref{Type1-2}.
		
		\begin{figure}[h]
			\centering
			\includegraphics[scale=.58]{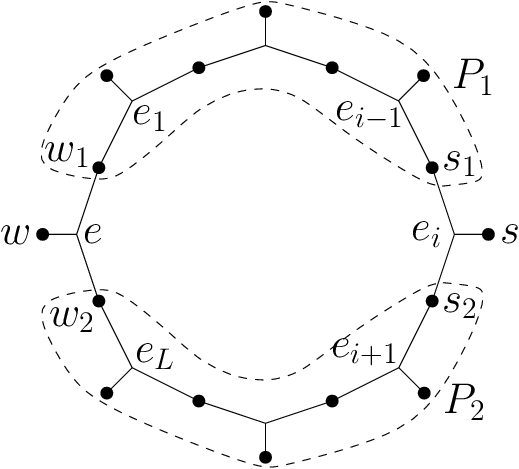}
			\caption{The cycle $C$.}\label{Type1-2}
		\end{figure}
		
		By Claim \ref{claim_Type1}\ref{item_a}, we have $|V(P_w) \cap (V(C) \setminus \{w\})| \in \{1,2\}$. Note that $V(C) \setminus \{w\} = V(P_1) \cup V(P_2) \cup \{s\}$. Since $s \not\in V(P_w)$ by Claim \ref{claim_Type1}\ref{item_c}, we obtain that $|V(P_w) \cap V(P_1)| \in \{1,2\}$ or $|V(P_w) \cap V(P_2)| \in \{1,2\}$. By symmetry, assume that $|V(P_w) \cap V(P_1)| \in \{1,2\}$.
		
		\begin{itemize}[noitemsep,nolistsep]
			\item First suppose $|V(P_w) \cap V(P_1)|=1$. Let $u$ be the only vertex in $V(P_w) \cap V(P_1)$: in particular, $P_w$ is a $wu$-chain. Recall that $w_2 \not\in V(P_w)$ by assumption, moreover $w_2 \not\in V(P_1)$ by definition of $P_1$. Therefore, the walk $(w,e,w_1) \oplus \ora{w_1P_1}\vert_{\{u\}} \oplus \ora{uP_ww}$ represents a $w$-cycle in $C \cup P_w$, which does not contain $s$ since $s \not\in e \cup V(P_1) \cup V(P_w)$. This is the desired cycle.
			\item Now suppose $|V(P_w) \cap V(P_1)|=2$. Since $V(P_1) \cap V(P_2)=\varnothing$, this yields $V(P_w) \cap V(P_2)=\varnothing$ by Claim \ref{claim_Type1}\ref{item_a}. In particular $s_2 \not\in V(P_w)$, so there cannot be an $s_1$-tadpole $T$ in $P_w \cup P_1$: indeed, since $s,s_2 \not\in V(P_w) \cup V(P_1)$, the walk $(s,e_i,s_1) \oplus \ora{s_1T}$ would otherwise represent an $s$-tadpole in $D$, contradicting Proposition \ref{prop_s_bis}. Therefore, Union Lemma \ref{Lemma1} with $a=w_1$, $b=s_1$ and $c=w$ ensures that $P_1 \cup P_w$ contains a $ww_1$-chain $P_{ww_1}$. Since $w_2 \not\in V(P_1)$, and $w_2 \not\in V(P_w)$ by assumption, the walk $\ora{wP_{ww_1}w_1} \oplus (w_1,e,w)$ represents a $w$-cycle, which does not contain $s$ since $s \not\in e \cup V(P_{ww_1})$. This is the desired cycle.
		\end{itemize}
	
\end{enumerate}

\noindent In conclusion, Proposition \ref{prop_roadmap} applies and yields $\I{H^{+w}}{w\D_1(H)}=\varnothing$, contradicting property $J(\D_1,H)$. This ends the proof of Lemma \ref{lemma_main_structure}, so that all results of this paper are now proved.

\section{Conclusion}\label{Section6}

\hspace{\parindent}In this paper, we have obtained a structural characterization of the outcome of the Maker-Breaker game played on hypergraphs of rank 3, and a description of both players' optimal strategies, all based on danger intersections. We have shown that Maker wins if and only if she can force the appearance of a nunchaku or a necklace within at most three rounds of play. As a consequence, the winner of the Maker-Breaker game on a hypergraph of rank 3 can be decided in polynomial time. This result, together with a recent preprint which shows \PSPACE-completeness for 4-uniform hypergraphs \cite{Gal25}, closes the complexity gap for the \makerbreaker problem depending on the size of the edges.
\\ \indent Moreover, some games played on graphs that have been studied in the literature can be seen as Maker-Breaker games on some underlying hypergraph of rank 3 and, as such, are now solved as well. One such example is the \textit{total domination game} played on the vertex set of a cubic graph $G$, where Maker's winning sets are the open neigborhoods of the vertices of $G$, for which only a partial solution was known before \cite{FM22}. Another example is the \textit{$H$-game} played on the vertex set \cite{KMN19} (resp. the edge set \cite{DGM25}) of a graph $G$, where Maker's winning sets are the copies in $G$ of some fixed graph $H$ of order 3 (resp. of size 3).
\\ \indent We have also obtained a logarithmic upper bound for $\tau_M$ (the duration of the game) on 3-uniform Maker wins, which is tight up to an additive three rounds as shown by the case of nunchakus. On the other hand, for any fixed $k \geq 4$, it is known that a $k$-uniform hypergraph $H_n$ on $n$ vertices such that $\tau_M(H_n) = \lceil\frac{n}{2}\rceil$ exists for all $n$ large enough \cite[Proposition 16]{BDG26}. The question of finding the best general bounds for $\tau_M$ depending on the size of the edges is thus resolved.
\bigskip
\\\indent In contrast, all these questions remain open for the more general version of the game (\UQBF) which is played on a 3-CNF formula instead of a hypergraph of rank 3. The concept of danger should translate to this version, except that there would be two types of dangers at $x$ depending on which truth value Falsifier assigns to $x$. For Satisfier to win, it would then be necessary that property $J()$ holds for both types of dangers. We do not know if our proofs could partly generalize to this version, thus helping towards proving Conjecture \ref{conjecture}. For now, we have only validated Conjecture \ref{conjecture} for positive 3-CNF formulas.
\bigskip
\\\indent The algorithmic complexity of positional games on hypergraphs of rank 3 is still open for some other conventions.
\\ \indent One such example is the Maker-Maker convention. A famous \textit{strategy-stealing} argument \cite{HJ63} ensures that optimal play can only lead to a first-player win or a draw. A first-player win implies a Maker win in the Maker-Breaker convention, but the converse is false, as evidenced by tic-tac-toe. The problem of deciding the outcome is trivially tractable for hypergraphs of rank 2, but it is \PSPACE-complete for 4-uniform hypergraphs \cite{Gal25, GS26b}. When it comes to hypergraphs of rank 3, it seems difficult to obtain an analogue of Theorem \ref{theo_main_structure1}, because there is no subhypergraph monotonicity. The fact that the second player can make threats of his own may cause untractability. Actually, in a recent preprint, \PSPACE-completeness for the Maker-Maker game on 3-uniform hypergraphs has been established for positions obtained after one round of (non-optimal) play \cite{GS26a}.
\\ \indent Complexity results have recently been obtained on variations of the Maker-Breaker convention called \textit{Waiter-Client} and \textit{Client-Waiter} \cite{GOT25}. Waiter, who has the role of Maker in Waiter-Client or of Breaker in Client-Waiter, offers a choice of two vertices in each round, from which Client picks one for himself and gives the other to Waiter. The Client-Waiter convention on hypergraphs of rank 3 presents interesting similarities with the Maker-Breaker convention: the key is the existence of a pair of vertices such that each one destroys all the $\D_1$-dangers (snakes and tadpoles) at the other. If such a pair exists, then it is optimal (though not necessarily winning) for Waiter to offer it, but if none exist, then Client wins. This shows membership in the complexity class {\sf P}${}^{\text{\sf NP}}$, and it would show membership in \textsf{P} if the tadpole existence problem could be proved to be solvable in polynomial time. As for the Waiter-Client convention, it is tractable for any fixed rank.
\\ \indent Another convention is called \textit{Avoider-Enforcer}: one player wants to avoid picking all the vertices of any edge, while the other tries to force them to. Same as Maker-Breaker, those are \textit{weak games}, where the players have complementary objectives and convenient properties such as subhypergraph monotonicity hold. As such, the Avoider-Enforcer convention is likely a better candidate for tractability on hypergraphs of rank 3 than the Maker-Maker convention for instance. Structural characterizations of the outcome and polynomial-time algorithms are already known for hypergraphs of rank 2 and for a subcase of linear hypergraphs of rank 3 \cite{GGP26}, with the latter underlining the importance of nunchakus in this convention as well. On the other hand, \PSPACE-completeness has been established for 6-uniform hypergraphs \cite{GO23}.
\bigskip
\\\indent Positional games also exist in a \textit{biased} version $(p:q)$, where the players get to pick $p$ and $q$ vertices respectively each round instead of only one vertex. Many instances of biased Maker-Breaker games have been studied in the literature, notably to study the \textit{threshold bias} which is the smallest $q$ such that Breaker wins with a $(1:q)$ bias \cite{HKS14}. The concept of danger and the basic results around it naturally generalize to the biased version of the game \cite[I.2.2]{Gal23}, and might be useful in some cases. For example, one could look into biased games on hypergraphs of rank 3 with a danger-based approach, defining elementary dangers depending on the bias, and trying to get a characterization of Breaker wins similar to the one we obtained without bias.

\section*{Acknowledgments}

\noindent We thank Md Lutfar Rahman and Thomas Watson for useful exchanges, as well as the anonymous reviewers for their insightful feedback. This work was partly supported by the ANR P-GASE project (grant ANR-21-CE48-0001).

%% file: appendix.tex
\newpage

\appendix

\section{Glossary of technical terms and symbols}\label{appendix}

\begin{tabular}{|c|c|c|c|}
	\hline
	Symbol & Description & Reference & Figure
	\\ \hline\hline
	\rowcolor{mygray1} \multicolumn{4}{|c|}{\bf Hypergraphs}
	\\ \hline
	& marked hypergraph & \Def\ref{def:markedhypergraph} & 
	\\ \hline
	$\subseteq$ & subhypergraph of a marked hypergraph & \Def\ref{def:subhypergraph} & 
	\\ \hline
	$\union{\X}$, $\cup$ & union of a collection $\X$ of subhypergraphs & \Def\ref{def:union} & 
	\\ \hline
	$H^{+x}$ & marked hypergraph obtained from $H$ by marking $x$ & \Not\ref{not:updated} &
	\\ \hline
	$H^{-y}$ & marked hypergraph obtained from $H$ by deleting $y$ & \Not\ref{not:updated} &
	\\ \hline
	& linear hypergraph & \Def\ref{def:linearhypergraph} &
	\\ \hline
	& hyperforest & \Def\ref{def:hyperforest} &
	\\ \hline\hline
	\rowcolor{mygray1} \multicolumn{4}{|c|}{\bf Maker-Breaker game}
	\\ \hline
	& trivial Maker win & \Def\ref{def:trivialmakerwin} & 
	\\ \hline
	& Maker win & \Def\ref{def:makerwin} &
	\\ \hline 
	& \makerbreaker~decision problem & \Not\ref{not:makerbreaker} &
	\\ \hline
	$\tau_M(H)$ & minimum number of rounds for Maker to win on $H$ & \Not\ref{not:tau} &
	\\ \hline\hline
	\rowcolor{mygray1} \multicolumn{4}{|c|}{\bf Walks}
	\\ \hline
	$\ora{W}$ & walk & \Def\ref{def:walk} &
	\\ \hline
	& equivalent walks & \Def\ref{def:equivalent} &
	\\ \hline
	$\ola{W}$ & reverse walk & \Not\ref{notation_walk} &
	\\ \hline
	$\oplus$ & concatenation of walks & \Not\ref{notation_walk} &
	\\ \hline
	$\ora{W}\vert_Z$ & walk ``cut at $Z$'' & \Not\ref{notation_walk} &
	\\ \hline
	$\HH{\ora{W}}$ & marked hypergraph induced by a walk & \Def\ref{def:induced} &
	\\ \hline
	& linear walk & \Def\ref{def:linearwalk} &
	\\ \hline
	& simple walk & \Def\ref{def:simplewalk} &
	\\ \hline\hline
	\rowcolor{mygray1} \multicolumn{4}{|c|}{\bf Chains}
	\\ \hline
	$P$ & chain, $ab$-chain, $a$-chain & \Def\ref{def_chain} & \Fig\ref{Example_Paths}
	\\ \hline
	& walk representing a chain & \Def\ref{def_chain} &
	\\ \hline
	$\ora{aPb}$, $\ora{aP}$ & usual walks for representing an $ab$-chain $P$ & \Not\ref{not:walk_chain} &
	\\ \hline
	$\inn(P)$ & inner vertices of a chain $P$ & \Def\ref{def:chain_inner} & \Fig\ref{Example_Paths}
	\\ \hline
	$o(a,\ora{aPb})$ & neighbor of $a$ of degree 1 in an $ab$-chain $P$ & \Not\ref{not:o_neighbor} & \Fig\ref{Example_Paths}
	\\ \hline
	& snake, $ab$-snake, $a$-snake & \Def\ref{def:snake} & \Fig\ref{Example_Marked}
	\\ \hline
	& nunchaku, $ab$-nunchaku, $a$-nunchaku & \Def\ref{def:nunchaku} & \Fig\ref{Example_Marked}
	\\ \hline
	$\proj{Z}{u}{P}$ & projection of $u$ onto $Z$ in a chain $P$ & \Def\ref{def:projection} & \Fig\ref{Projections}
	\\ \hline
	$\perp$ & non-linear intersection between an edge and a chain & \Not\ref{not:perp} & \Fig\ref{Perpendicular}
	\\ \hline
	$LCC_H(a)$ & linear connected component of $a$ in $H$ & \Def\ref{def:LCC} &
	\\ \hline
	$\dist_H(a,b)$ & length of a shortest $ab$-chain in $H$ & \Not\ref{not:dist} &
	\\ \hline\hline
	\rowcolor{mygray1} \multicolumn{4}{|c|}{\bf Cycles and tadpoles}
	\\ \hline
	$C$ & cycle, $a$-cycle & \Def\ref{def_cycle} & \Fig\ref{Example_Cycles}
	\\ \hline
	$\ora{aC}$, $\ora{(a-e)C}$ & usual walks for representing an $a$-cycle $C$ & \Not\ref{not:walk_cycle} &
	\\ \hline
	$\inn(C), \out(C)$ & inner/outer vertices of a cycle $C$ & \Def\ref{def:cycle_inner} & \Fig\ref{Example_Cycles}
	\\ \hline
	& necklace, $a$-necklace & \Def\ref{def:necklace} & \Fig\ref{Example_Marked}
	\\ \hline
	$T$ & tadpole, $a$-tadpole & \Def\ref{def_tadpole} & \Fig\ref{Example_Tadpoles}
	\\ \hline
	$P_T, C_T$ & chain part and cycle part of a tadpole $T$ & \Def\ref{def_tadpole} & \Fig\ref{Example_Tadpoles}
	\\ \hline
	$\ora{aT}$ & usual walk for representing an $a$-tadpole $T$ & \Not\ref{not:walk_tadpole} &
	\\ \hline
	$\proj{Z}{u}{T}$ & projection of $u$ onto $Z$ in a tadpole $T$ & \Def\ref{def:projection} & \Fig\ref{Projections}
	\\ \hline
\end{tabular}

\newpage

\noindent
\begin{tabular}{|c|c|c|c|}
	\hline
	Symbol & Description & Reference & Figure
	\\ \hline\hline
	\rowcolor{mygray1} \multicolumn{4}{|c|}{\bf General dangers}
	\\ \hline
	& danger at a given vertex in a marked hypergraph & \Def\ref{def:dangeratx} &
	\\ \hline
	& pointed marked hypergraph & \Def\ref{def:pointed} &
	\\ \hline
	& isomorphism of pointed marked hypergraphs & \Def\ref{def:isomorphic} &
	\\ \hline
	& danger & \Def\ref{def:danger} &
	\\ \hline
	& trivial danger & \Def\ref{def:trivialdanger} & \Fig\ref{Trivial_Danger}
	\\ \hline
	$\I{H}{\X}$ & intersection of a collection of marked hypergraphs $\X$ in $H$ & \Def\ref{def:intersection} &
	\\ \hline
	& fork at a given vertex in a marked hypergraph & \Def\ref{def:fork} &
	\\ \hline
	& $\D$-danger & \Def\ref{def:D-danger} &
	\\ \hline
	& $\D$-danger at a given vertex in a marked hypergraph & \Def\ref{def:D-danger} &
	\\ \hline
	$x\D(H)$ & collection of all $\D$-dangers at $x$ in $H$ & \Def\ref{def:D-danger} &
	\\ \hline
	& $\D$-fork at a given vertex in a marked hypergraph & \Def\ref{def:D-danger} &
	\\ \hline
	$J(\D,H)$ & intersecting property for the $\D$-dangers in $H$ & \Not\ref{not:J} &
	\\ \hline\hline
	\rowcolor{mygray1} \multicolumn{4}{|c|}{\bf Families of 3-uniform dangers}
	\\ \hline
	$\D_0$ & dangers that are snakes & \Not\ref{not:D0_D1} & \Fig\ref{D1-dangers}
	\\ \hline
	$\D_1$ & dangers that are snakes or tadpoles & \Not\ref{not:D0_D1} & \Fig\ref{D1-dangers}
	\\ \hline
	$\widehat{\D_1}$ & dangers that are potential $\D_1$-forks & \Not\ref{not:D1-hat} & \Fig\ref{Example_Dangers}
	\\ \hline
	$\D_2$ & union of the families $\D_1$ and $\widehat{\D_1}$ & \Not\ref{not:D1-hat} &
	\\ \hline
	& decomposition of a $\widehat{\D_1}$-danger & \Not\ref{not:D1-hat} & \Fig\ref{Example_Dangers}
	\\ \hline
	& maximal decomposition of a $\widehat{\D_1}$-danger & \Def\ref{def:maximal_decomposition} &
	\\ \hline
	
\end{tabular}